\documentclass[12pt]{article}

\usepackage{
    amsmath,
    amssymb,
    amsthm,
    authblk,
    bm,
    booktabs,
    enumitem,
    graphicx,
    mathtools,
    multirow,
    natbib,
    nicefrac,
    physics,
    subcaption,
    todonotes,
    color,
    xr-hyper,
    float
}
\usepackage[pdffitwindow=false,
            plainpages=false,
            pdfpagelabels=true,
            pdfpagemode=UseOutlines,
            pdfpagelayout=SinglePage,
            bookmarks=false,
            colorlinks=true,
            hyperfootnotes=false,
            linkcolor=blue,
            urlcolor=blue!30!black,
            citecolor=green!50!black]{hyperref}
\usepackage{cleveref}
\definecolor{darkgreen}{rgb}{0.0, 0.6, 0.}

\newtheorem{observation}{Observation}

\newtheorem{theorem}{Theorem}
\newtheorem{proposition}{Proposition}

\newtheorem{lemma}{Lemma}

\newtheorem{definition}{Definition}

\newcommand{\revisionone}[1]{#1}
\definecolor{DeepPurple}{RGB}{61, 0, 249}
\newcommand{\revisiontwo}[1]{{#1}}
\newcommand{\qt}[1]{\left(#1\right)}\newcommand{\set}[1]{\left\{#1\right\}}
\newcommand{\ex}[1]{\exp\qty(#1)}
\newcommand{\wh}[1]{\widehat{#1}}\newcommand{\tl}[1]{\widetilde{#1}}
\newcommand{\E}{\mathbb{E}}\newcommand{\R}{{\mathbb R}}\newcommand{\T}{{\mathbb T}}\newcommand{\Z}{{\mathbb Z}}
\newcommand{\Dd}{\mathcal{D}}\newcommand{\Hh}{\mathcal{H}}\newcommand{\Kk}{\mathcal{K}}\newcommand{\Ll}{\mathcal{L}}\newcommand{\Nn}{\mathcal{N}}\newcommand{\Tt}{{\mathcal T}}\newcommand{\Ww}{{\mathcal W}}\newcommand{\Uu}{{\mathcal U}}
\newcommand{\rmm}[1]{\mathrm{#1}}
\newcommand{\qb}[1]{\left[#1\right]}
\newcommand{\x}{{\bm{x}}}
\newcommand{\y}{{\bm{y}}}
\renewcommand{\d}{\,\mathrm{d}}\newcommand{\dx}{\,\mathrm{d} x}
\newcommand\restr[2]{{\left.\kern-\nulldelimiterspace #1\vphantom{\big|} \right|_{#2}}}
\renewcommand{\b}[1]{{\bm{#1}}}

\newcommand{\xxi}{{\bm{\xi}}}
\newcommand{\br}[1]{\left\langle#1\right\rangle}
\renewcommand{\b}[1]{{\bm{#1}}}
\newcommand{\om}{\bm{\omega}}
\newcommand{\blind}{1}
\newenvironment{aligneq}{\begin{equation}\begin{aligned}}{\end{aligned}\end{equation}}
\begin{document}

\bibliographystyle{apalike}

\def\spacingset#1{\renewcommand{\baselinestretch}%
    {#1}\small\normalsize} \spacingset{1}


\if1\blind
    {
        \title{\bf A parameterization of anisotropic Gaussian fields with penalized complexity priors}
        \author[1]{L. Llamazares-Elias\thanks{Corresponding author: L.S.Llamazares-Elias@sms.ed.ac.uk}}
        \affil[1]{School of Mathematics \& Maxwell Institute for Mathematical Sciences, University of Edinburgh, United Kingdom}
        \author[2]{J. Latz}
        \affil[2]{Department of Mathematics, University of Manchester, United Kingdom}
        \author[1]{F. Lindgren}
        \maketitle
    } \fi

\if0\blind
    {
        \bigskip
        \bigskip
        \bigskip
        \begin{center}
            {\LARGE\bf Anisotropic Gaussian random fields with identifiable parameters and penalized-complexity priors}
        \end{center}
        \medskip
    } \fi

\bigskip

\begin{abstract}
    Gaussian random fields (GFs) are fundamental tools in spatial modeling and can be represented flexibly and efficiently as solutions to stochastic partial differential equations (SPDEs). The SPDEs depend on specific parameters, which enforce various field behaviors and can be estimated using Bayesian inference. However, even under in-fill asymptotics, the likelihood only provides limited insights into the covariance structure. In response, it is essential to leverage priors to achieve appropriate, meaningful covariance structures in the posterior. This study introduces a smooth, invertible parameterization of the correlation length and diffusion matrix of an anisotropic GF and constructs penalized complexity (PC) priors for the model when the parameters are constant in space. The formulated prior is weakly informative, effectively penalizing complexity by pushing the correlation range toward infinity and the anisotropy to zero.
\end{abstract}

\noindent%
{\it Keywords:} Anisotropy, Bayesian, Penalized Complexity, Prior, Spatial Modelling, Stochastic Partial Differential Equations.
\vfill

\newpage
\spacingset{1.9} 
\section{Introduction}

Gaussian random fields (GFs) are widely used to model spatial phenomena \citep{banerjee2003hierarchical,bhatt2015effect,wang2021spatial,rue2005gaussian} while accounting for the uncertainty that may arise due to measurement error, model misspecification, or incomplete information. The prevalence of GFs is due to the fact that they are well understood theoretically, verify desirable properties, and are easily characterized -- they are entirely specified by their mean and covariance \citep{cramer1967stationary,adler2007random,rasmussen2003gaussian}.

A convenient way of representing certain GFs is as solutions to stochastic partial differential equations (SPDEs). This representation allows for a physical interpretation to be assigned to the parameters of the equation. Furthermore, it allows for computationally efficient inference, prediction, and uncertainty quantification using a finite element (FEM) approximation of the field \citep{Lindgren2011AnEL,simpson2012think}.

In the literature, a common choice is to model using isotropic fields. That is, the correlation of the field at two locations depends only on the distance between said locations. While this may be an appropriate assumption in some cases, in others, it is inadvisable. \revisionone{This limitation can be overcome by introducing additional parameters to model the anisotropy present in the field \citep{higdon1999non, paciorek2006spatial,fuglstad2015exploring}}. In the following, we consider the semi-parametric estimation of the random field and its anisotropy parameters. The existing work leaves us with two significant challenges:

\begin{enumerate}
    \item \revisionone{The anisotropy parameterizations in \cite{higdon1999non}, \cite{paciorek2006spatial}, and \cite{fuglstad2015exploring}} are non-identifiable as they have multiple parameter combinations for each anisotropy matrix, leading to a multi-modal likelihood, making them unsuitable as general parameterizations. \revisionone{Furthermore, as a consequence of the lack of identifiability, under the aforementioned parameterizations, it is not possible to recover the parameters in a continuous fashion. This is especially problematic when the model parameters are extended to be spatially varying.} 
    \item Given that not all parameters can be recovered under in-fill asymptotics \citep{zhang2004inconsistent},  the \emph{choice of prior distribution} on the parameters of the model may signific\-antly impact the posterior distribution. As a result, suitable priors need to be defined.
\end{enumerate}
To address these issues, we make the following contributions:
\begin{enumerate}
    \item \textbf{An Identifiable Parameterization}: We present a  smooth and invertible parameterization of the anisotropy that preserves parameter interpretability.\label{contribution 1}
    \item \textbf{Prior definition}: We construct penalized complexity (PC) priors \citep{Simpson2014PenalisingMC,fuglstad2019constructing}  for the parameters in the model. An additional benefit of this construction is that it avoids overfitting by favoring simpler base models.\label{contribution 2}
    \item \textbf{Validation and prediction}: We conduct a simulation study that shows that PC priors outperform ``non-informative'' priors. We then use the derived model to study precipitation in Norway and show that the anisotropic model outperforms the isotropic one in the presence of limited information.
\end{enumerate}

The outline of the work is as follows. In \Cref{model section}, we introduce and motivate our anisotropic model. In \Cref{parameterization section}, we address Contribution \ref{contribution 1} and also construct a transformation that renders the parameters a Gaussian vector, providing the convenience of working with Gaussian random variables. In \Cref{priors section}, we present Contribution \ref{contribution 2}. Next, in \Cref{simulation section}, we conduct a simulation study of the designed PC priors and compare the results with those of other possible priors on the parameters. In \Cref{precipitation section}, we study the performance of the model and priors on a real data set. Finally, in \Cref{Discussion section}, we synthesize the obtained results and discuss future avenues of research.

\section{\revisionone{Modeling anisotropy with SPDEs} }\label{model section}
\revisionone{
One of the fundamental goals of spatial statistics is inferring some quantity $u(s)$ defined over a spatially continuous domain $\Dd$. Since $u$ is a function, this amounts to describing a probability distribution over the space of functions to which we allow the realizations of $u$ to belong. A common approach is to assume that $u$ is a GF, that is, a random function $u:\Dd\to\R$ such that for any finite set of points $\bm{x}_1,\ldots,\bm{x}_n\in\Dd$, the vector $\b{u}=(u(\bm{x}_1),\ldots,u(\bm{x}_n))$ is a multivariate Gaussian random variable $\b{u} \sim \mathcal{N}(\b{0},\bm{K})$, where $\bm{K}$ is the covariance matrix of $\b{u}$. One of the fundamental challenges of such models is computational. The size of the covariance matrix scales with the number of observations, leading to a computational cost of $\mathcal{O}(n^3)$ for inference. This is prohibitive for large $n$.
}

\revisionone{
To overcome this, one can assume that the discretized field $\b{u}$ is a Gaussian Markov random field (GMRF). That is, that the precision matrix $\b{Q}= \b{K} ^{-1}$ is sparse. In two dimensions, this allows for inference with a cost of $\mathcal{O}(n^{3/2})$ \citep{rue2005gaussian}. This provides a significant speed-up as compared to the $\mathcal{O}(n^3)$ cost of working with a dense covariance matrix. The question is, then, how can one specify a continuous object, the GF $u$, such that the discretized field $\b{u}$ is a GMRF? As shown by \citet{rozanov1977markov}, a stationary GF $u$ verifies the Markov property if and only if the inverse of its spectrum is a positive symmetric polynomial. Equivalently, if and only if $u$ solves a SPDE of the form
\begin{align}\label{SPDE approach}
  \Ll^{\nicefrac{1}{2}} u = \dot{\Ww},
\end{align}
where $\Ww$ is Gaussian white noise and $\Ll $ is a partial differential operator whose symbol $\wh{\Ll }(\b{\xxi })$ is a positive, symmetric polynomial in $\xxi $.} \revisiontwo{
That is,}
\begin{align*}
  \int_{\R^d} (\Ll f)(\bm{x}) \overline{g(\bm{x})}= \int_{\R^d} \wh{\Ll }(\b{\xxi })\wh{f}(\b{\xxi })\overline{\wh{g}(\b{\xxi })}\d \b{\xxi }, \quad \forall f,g \in C^\infty_c(\R^d), 
\end{align*}
\revisiontwo{ where $C_c^\infty(\R^d)$ is the space of smooth functions with compact support in $\R^d$.}
\revisionone{
The above fact motivates the SPDE approach \citep{Lindgren2011AnEL,lindgren2022spde}. The continuous GF $u$ is modeled as the solution to an SPDE of the form \eqref{SPDE approach}. In this way, we can do the modeling using continuous GFs that are \emph{independent of any spatial discretization} and \emph{enforcing physical properties} such as reaction, advection, diffusion, smoothing, etc., through the operator $\Ll $ and at the same time, reaping the \emph{computational benefits} of working with GMRFs.}
\subsection{Model formulation}
In this section, we address the first main contribution. How can we introduce anisotropy into our model and parameterize it such that each parameter has an interpretable effect? We work in $2$ dimensions and within the framework of the SPDE approach where, most commonly, $u$ is chosen to be the solution of 
\begin{equation}\label{SPDE2011}
 (\kappa^2-\Delta)^\frac{\alpha}{2}u=\dot{\Ww}.
\end{equation}
Here, $\dot{\Ww}$ is Gaussian white noise on $L^2(\R^2)$, and is defined such that given $f,g \in L^2(\R^2)$ the random measure $ \revisiontwo{\dot{\Ww}}(f)$ verifies
\begin{align*}
  \E[\revisiontwo{\dot{\Ww}}(f)\revisiontwo{\dot{\Ww}}(g)]=\br{f,g}_{L^2(\R^2)}.
\end{align*}
The resulting field $u$ has a mean of zero and is chosen to be isotropic, ensuring uniqueness \cite[Section 2.3]{lindgren2022spde}. We recall that, by definition, the field $u$ is isotropic if there exists some function $r$ such that, for all $\bm{x,y}\in\R^2$,
\begin{align}\label{stationary}
 K(\bm{x},\bm{y}):={\mathrm{Cov}[u(\bm{x}),u(\bm{y})]}=r(\norm{\bm{y}-\bm{x}}).
\end{align}
The isotropy of $u$ may be an appropriate assumption in some cases. However, it is inadvisable if the correlation of the field is not equal in all spatial directions. To model this anisotropy, we will consider the model given by
\begin{equation}\label{SPDE2D}
 (\kappa^{2}-\nabla\cdot \bm{H}\nabla)\frac{u}{\sigma_u }=\sqrt{4 \pi } \kappa\dot{\Ww},
\end{equation}
where the parameters are $\kappa, \sigma_u \in (0,\infty)$, which are positive and bounded away from zero, and a symmetric positive definite matrix $ \bm{H} \in \R^{2\times 2}$ with determinant $1$ and  eigenvalues bounded away from zero. The parameters control the length scale, marginal variance, and anisotropy, respectively. \revisionone{\Cref{SPDE2D} is derived by the geometric deformation of a stationary isotropic field. This provides a clear physical interpretation of the anisotropy matrix $\b{H}$ and the parameter $\kappa $ as detailed in Section \ref{app:model derivation}}.

The formulation in \eqref{SPDE2D} preserves the advantages of the SPDE approach. Namely, representing $u$ as the solution of an SPDE gives it a physical interpretation. The term $\kappa^2 u$ represents reaction whereas $-\nabla \revisiontwo{\cdot}\bm{H}\nabla u$ represents diffusion \citep{evans2010partial,mech}.

Furthermore, using a finite element method, $u$ can be projected onto the finite-dimensional Hilbert space $\Hh_n$ spanned by the basis functions $\set{\psi_1, \ldots, \psi_n}$ linked to a mesh $M_n$ of the domain. This gives a sequence of Gaussian Markov random fields $u_n$ with sparse precision matrices, which converges in distribution to $u$ as the mesh becomes finer and finer. This sparsity allows for a significant speed-up in computations (Kriging, posterior simulation, likelihood evaluation, etc.) \citep{Lindgren2011AnEL, simpson2012think}.

\section{Parameterization}\label{parameterization section}
In this section, we parameterize $\bm{H}$ so that the parameters convey intrinsic geometric meaning about the field $u$. We recall that $\b{H}$ is symmetric positive definite and has determinant $1$.   Suppose for example that $\kappa$ is fixed to $1$ so that $\bm{\Psi}= \sqrt{\bm{H}}$. Write $$\{(\bm{v},\lambda^2),(\bm{v}_\perp,\lambda^{-2})\}$$ for the eigensystem of $\bm{H}$, where $\bm{v}:=(v_1,v_2)\in\mathbb{R}^2$ and $\bm{v}_\perp=(-v_2,v_1)$ and we can suppose $\lambda\geq 1$ by reordering if $\lambda< 1$. Then, as shown in \Cref{app:model derivation},  $u$ corresponds to deforming and rescaling the stationary field $\tl{u}$ through
\begin{equation*}
    u(\bm{x})= \sigma_u \cdot   \tl{u}(\bm{\Psi}\bm{x}),\quad \bm{\Psi}\bm{x}= \lambda\br{ \bm{x},\bm{v}}\bm{v}+\lambda^{-1} \br{\bm{x},\bm{v}_\perp}\bm{v}_\perp.
\end{equation*}
The rescaling by $\sigma _u$  corresponds to changing the variance of the field. The deformation corresponds to stretching the initial domain by a factor of $\lambda$ in the direction of  $\bm{v}$ and contracting, also by a factor of $\lambda$, in the orthogonal direction $\bm{v}_\perp$. The above shows that the eigensystem of $\bm{H}$ carries fundamental geometric information, motivating a parameterization of $\bm{H}$ in terms of its eigensystem.

Equation \eqref{SPDE2D} was also considered by \cite{fuglstad2015exploring}. Here, the authors defined $\bm{v}(\alpha):=(\cos(\alpha),\sin(\alpha))$ and parameterized $\bm{H}$ as
\begin{align}\label{fuglstadtparam}
    \bm{H}_{\bm{v}(\alpha)}=\gamma \bm{I}+\beta \bm{v}(\alpha)\bm{v}(\alpha)^T.
\end{align}
However, the map $\alpha \mapsto \bm{H}_{\bm{v}(\alpha)}$ is not injective as $\bm{H}_{\bm{v}(\alpha)}=\bm{H}_{-\bm{v}(\alpha)}$. Because of this, it is impossible to recover the sign of $\bm{v}$, and thus, this parameterization is not \revisiontwo{injective, typically} leading to \revisiontwo{an at least} bimodal likelihood. \revisionone{Similarly, in the case of \cite{higdon1999non}, with their notation, using $(\phi_x,\phi_y)$ gives the same covariance as $(-\phi _x,-\phi _y)$. Likewise, in the case and with the notation of \cite{paciorek2006spatial}, using $(\gamma_1,\gamma _2)$ gives the same covariance as using $(-\gamma_1,-\gamma_2)$.}

The crucial step to obtain an \revisiontwo{invertible} parameterization is to consider the ``half-angle'' version $\tl{\bm{v}}$ of $\bm{v}$ as an eigenvector of $\bm{H}$.
\begin{theorem}[parameterization]\label{parameterizATION THEOREM}
    Given $\bm{v}=(v_1,v_2) \in \R^2$ define
    \begin{align}\label{half}
        \tl{\bm{v}}:= \norm{\bm{v}} \exp({i \alpha /2 }), \quad  \text{ where }  \alpha := \arg(\bm{v}) \in [0,2 \pi).
    \end{align}
    Write $\tl{\bm{v}}=(\tl{v}_1,\tl{v}_2)$ and $\tl{\bm{v}}_\perp=(-\tl{v}_2,\tl{v}_1)$.
    Then, the following defines a smooth, invertible parameterization on the space of symmetric positive definite matrices of determinant $1$.
    \begin{align}\label{cool formula}
        \bm{H}_{\bm{v}} & =\frac{\exp({\norm{\bm{v}} })}{\norm{\bm{v}}^2 } \tl{\bm{v}}\tl{\bm{v}}^T+\frac{\exp({-\norm{\bm{v}} })}{\norm{\bm{v}}^2 } \tl{\bm{v}}_\perp \tl{\bm{v}}_\perp^T
        =\cosh\left( \lvert \bm{v} \rvert \right) \bm{I} + \frac{\sinh\left( \lvert \bm{v} \rvert \right)}{\lvert \bm{v} \rvert}
        \begin{bmatrix}
            v_1 & v_2  \\
            v_2 & -v_1
        \end{bmatrix}.
    \end{align}
\end{theorem}
The link with the parameterization by \citet{fuglstad2015exploring} is the following
\begin{proposition}\label{relationship prop}
    Let $\bm{H}_{\bm{v}(\alpha)},\bm{H}_{\bm{v}}$ be as in \eqref{fuglstadtparam}, \eqref{cool formula}, then if we set
    \begin{align*}\label{relationship}
        \bm{v}(\alpha)=\pm{\tl{\bm{v}}}/{\norm{\bm{v}}}, \quad \gamma = \exp({-\norm{\bm{v}} }),\quad \beta =(1-\gamma^2)/{\gamma}.
    \end{align*}
    We obtain $\bm{H}_{\bm{v}(\alpha)}=\bm{H}_{\bm{v}}$.
\end{proposition}

The idea behind the usage of the half angle version $\tl{\bm{v}}$ of $\bm{v}$ is that it avoids any issues of identifiability in the sign as $\tl{\bm{v}}$ and $-\tl{\bm{v}}$ do not simultaneously belong to the parameter space. Parameterization using a Cholesky decomposition of $\bm{H}$ is also possible and is more readily generalized to higher dimensions. However, it is not as easily interpretable in terms of the intrinsic properties of $u$.

In \Cref{half-angle figure}, we show the half-angle vector field together with the parameterized diffusion matrices $\bm{H}_{\bm{v}}$. Here, each $\bm{H}_{\bm{v}}$ is represented by the ellipse centred at $\bm{v}$  whose main axis is $\exp(\norm{\bm{v}})\tl{\bm{v}}/\norm{\bm{v}}$ and whose secondary axis is $\exp(-\norm{\bm{v}})\tl{\bm{v}}_\perp/\norm{\bm{v}}$. That is, the axes of the ellipses correspond to the eigenvectors of  $\bm{H}_{\bm{v}}$ scaled by their respective eigenvalues. \Cref{half-angle figure} shows visually how the anisotropy increases with $\norm{\bm{v}}$ and is directed towards $\tl{\bm{v}}$. It can also be seen how the parameterization is injective (no two ellipses are the same) and smooth (the ellipses vary smoothly with $\bm{v}$).
\begin{figure}[H]
    \centering
    \includegraphics{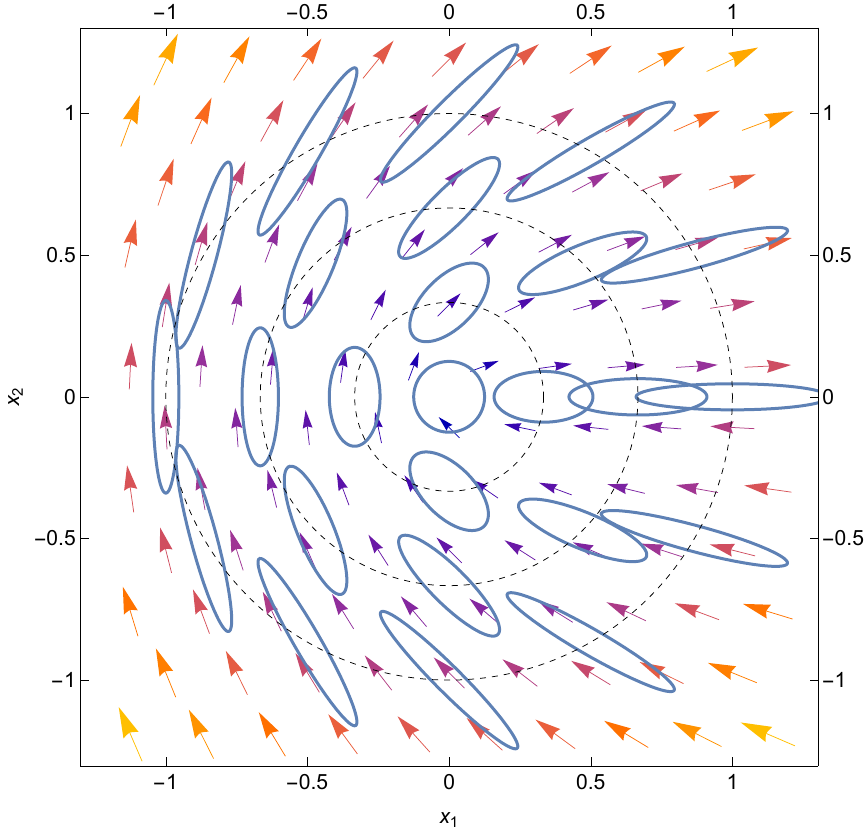}
    \caption{Figure of the half angle field $\tl{\bm{v}}$ and matrix $\bm{H}_{\bm{v}}$ in \eqref{half},\eqref{cool formula}. Each $\bm{H}_{\bm{v}}$ is represented by the ellipse centered at $\bm{v}$  with axes the eigenvectors of $\bm{H}_{\bm{v}}$ scaled by their respective eigenvalues.}
    \label{half-angle figure}
\end{figure}

In \Cref{fig:covariance_and_simulation_figure} (a), we show the plot of the covariance $K(\bm{x},\bm{0})=\mathbb{E}[u(\bm{x})u(\bm{0})]$ for $\kappa=1$ constant and as $\bm{v}$ starts at the $X$-axis and is rotated by $90^\circ$  in each plot. As we can see, the vector field is most correlated in the direction of $\tl{\bm{v}}$, which is rotated by $45^\circ$ in each image, half the speed of rotation of $\bm{v}$, and is least correlated in the direction of $\tl{\bm{v}}_\perp$. In Figure \Cref{fig:covariance_and_simulation_figure} (b), we show simulations of the field where again we leave $\kappa=1$ constant and rotate $\bm{v}$ by $90^\circ$ in each plot. The figure shows that the field diffuses the most in the direction of $\tl{\bm{v}}$ and the least in the direction of $\tl{\bm{v}}_\perp$. The realizations of $u$ are obtained using a FEM to solve the SPDE as detailed in \Cref{simulation section}.

\begin{figure}[H]
    \centering
    \begin{subfigure}[b]{0.45\linewidth}
        \centering
        \includegraphics[width=\linewidth]{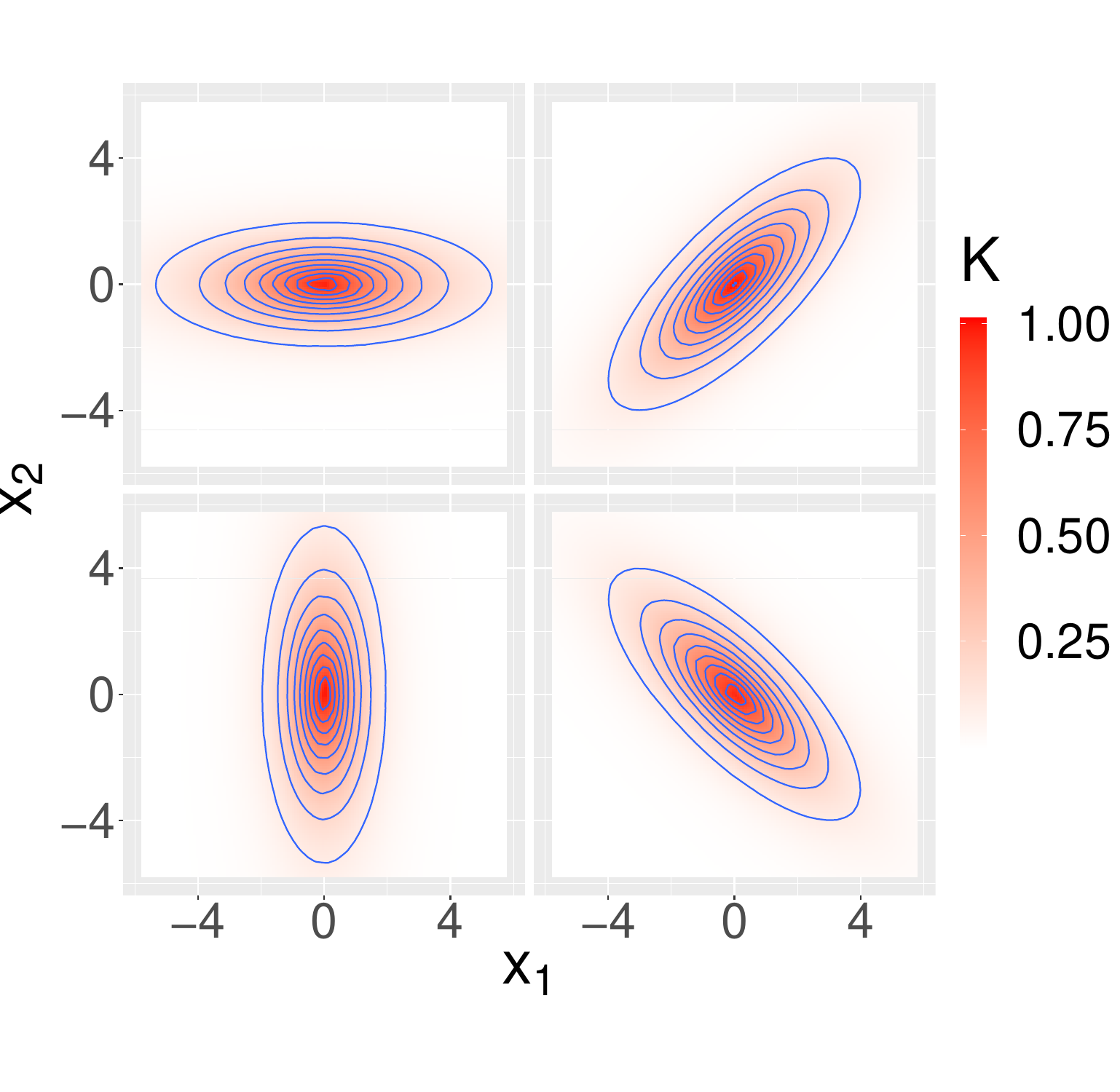}
        \caption{Covariance $K(\bm{x},\bm{0})$ as $\bm{v}$ varies}
        \label{fig:covariance_figure}
    \end{subfigure}
    \hfill
    \begin{subfigure}[b]{0.45\linewidth}
        \centering
        \includegraphics[width=\linewidth]{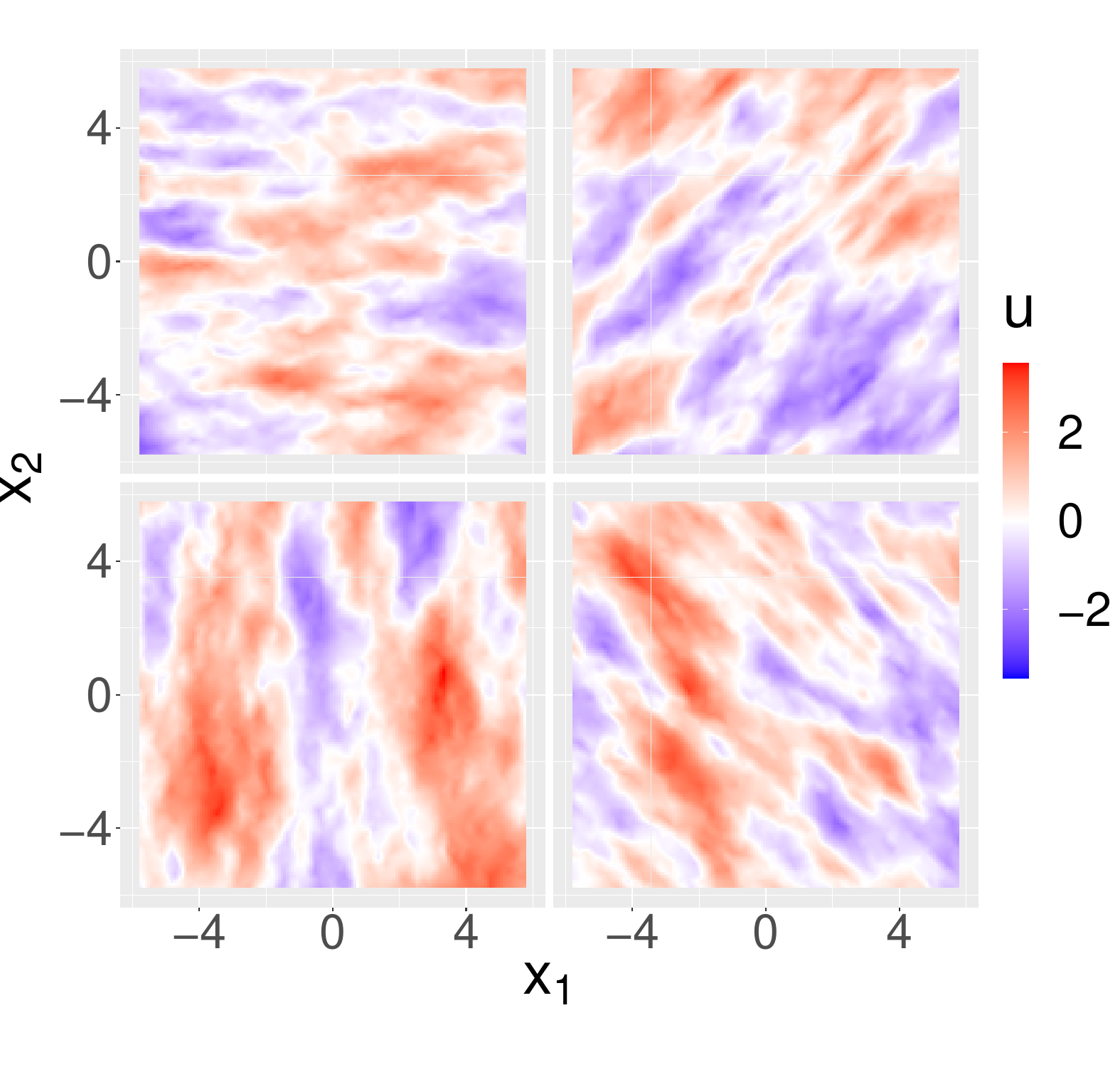}
        \caption{Realizations of $u$ as $\bm{v}$ varies}
        \label{fig:simulation_figure}
    \end{subfigure}
    \caption{We fix $\kappa =\sigma_u = 1$ and plot the covariance $K(\cdot,\bm{0})$  and a realization of the solution $u$ to \eqref{SPDE2D} as $\bm{v}$ varies from left to right and top to bottom through $(1,0),(0,1),(-1,0),(0,-1)$.}
    \label{fig:covariance_and_simulation_figure}
\end{figure}
Another method that could be used to simulate the field $u$  and obtain the variance plots is the \emph{spectral method}. The spectral method uses the \emph{spectral density} of the field (the Fourier transform of the covariance function $r$ if it exists)
\begin{equation}\label{sddef}
    S(\bm{\xi}):= \int_{\R^2 } \exp(-2\pi i\bm{\xi}\cdot\bm{x})r(\bm{x})\dx, \quad \bm{\xi}\in \R^2,
\end{equation}
to sample from the field. The field $u$ is then sampled using the stochastic integral.
The spectrum for stationary fields $u$ solving \eqref{SPDE2D} is given by (see \Cref{Sobolev distance section} and \citet{lindgren2012stationary})
\begin{equation}\label{spectral density}
    S(\bm{\xi})=\frac{4\pi \kappa^2\sigma_u^2}{(\kappa^2+4\pi^2\bm{\xi}^T\bm{H}\bm{\xi})^2}.
\end{equation}
Then, the covariance can be calculated as the Fourier transform of the spectrum, whereas the field $u$ can be sampled using the stochastic integral
\begin{equation}\label{spectral sim}
    u(\bm{x})=\int_{\R^2}\ex{2\pi i\bm{\xi}\cdot\bm{x}}\d Z(\bm{\xi}),
\end{equation}
where $Z(\bm{\xi})$ is called the \emph{spectral process}. Using the fast Fourier transform, high-resolution samples of $u$ can be obtained. See \cite{lindgren2024diffusion} for the details.

Thus far, we have worked with constant $\bm{v}$. The same parameterization goes through when $\bm{H}(\bm{x})$ is allowed to be spatially varying. In this case, $\bm{v}(\bm{x})$ is a spatial vector field. In \Cref{fig:anisotropic_figures}, we take $\kappa= \sigma_u=1$ and show the covariance $K(\bm{x},(2,2))$ and field $u$  when $\bm{v}(\bm{x})$ is chosen to be the “twice-angle” field of the rotational field $\bm{\tl{v}}(\bm{x})=(-x_2,x_1)$ (left of each subfigure), and from when $\bm{v}(\bm{x})=(x_1,x_2)$ (right of each subfigure). The figures show how the information of the field diffuses infinitesimally in the direction of $\tl{\bm{v}}$. The covariance and samples are obtained using the finite element method. 
\begin{figure}[H]
    \centering
    \begin{subfigure}[b]{0.45\linewidth}
        \centering
        \includegraphics[width=\linewidth]{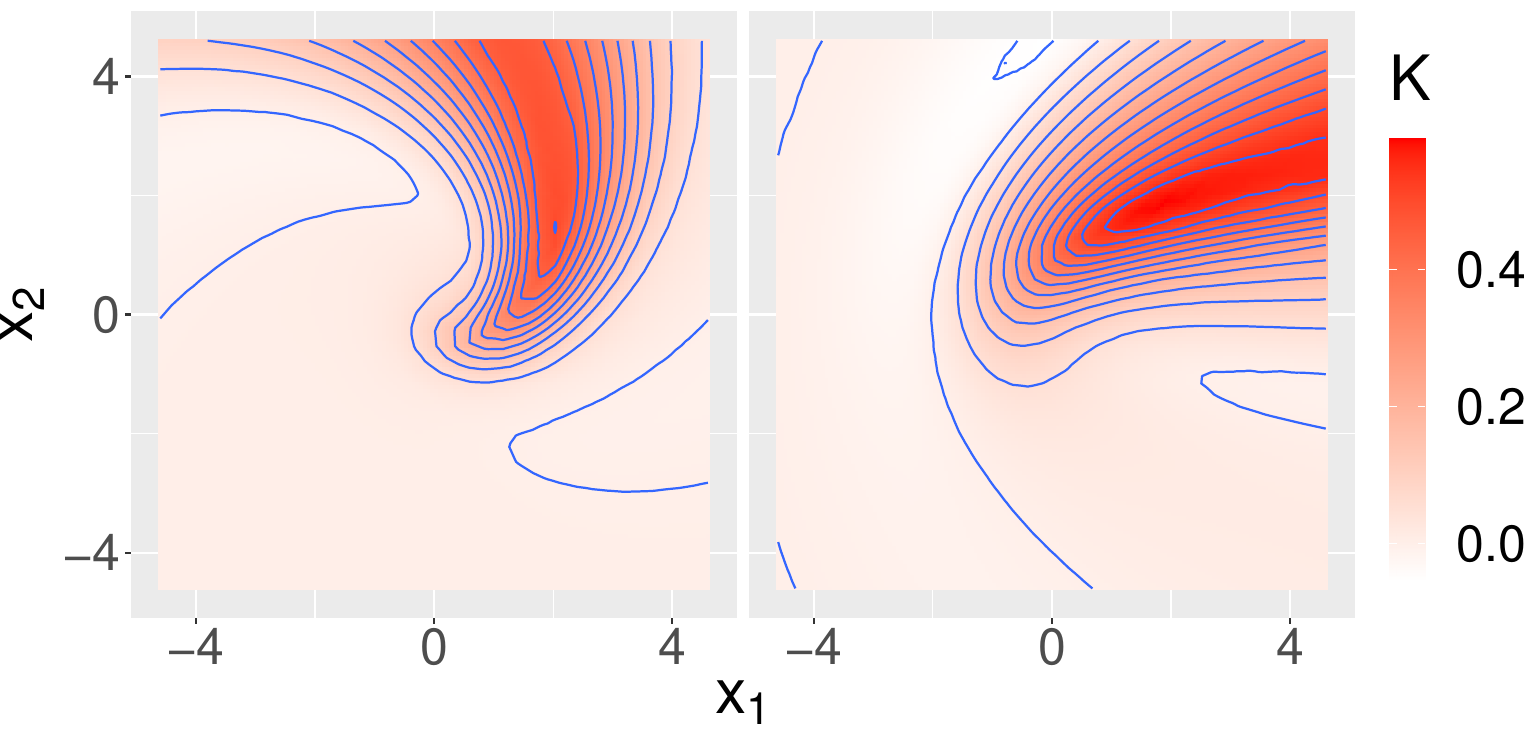}
        \caption{Covariances $K(\bm{x},(2,2))$}
        \label{fig:anisotropic_figure}
    \end{subfigure}
    \hfill
    \begin{subfigure}[b]{0.45\linewidth}
        \centering
        \includegraphics[width=\linewidth]{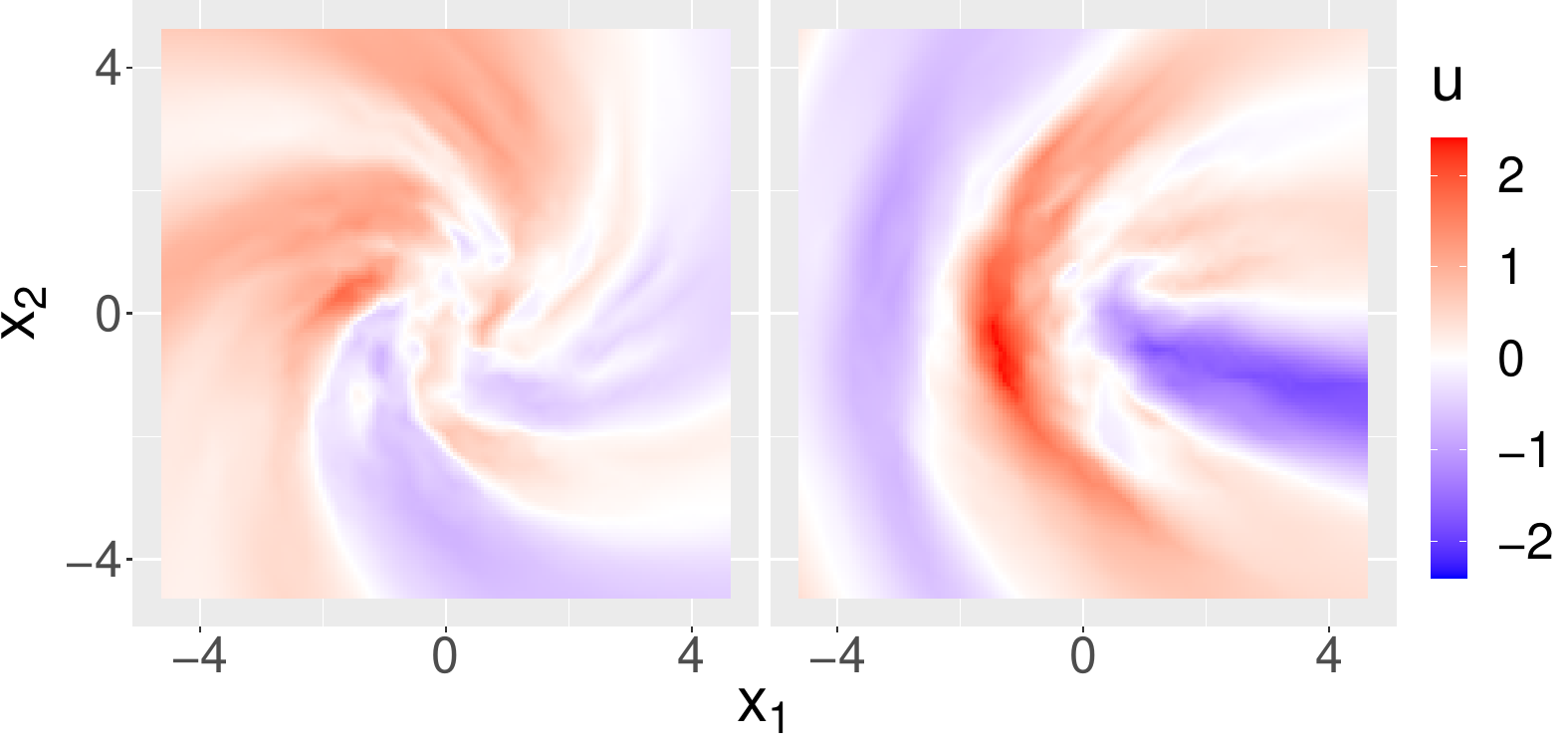}
        \caption{Realizations of $u(\bm{x})$}
        \label{fig:anisotropic_figure_2}
    \end{subfigure}
    \caption{We fix $\kappa =\sigma_u = 1$ and plot the covariance $K(\cdot,(2,2))$  and realizations of solution $u$ to \eqref{SPDE2D} for $\tl{\bm{v}}(\bm{x})=(-x_2,x_1)$ on the left and $\bm{v}(\bm{x})=(x_1,x_2)$ on the right of each subfigure.}
    \label{fig:anisotropic_figures}
\end{figure}

In summary, we have parameterized the anisotropic field $u$ using parameters $(\kappa, v_1,v_2,\sigma_u)$, where $u$ solves \eqref{SPDE2D} and $\bm{H}:= \bm{H_v}$ is given by \eqref{cool formula}. The parameterization is \revisiontwo{invertible} and smooth, and the parameters have an intrinsic geometric interpretation. This will be the parameterization used throughout the paper.
\section{Penalized complexity priors}\label{priors section}
PC priors were originally developed by \citet{Simpson2014PenalisingMC} to construct weakly informative priors while adhering to certain principles. The main idea is that one has a parametric family
of models $M_{\bm{\theta}}$ with parameter $\bm{\theta}$ \revisiontwo{belonging to some parameter space $E$} and a base model $M_{\bm{0}}$ (by convention corresponding to $\bm{\theta}=\bm{0}$).

One views $M_{\bm{0}}$ as the most suitable in the absence of any information. The larger the distance $\zeta(M_{\bm{\theta}}, M_{\bm{0}}) $ between a model $M_{\bm{\theta}}$ and $M_{\bm{0}}$, the smaller the prior density for $\bm{\theta}$ should be, and the decrease is set to be exponential. That is, we choose the prior distribution for $\bm{\theta}$, such that:
\begin{equation}\label{principle 1}
 d(\bm{\theta} ):=\zeta (M_{\bm{\theta}},M_{\bm{0}})\sim \text{Exp}(\lambda_{\bm{\theta}} ).
\end{equation}
Here, the rate $\lambda_{\bm{\theta}} > 0$ of the exponential distribution serves as a hyperparameter selected by the user, which governs the model's flexibility. A smaller value of $\lambda_{\bm{\theta}}$ increases the model's flexibility, allowing for greater deviations of $M_{\bm{\theta}}$ from the base model $M_{\bm{0}}$. Conversely, a larger value of $\lambda_{\bm{\theta}}$ imposes stricter penalties on these deviations, thereby reducing the model's flexibility. In the previously cited \citet{Simpson2014PenalisingMC}, this notion of ``distance'' was taken to be
\begin{equation*}
  \eta (M_{\bm{\theta}},M_{\bm{0}}):=\sqrt{2 \rmm{KLD}(M_{\bm{\theta}} | |M _0)},
\end{equation*}
where $\mathrm{KLD}$ is the Kullback-Leibler divergence
\begin{equation*}
  \text{KLD}(M_{\bm{\theta}} | | M_0 ):=\int_E \log \qty(\frac{\d M_{\bm{\theta}}}{\d M_0 })\d M_{\bm{\theta}} .
\end{equation*}
One possible complication of using the KL divergence to define the distance is that the Radon--Nikodym derivative $\frac{\d M_{\bm{\theta}}}{\d M_{\bm{0}} }$ does not exist. In fact, in the case of our model, if the variances of $M_{\bm{\theta }}$ and $M_{\bm{0}}$ are equal, for any possible base model $u_{\bm{0}}$ and $\bm{\theta}\neq \bm{0}$ the measures $M_{\bm{\theta}}$ and $M_{\bm{0}}$ are singular. That is, with the usual convention,
\begin{equation}\label{infinite KL}
  \text{KLD}(M_{\bm{\theta}} | | M_{\bm{0}} )=\infty .
\end{equation}
The proof of this fact is given in \Cref{KLD section}.
Equation \eqref{infinite KL} makes an exact adherence to the previous steps impossible. As a result, we will merely adopt the principles of the PC prior construction (exponential penalization of complexity) while altering how we measure complexity
between models. The idea of modifying the metric used to measure complexity has also been explored in different settings, such as in \citet{bolin2023wasserstein, Uribe}, where a Wasserstein distance was used. The possibility of using the Wasserstein distance to measure the complexity of the model \eqref{SPDE2D} was also considered. However, the Wasserstein distance is bounded in this setting (see \Cref{wasserstein appendix} in the supplementary material) and, as a result, was discarded. Thus, one of the main challenges is to find a computationally feasible distance that captures relevant information about our model. To this aim, we give the following definition.
\begin{definition}\label{metric def}
 Given two sufficiently regular models $u_A\sim M_A,u_B \sim M_B$, with respective spectral densities $S_A, S_B$ and variances $\sigma_A,\sigma_B$, we define the pseudometric
  \begin{equation}\label{pseudo metric}
 D_2(M_A ,M_B):=\qt{\int_{\R^2} \norm{2\pi\bm{\xi}}^{4}\qt{\frac{S_A(\bm{\xi})}{\sigma_A^2}-\frac{S_B(\bm{\xi})}{\sigma_B^2}}^2 \d\bm{\xi}}^{\frac{1}{2}} .
  \end{equation}
\end{definition}
The above definition uses the rescaled spectral densities of each field to define a Sobolev seminorm on the difference of the correlations of $M_A$ and $M_B$. \revisionone{Sobolev norms are a common way to measure the distance between two functions due to their sensitivity to smoothness and shape. In this case, it measures the difference in speed of oscillation of the normalized covariance up to order $2$.}

\revisionone{Other measures of complexity between models are also possible, and we give such an example in Appendix \ref{Sobolev distance section}. If a particular application suggests the clear use of a specific distance, that distance should be used instead.} 

\revisionone{In this work, we considered many other common distances and showed that they were unsuitable. As we already mentioned, the KL divergence between two different models is infinite. Furthermore, the Wasserstein distance, the $L^2$ distance, and the Hellinger distance are bounded, making them unsuitable for the exponential penalization of PC priors. Sobolev seminorms of integer order different from $2$ can also be shown to be unsuitable for this model with this smoothness.} \revisiontwo{For different smoothneess and dimension than the ones considered here, the distance can be adapted by modifying the weight used to make the integrand finite. The weight could also be adapted for models whose spectral density has a singularity at the origin of order $\alpha$ by multiplying by a factor $\norm{\b{\xi}}^\alpha (1+\norm{\b{\xi}})^{-\alpha}$.}

\revisionone{Finally, we note that the seminorm distinguishes different models (it defines a norm on our family of models) in the following sense. Write $M_{\kappa,\bm{v}}$ for the model given by SPDE \eqref{SPDE2D} with parameters $(\kappa,\bm{v})$ and with variance fixed to $\sigma _u=1$. Due to the rescaling of the spectral densities in \Cref{metric def}, the choice of $\sigma _u$ does not affect the distance between models and can be set to any other positive value. Then,}
\revisionone{\begin{align*}
 D_2(M_{\kappa_A,\bm{v}_A},M_{\kappa_B,\bm{v}_B})&=0 \iff (\kappa_A, \b{v}_A)=(\kappa _B, \b{v} _B) .
\end{align*}}
\revisionone{We give a more detailed discussion in Section \ref{Sobolev distance section} of the Appendix.}

 We define the base model as the limit in distribution of $M_{\kappa,\bm{v}}$ as $\kappa,\bm{v}$ go to zero
\begin{align*}
 {M}_{\bm{0}}:= \lim_{(\kappa_0,\bm{v}_0) \to \bm{0}} {M}_{\kappa_0,\bm{v}_0} =\Nn\qt{0,\mathbf{1}\otimes\mathbf{1}},
\end{align*}
where $\mathbf{1}: \mathcal{D}\rightarrow \mathbb{R}$ is the constant function, $z \in \Dd \mapsto 1$ and $(f\otimes g)(\bm{x},\bm{y}):=f(\bm{x})g(\bm{y})$.
That is, $u_{\bm{0}}$ is constant over space and follows a Gaussian distribution with variance $1$. We choose this $u_{\bm{0}}$ as our base model because it is simple, and $(\kappa,\bm{v}) = \bm{0}$ is the only distinguished point in the parameter space.

To reflect the dependency of this distance on the parameters, we use the notation
\begin{equation}\label{distance metric}
 d(\kappa,\bm{v}):=D_2(M_{\kappa,\bm{v}},M_{\bm{0}}) := \lim_{(\kappa_0,\bm{v}_0) \to \bm{0}}D_2(M_{\kappa,\bm{v}},M_{\kappa_0,\bm{v}_0}).
\end{equation}

The exponential penalization in \eqref{principle 1} imposes one condition on the prior of $(\kappa,\bm{v}) $ whereas, since $(\kappa,\bm{v}) $ is three dimensional, two more conditions are necessary to determine the prior distribution uniquely. \citet{fuglstad2019constructing} circumvented this issue by working iteratively, fixing one parameter while allowing the other to vary, building PC priors for each parameter, and then multiplying them together to obtain a joint prior. However, we prefer to work jointly from the start. This approach is made possible by the structure of $d(\cdot,\cdot )$, which is calculated to have the form
\begin{equation}\label{distance form}
 d(\kappa,\bm{v})=f(\norm{\bm{v}})g(\kappa)
\end{equation}
Since the angle $\alpha:= \arg(\bm{v})$ does not affect \eqref{distance form}, we impose that $\alpha$ is uniformly distributed in $[0,2\pi)$. This choice guarantees that we do not favor the alignment of the covariance of $u$ in any direction of the plane. Next, since the contribution of $f$ and $g$ to the distance is symmetric, we impose that $f-f(0)$ and $g$ knowing $f$ are exponentially distributed. The translation is necessary as $f$ takes a nonzero minimum $f(0)$ at $\norm{\bm{v}}= 0$, whereas $g$ takes a minimum of $0$ at $\kappa =0$.

This construction leads to the distance $d(\kappa,\bm{v})$ being exponentially distributed, and it provides us with three conditions that uniquely determine $(\kappa,\bm{v})$.

\begin{theorem}[PC prior for $(\kappa,\bm{v})$]\label{Prior r kappa stationary theorem}
 A PC prior for $(\kappa,\bm{v}) $ with base model $(\kappa,\bm{v}) =\bm{0}$ is
  \begin{align}\label{pc1}
    \pi(\kappa,\bm{v}) & =\frac{\lambda_{\bm{\theta}} \lambda_{\bm{v}}f'(\norm{\bm{v}} ) f(\norm{\bm{v}} ) }{2\pi \norm{\bm{v}} } \exp({-\lambda_{\bm{v}}\qty( f(\norm{\bm{v}} )-f(0) )} -\lambda_{\bm{\theta}} f(\norm{\bm{v}} ) \kappa ),\end{align}
 where $\lambda_{\bm{\theta}}>0 ,\lambda_{\bm{v}}>0$ are hyperparameters and
  \begin{align}\label{pc2}
 f(r):=\qty(\frac{\pi }{3 } (3 \cosh (2 r )+1))^\frac{1}{2}, \quad f'(r)=\frac{\sqrt{\pi } \sinh (2 r)}{\sqrt{\cosh (2 r)+1/3}}.
  \end{align}
\end{theorem}
See Section \ref{distance calc stationary sec} in the supplementary material for a proof of this and other results of this section.
We plot the marginal density of the prior on $\kappa $ and $\bm{v}$ for $\lambda_{\bm{\theta}}=\lambda_{\bm{v}}=16\pi^2$ in \Cref{prior figure}. The marginal prior densities take a maximum at $\kappa =0$ and $\bm{v}=0$ respectively, and by construction, the prior on $\bm{v}$ is radially symmetric.

\begin{figure}[H]
  \centering
  \includegraphics[width=\textwidth]{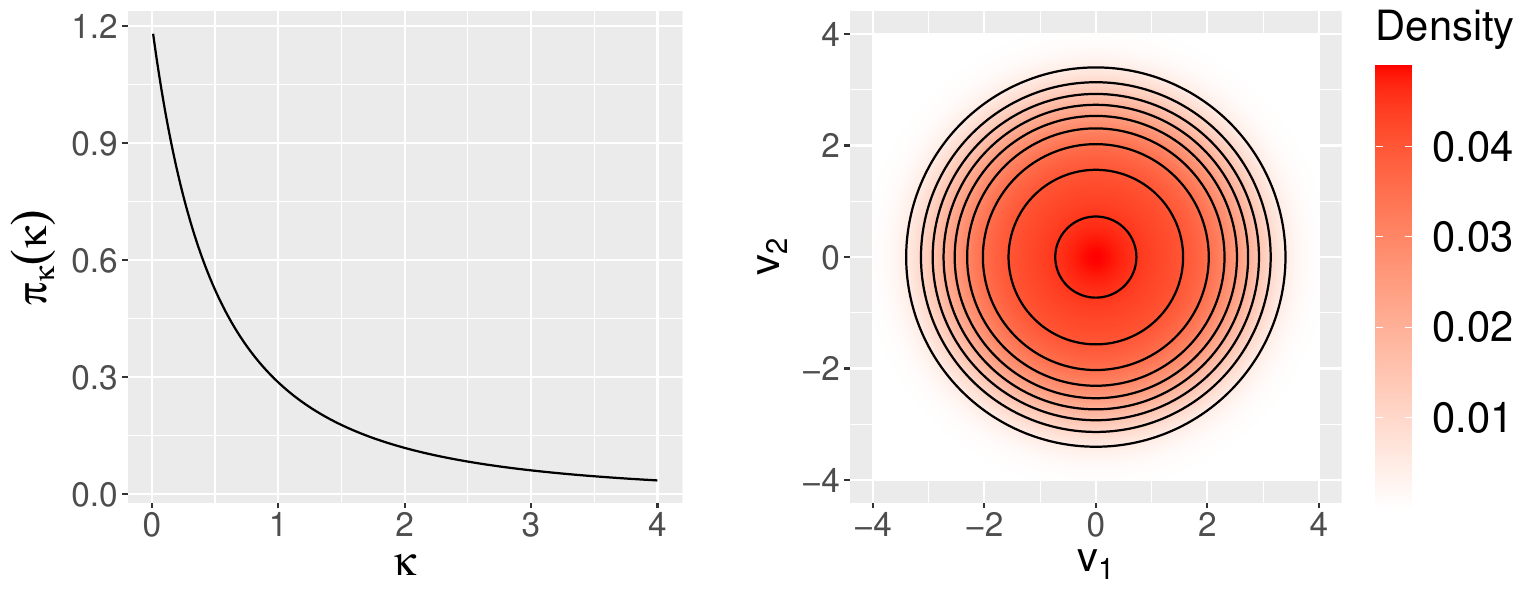}
  \caption{Marginal PC prior density of $\kappa $ and $\bm{v}$ obtained in \Cref{Prior r kappa stationary theorem} for $\lambda_{\bm{\theta}}=\lambda_{\bm{v}}=16\pi^2$.  }
  \label{prior figure}
\end{figure}

The hyperparameter $\lambda_{\bm{\theta}} $ determines the flexibility of the model (how much we penalize large values of $d(\kappa,\bm{v} )$), whereas $\lambda_{\bm{v}}$ controls the degree of anisotropy (how much we penalize large values of $\norm{\bm{v}}$). Their values can be set to agree with desired quantiles using the following two results.
\begin{theorem}\label{par r theorem}
 The prior for $r:=\norm{\bm{v}} $ satisfies $\mathbb{P}[r >r_0]=\beta $ if and only if
  \begin{equation*}
    \lambda_{\bm{v}} =-\frac{\log(\beta )}{f(r_0)-f(0)} .
  \end{equation*}
\end{theorem}

\begin{theorem}\label{par kappa theorem} The prior for $\kappa $ satisfies $\mathbb{P}[\kappa >\kappa _0]=\alpha$ if and only if
  \begin{equation*}
    \lambda_{\bm{\theta}} = \frac{1}{\kappa_0} \qty(\frac{1}{f(0)} W_0\qty(\frac{\ex{\lambda_{\bm{v}} f(0)}\lambda_{\bm{v}}f(0)}{\alpha} ) -\lambda_{\bm{v}}),
  \end{equation*}
 where $W_0$ is the principal branch of the Lambert function. That is, $W_0(x)$ is the real-valued inverse of $x\ex{x}$ for $x\ge 0$,
  \begin{align*}
 x=W_0(x)\ex{W_0(x)}, \quad\forall x \geq 0.
  \end{align*}
\end{theorem}
When specifying values of $\lambda_{\bm{\theta}},\lambda_{\bm{v}}$, it is useful to consider the ratio between the eigenvalues of $\bm{H}_{\bm{v}}$ and the empirical correlation range. These measure respectively how much more correlated the field is in one direction in space and the distance at which the field becomes essentially uncorrelated (see Section \ref{parameterization section}). We denote these by
\begin{align*}
 a:= \exp(\norm{\bm{v} }), \quad \rho:=\sqrt{8}{\kappa^{-1} } .
\end{align*}
\Cref{par r theorem} and \Cref{par kappa theorem} can then be rewritten as follows.
\begin{theorem}\label{quantiles theorem}
 The PC priors satisfy that $\mathbb{P}[a>a_0]=\beta $ and $\mathbb{P}[\rho <\rho_0]=\alpha$ if and only if
  \begin{align*}
    \lambda_{\bm{v}} =-\frac{\log(\beta )}{f(\log(a_0)) -f(0)}, \quad \lambda_{\bm{\theta}} = \frac{\rho _0}{\sqrt{8} } \qty(\frac{1}{f(0)} W_0\qty(\frac{\ex{\lambda_{\bm{v}} f(0)}\lambda_{\bm{v}}f(0)}{\alpha} ) -\lambda_{\bm{v}})
 .
  \end{align*}
\end{theorem}
In a practical application, $\alpha,\beta $ can be chosen to be small (for example $0.01$), $a_0$ can be chosen to be an unexpectedly large amount of anisotropy, and $\rho_0$ can be chosen as a surprisingly small correlation range for the field.

The parameters $(\kappa,\bm{v}) $ can be written as a joint transformation of a three-dimensional vector with a standard multivariate Gaussian distribution. Thus, it is possible to efficiently generate independent samples $(\kappa,\bm{v})$ through sampling multivariate Gaussian random variables, which is straightforward.

\begin{theorem}\label{reparameterization theorem stationary}
 Let $\bm{Y} \sim \Nn(\bm{0}, \mathbf{I}_3)$ and write respectively $\Phi, R$ for the CDFs of a univariate standard Gaussian and Rayleigh distribution with shape parameter $1$. Define
  \begin{align*}
 A := \sqrt{Y_1^2 + Y_2^2}, \quad B := f(0) - {\log(1 - R(A))}/{\lambda_{\bm{v}}}.
  \end{align*}
 Then, it holds that \revisiontwo{for $f$ as in \eqref{pc2} }
  \begin{align*}
 (v_1, v_2, \kappa) \stackrel{d}{=} \varphi(Y_1, Y_2, Y_3) := \left( f^{-1}(B) \frac{Y_1}{A}, f^{-1}(B) \frac{Y_2}{A}, -\frac{\log(1 - \Phi(Y_3))}{\lambda_{\bm{\theta}} B} \right),
  \end{align*}
 where $f^{-1}(x) = \frac{1}{2} \cosh^{-1} \left( \frac{x^2}{\pi } - \frac{1}{3} \right)$.
\end{theorem}
The transformation $\varphi$ only involves standard functions and can be evaluated efficiently. For the proof, see Section \ref{reparameterization section} in the appendix.

To assess the practical performance of this theoretically-grounded prior against common alternatives, we next present a comprehensive simulation study

\section{Simulation study}\label{simulation section}
\subsection{Framework}
This section studies the performance of the PC priors. To do so, we set different priors $\pi_{\kappa, \bm{v}}$ on $\kappa, \bm{v}$, observe a noisy realization of the field, and compare the behavior of the posteriors, resulting from each prior. We consider the random field $u$ that solves our model \eqref{SPDE2D} on the square domain $\Dd =[0,10]^2$ and define the
observation process
\begin{align}\label{obs}
  \bm{y}= \bm{A u} +\bm{\varepsilon},
\end{align}
where $\bm{u}\in \mathbb{R}^n$ is a discrete approximation of the solution to \eqref{SPDE2D} obtained through a FEM on a mesh with $n$ nodes, and $\bm{A} \in \R^{m \times n}$ linearly interpolates to observe $\bm{u}$ at $m=15$ locations $\set{\bm{x}_j}_{j=1}^m$ which are obtained by sampling uniformly from $\Dd$, and $\bm{\varepsilon} \in \R^m$ is a noise vector with
\begin{equation}
  \bm{\varepsilon} \sim \Nn(\bm{0},\bm{Q_\varepsilon} ^{-1}),\quad \bm{Q}_{\bm{\varepsilon}}:=\sigma_{\bm{\varepsilon}}^{-2}\bm{I}_{m}.
\end{equation}

\revisionone{We choose a small number of data points as the focus of this section is on comparing the effect of different priors against the PC priors. For a large amount of data, it is difficult to discern the effect of the priors in the posterior; however, for a small amount of data, the effect of the prior is more pronounced, and the penalization on complexity is more relevant.}

The details of how $\bm{u}$ is sampled from and more detailed results can be found in \Cref{simulation details}.
We set independent priors
\begin{equation}\label{PC priors}
 (\kappa , \bm{v}) \sim \pi_{\kappa , \bm{v}}, \quad \sigma_u\sim \mathrm{Exp}(\lambda_{\sigma_u}), \quad \sigma_{\bm{\varepsilon}}\sim \mathrm{Exp}(\lambda_{\sigma_{\bm{\varepsilon}}}).
\end{equation}
The priors on $\sigma_u, \sigma_{\bm{\varepsilon}}$ correspond to their respective PC priors, see respectively \citet[Section 3.3]{Simpson2014PenalisingMC}, \citet[Theorem 2.1]{fuglstad2019constructing}. 

We compare the performance of our PC prior $\pi_{\rmm{PC}} $  against three alternative prior specifications, an Exponential-Gaussian prior $\pi_{\rmm{EG}}$  and uniform and beta priors $\pi_{\rmm{U}},\pi_\beta$  which are common non-informative choices. The detailed definitions of these alternative priors are provided in Appendix \ref{app:prior comparison}.

We simulate ${\bm{\theta }^{(j)}}^{\rmm{true}}$ from $\pi_\rmm{sim}\in \{\pi_\rmm{PC}, \pi_\rmm{EG}, \pi_\rmm{U}, \pi_\beta\}$, use the FEM to simulate $\bm{u}^{(j)}$ from \eqref{SPDE2D} and then simulate $\bm{y}^{(j)}$ from \eqref{obs}. Then, for each of the four priors $\pi_{\kappa,\bm{v}} \in \{\pi_\rmm{PC}, \pi_\rmm{EG}, \pi_\rmm{U}, \pi_\beta\}$ we approximate the posterior distribution of the parameters $\bm{\theta }=(\kappa, \bm{v}, \sigma_u, \sigma_{\bm{\varepsilon}})$ given the data $\bm{y}^{(j)}$ and evaluate the performance of each of the four priors. This process is repeated $J~=~600$ times and repeated for each of the four possible values of $\pi_{\rmm{sim}}$. Since the posterior is not available in closed form, we approximate it using Pareto smoothed importance sampling. See Appendix \Cref{Posterior sampling appendix} for the details.

\subsection{Results}
The focus of our study is the anisotropy parameters $(\kappa, v_1, v_2)$. As a result, in this section, we discuss the performance of the different priors on these parameters, reserving the results for the remaining parameters for \Cref{simulation details}.

We first show in \Cref{fig: MAP distances} the empirical cumulative distribution function (eCDF) of the vector of distances of the true anisotropy parameters
$(\log(\kappa^{\rmm{true}}),v_1^{\rmm{true}},v_2^{\rmm{true}})$ to the MAP estimates $(\log(\wh{\kappa}),\wh{v}_1,\wh{v}_2)$ for each of the four different distributions on $\bm{\theta}^{\rmm{true}}$.

We observe that $\pi_{\rmm{PC}}$ and $\pi_{\rmm{EG}}$ perform the best and give almost identical results, to the point where it is difficult to distinguish them from the plots. This behavior is to be expected, as both of these priors are similar and lead to almost identical posteriors, as is discussed further in \Cref{simulation details}.

The difference in performance between $\pi_{\rmm{PC}}, \pi_{\rmm{EG}}$ and $\pi_{\rmm{U}}, \pi_\beta $ is clearest for $v_1,v_2$ whereas for $\log(\kappa)$, when $\bm{\theta}^{\rmm{true}}\sim \pi_{\rmm{U}} $ or $\bm{\theta}^{\rmm{true}}\sim \pi_\beta $ all four models give comparable results. The figures relative to the parameters $v_1$ and $v_2$ are almost identical. This is expected by the symmetry in the priors and the likelihood, as there is no preferred direction of anisotropy.
\begin{figure}[H]
  \centering
  \includegraphics[width=\textwidth]{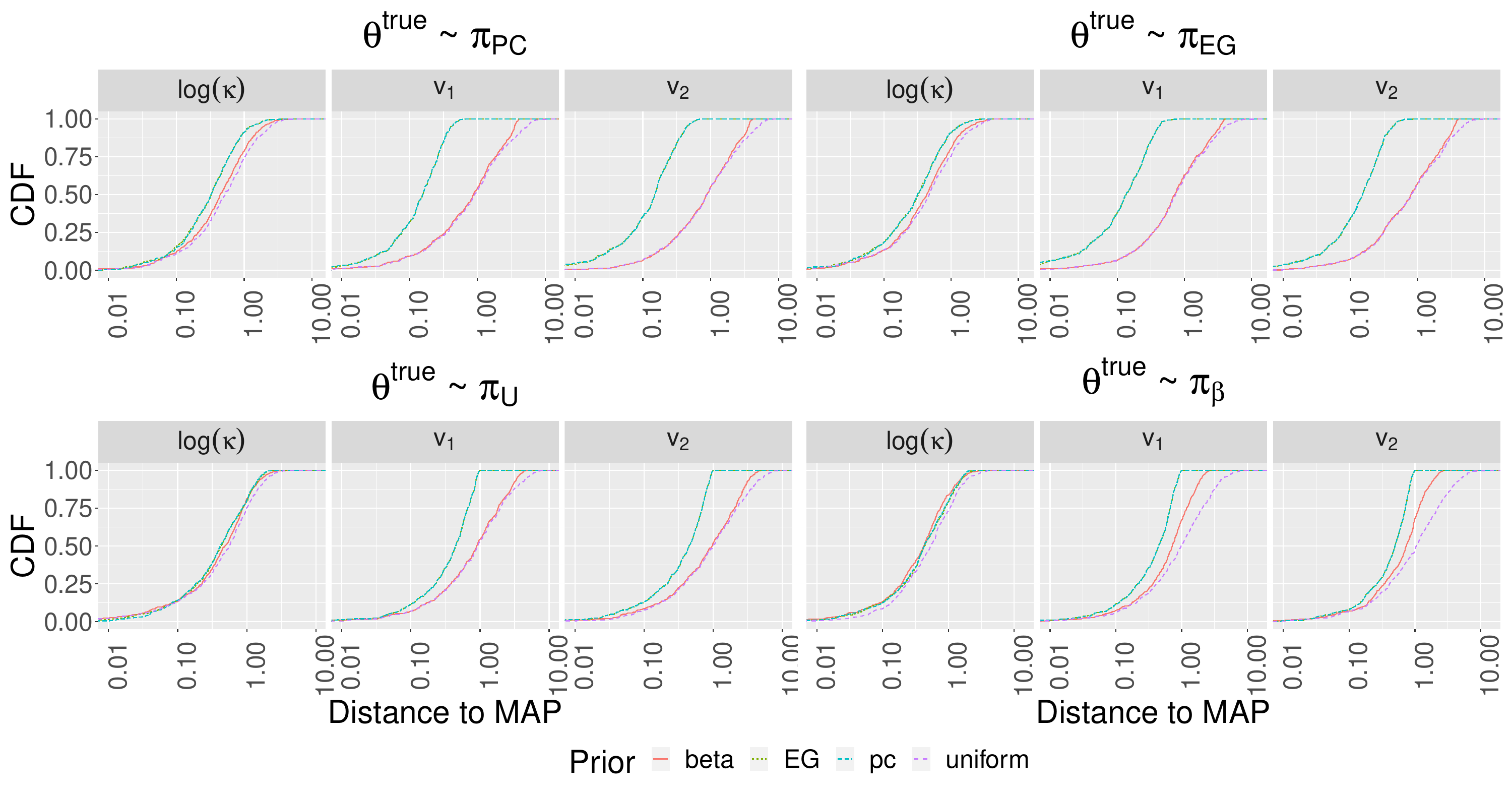}
  \caption{Empirical cumulative distribution function (eCDF) of the absolute distances between true parameter value ($\theta_i^{\text{true}}$) and their estimated values ($\widehat{\theta}_i$) for each of $\log(\kappa)$, $v_1$, $v_2$. In the plots, arranged from left to right and top to bottom, $\bm{\theta}^{\text{true}}$ is simulated from the four distributions—$\pi_{\text{PC}}$, $\pi_{\text{EG}}$, $\pi_{\text{U}}$, and $\pi_{\beta}$. Then, $\bm{y} = \bm{A} \bm{u} + \bm{\varepsilon}$ is observed with true parameter value $\bm{\theta}^{\text{true}}$. Finally, for each prior (red, green, teal, purple), the MAP estimate $\widehat{\bm{\theta}}$ is computed, and the eCDF over 600 simulations of the distances between $\bm{\theta}^{\text{true}}$ and $\widehat{\bm{\theta}}$ is plotted.}
  \label{fig: MAP distances}
\end{figure}
Next, in \Cref{fig: CI lengths}, we show the length of the symmetric equal-tailed $0.95$ credible interval for each parameter. As can be seen from the plots, both $\pi_{\rmm{PC}}$ and $\pi_{\rmm{EG}}$ limit the length of the credible intervals to a similar extent while $\pi_{\rmm{U}},\pi_\beta$ give much wider credible intervals. This aligns with the motivating factor behind the construction of the PC priors, which is to penalize the model's complexity.
\begin{figure}[H]
  \centering
  \includegraphics[width=\textwidth]{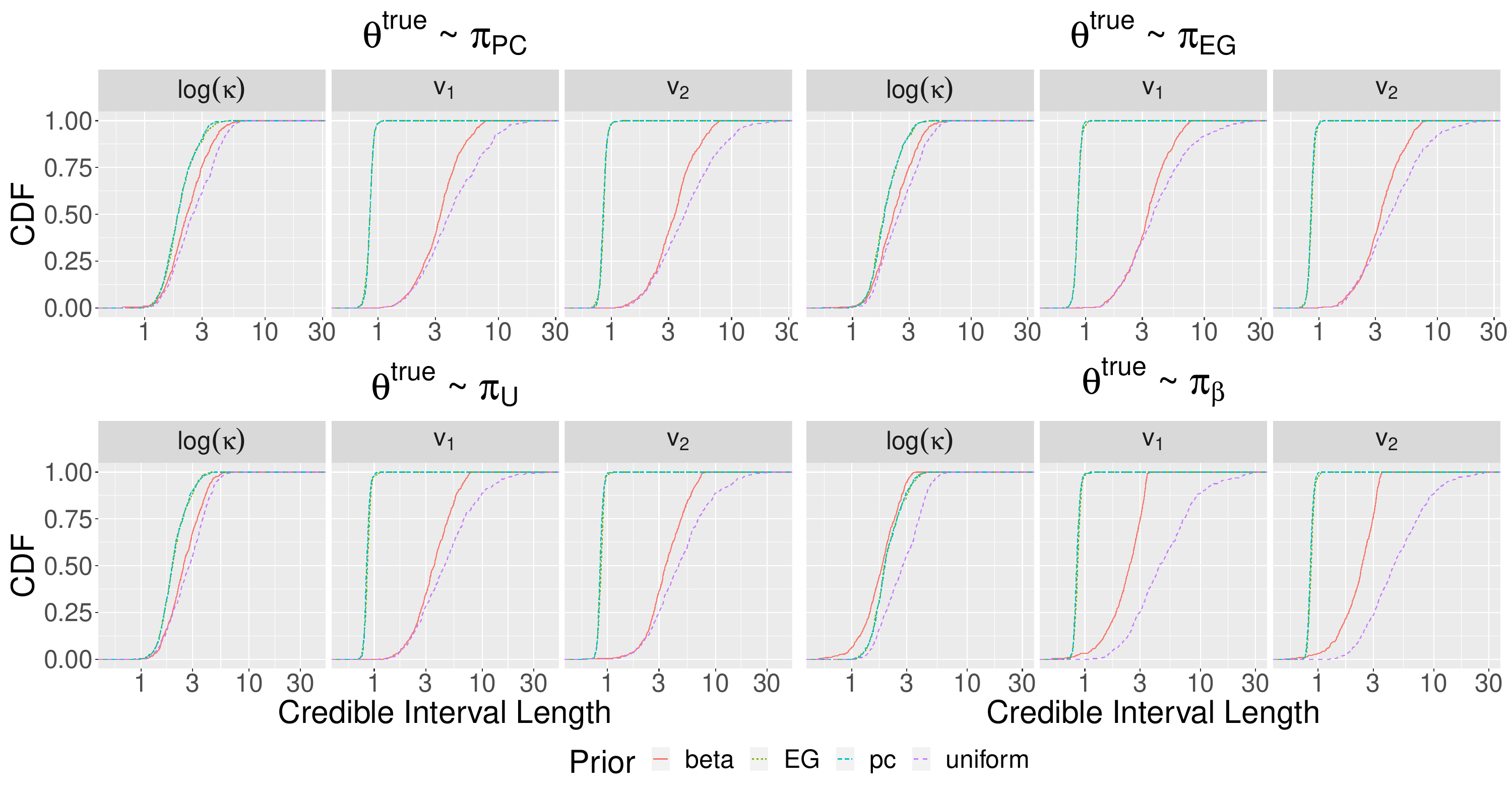}
  \caption{Empirical cumulative distribution function (eCDF) of the length of symmetric equal-tailed $0.95$ credible intervals (CIs) for the posterior of $\log(\kappa)$, $v_1$, $v_2$. In the plots, arranged from left to right and top to bottom, $\bm{\theta}^{\text{true}}$ is simulated from the four distributions—$\pi_{\text{PC}}$, $\pi_{\text{EG}}$, $\pi_{\text{U}}$, and $\pi_{\beta}$. Then, $\bm{y} = \bm{A} \bm{u} + \bm{\varepsilon}$ is observed with true parameter value $\bm{\theta}^{\text{true}}$. Finally, for each prior (red, green, teal, purple), the length of the CIs is estimated using smoothed importance sampling as explained in \Cref{Posterior sampling appendix} in the supplementary material.}
  \label{fig: CI lengths}
\end{figure}
In \Cref{fig: complexity}, we show the eCDF of the posterior mean complexity $\E_{\pi _{\bm{\theta }|\bm{y}}}\qb{d(\kappa , \bm{v})}$ , where the complexity $d$ is defined in \eqref{metric def} for each of the four different true distributions on $\bm{\theta}$. As can be seen, in all four cases, the PC and exponential-Gaussian priors reduce model complexity as compared to the uniform and beta priors.
\begin{figure}[H]
  \centering
  \includegraphics[width=\textwidth]{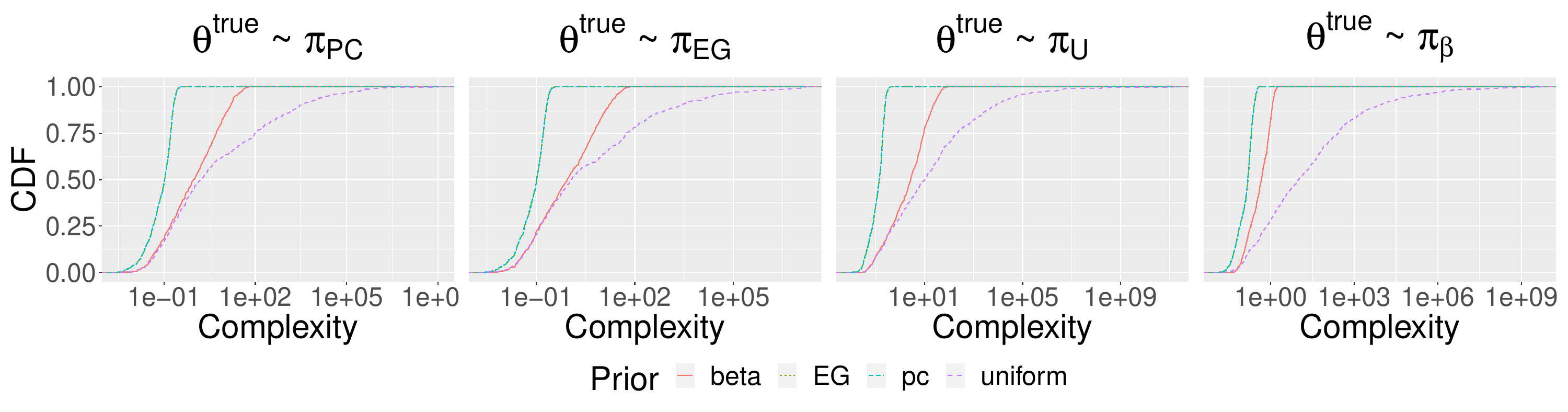}
  \caption{Empirical CDF of the complexity $d(\wh{\kappa } , \wh{\bm{v}})$ defined in \eqref{distance metric}. In the plots, arranged from left to right , $\bm{\theta}^{\text{true}}$ is simulated from the four distributions—$\pi_{\text{PC}}$, $\pi_{\text{EG}}$, $\pi_{\text{U}}$, and $\pi_{\beta}$. Then, $\bm{y} = \bm{A} \bm{u} + \bm{\varepsilon}$ is observed with true parameter value $\bm{\theta}^{\text{true}}$. Finally, for each prior (red, green, teal, purple), $\E_{\pi _{\bm{\theta }|\bm{y}}}\qb{d(\kappa, \bm{v})}$ is computed using smoothed importance sampling.}
  \label{fig: complexity}
\end{figure}

\section{An application to rainfall data}\label{precipitation section}
\subsection{Framework}
In this section, we analyze the performance of the anisotropic model and the PC priors on a dataset for total annual precipitation in southern Norway between September 1, 2008, and August 31, 2009. This data set was studied by \citet{Simpson2014PenalisingMC} and \citet{ingebrigtsen2015estimation} using a linear model.
\begin{align}\label{height model}
 y_i = \beta _0 + \beta _1 h_{i} + u(\bm{x}_i)+ \varepsilon _i, \quad i=1,\ldots,m= 233,
\end{align}
where $y_i$ is the total annual precipitation at location $\bm{x}_i$, $h_i$ is the altitude at location $\bm{x}_i$, $u(\bm{x}_i)$ is a spatially correlated random effect, and $\bm{\varepsilon}\sim \Nn(\bm{0},\sigma_{\bm{\varepsilon}}^2\bm{I})$ is a noise term.
In the articles above, $u$ was modeled using a non-stationary Matérn process
\begin{align*}
 (\kappa^2(\bm{x})-\Delta)\qty(\frac{u(\bm{x})}{\sigma_u(\bm{x})})=\sqrt{4\pi}\kappa(\bm{x})\dot{\Ww},
\end{align*}
and the non-stationarity $\kappa^2(\bm{x}),\sigma_u(\bm{x})$ was parameterized by a log-linear model with covariates elevation and gradient of elevation. \citet{Simpson2014PenalisingMC} then built PC priors for these extra parameters.
\begin{figure}[H]
  \centering
  \begin{subfigure}[b]{0.32\linewidth}
    \centering
    \includegraphics[width=\linewidth]{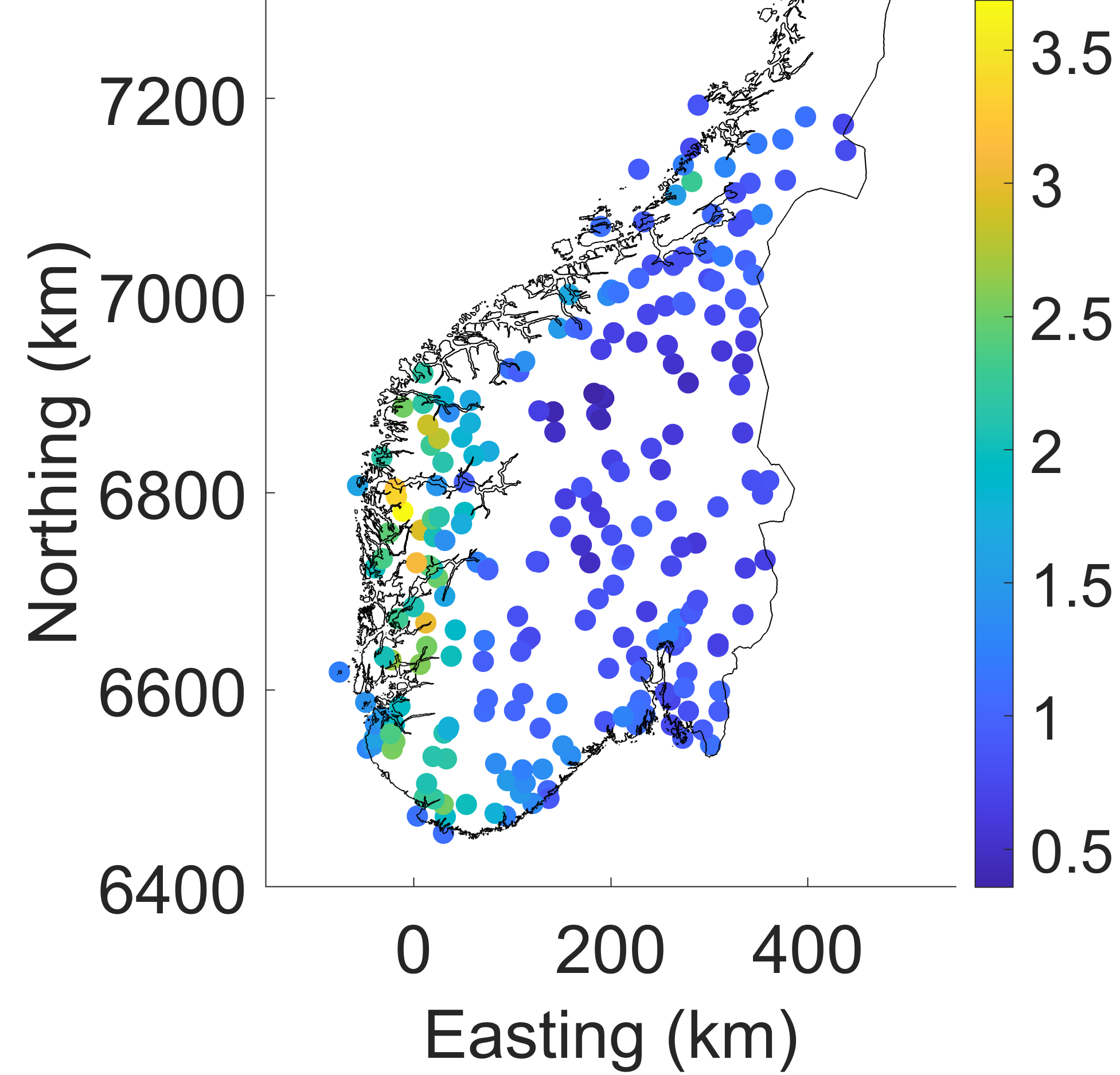}
    \caption{Observed precipitation $\bm{y}$}
    \label{fig:obs_precip}
  \end{subfigure}
  \hfill
  \begin{subfigure}[b]{0.32\linewidth}
    \centering
    \includegraphics[width=\linewidth]{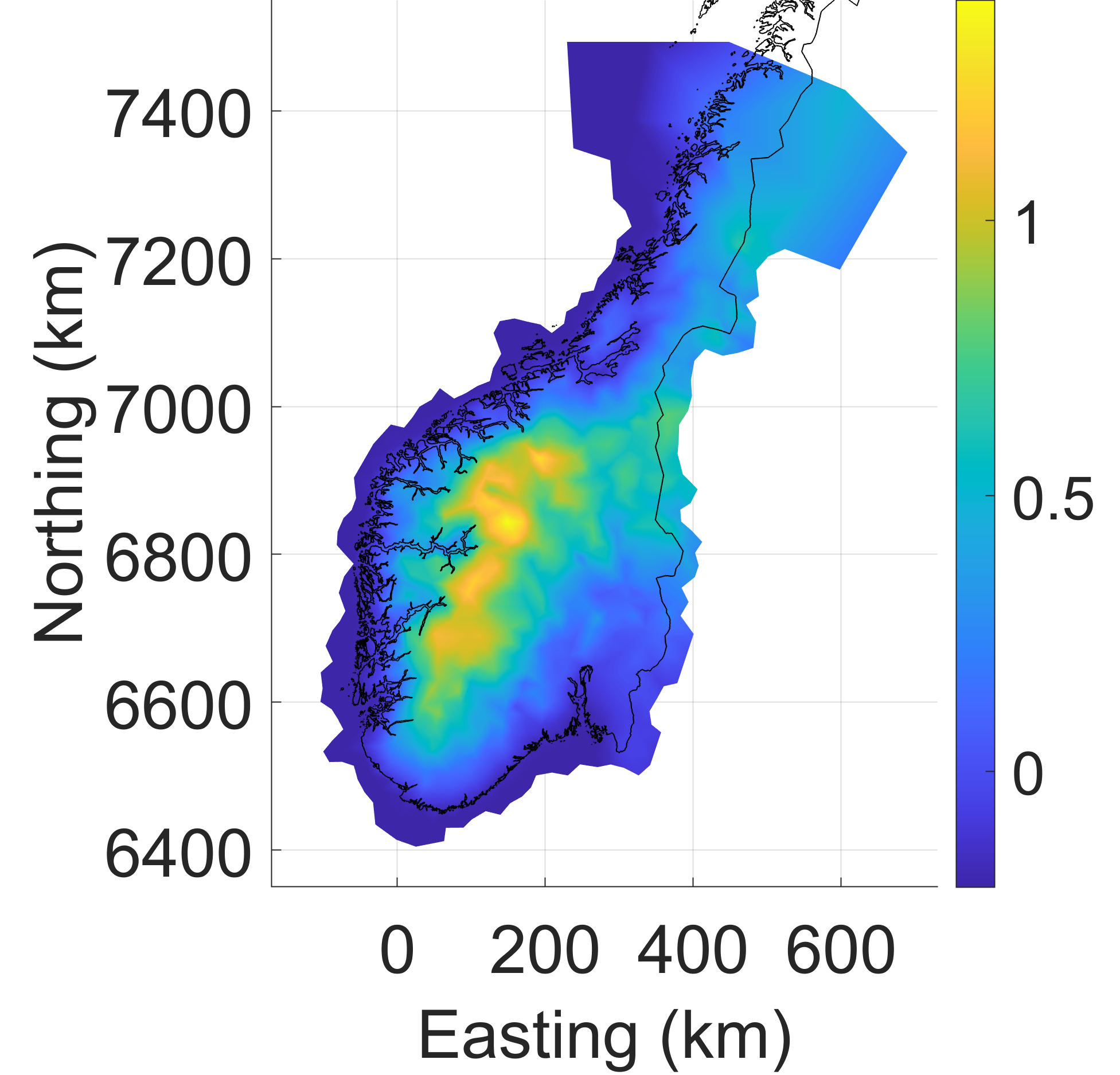}
    \caption{Elevation $h$}
    \label{fig:height}
  \end{subfigure}
  \hfill
  \begin{subfigure}[b]{0.32\linewidth}
    \centering
    \includegraphics[width=\linewidth]{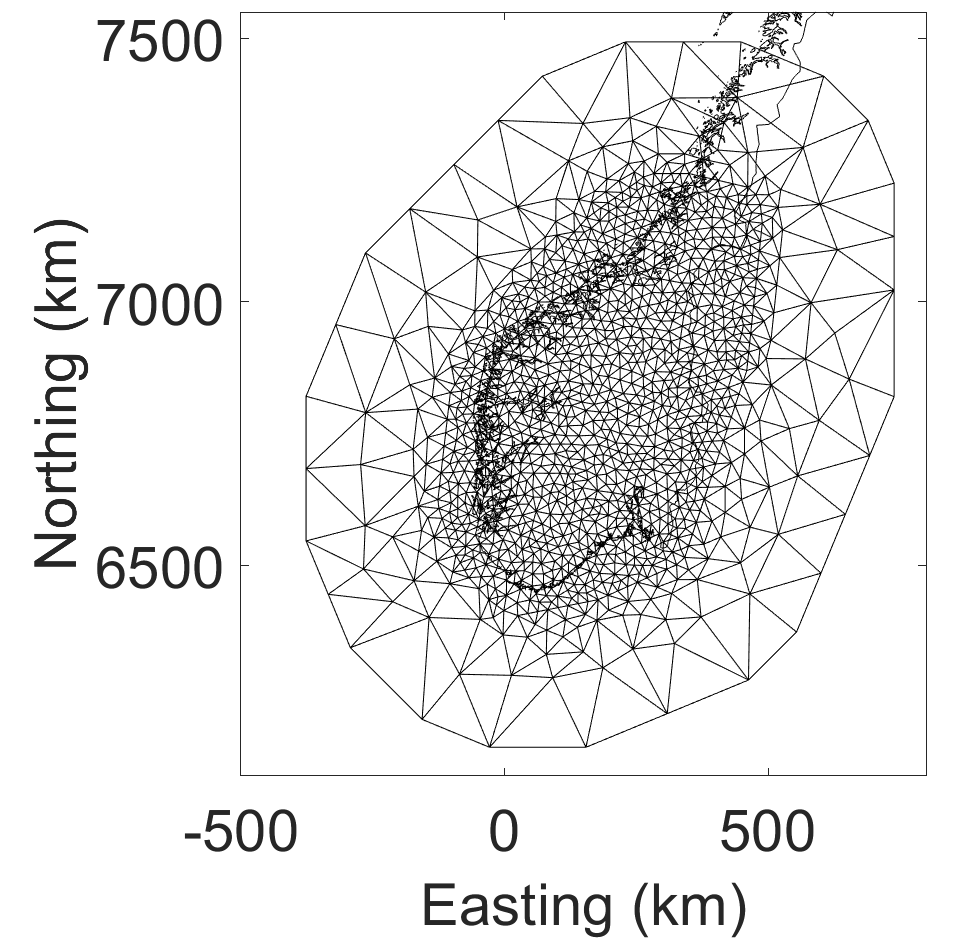}
    \caption{Mesh of domain}
    \label{fig:mesh_domain}
  \end{subfigure}
  \caption{The observed precipitation $\bm{y}$, the elevation $h$, and the mesh of the domain for the precipitation data set.}
  \label{Precip:figs}
\end{figure}

This analysis focuses on the stationary anisotropic setting where $u| \bm{\theta }$ is a solution to \eqref{SPDE2D} with spatially constant parameters. By incorporating $\bm{\beta }:=(\beta _0,\beta _1)$ into $\bm{u}$, our linear model \eqref{height model} fits into the framework of \Cref{simulation section}, where now
\begin{align*}
  \bm{A}_{\bm{\beta }} := (\bm{1}_m, \bm{h},\bm{A}), \quad \bm{u}_{\bm{\beta } } := (\bm{\beta }, \bm{u})
\end{align*}
take the place of $\bm{A}, \bm{u}$ in \eqref{obs}. We will consider
\begin{align*}
  \bm{\beta } \sim \Nn (\bm{0}, \bm{Q}_{\bm{\beta }}^{-1}), \quad \bm{Q}_{\bm{\beta }} = \tau_{\bm{\beta }} \bm{I}_{2}
\end{align*}
independent from $\bm{u}, \bm{\theta }$ and where we set the precision parameter to be $\tau_{\bm{\beta }}= 10^{-4}$.
\subsection{Maximum a posteriori estimates}\label{MAP estimates}
We compare the model in \eqref{height model} under anisotropic PC and EG priors (see \Cref{option1} and \Cref{option2}) on ($\kappa, \bm{v}$) and isotropic PC priors where $\bm{v}$ is set to $\bm{0}$. To do so, we must first derive an isotropic PC prior on $\kappa $ using the distance metric \eqref{pseudo metric} restricted to the case where $\bm{v}=\bm{0}$. Using equation~\eqref{pseudo metric calculation}, we obtain the following result.
\begin{align*}
 D_2 (M_{\kappa , \bm{0}},M_{\bm{0}})= \frac{1}{\sqrt{12 \pi } }\kappa.
\end{align*}
So, by the principle of exponential penalization, $\kappa \sim \rmm{Exp}(\lambda_{\rmm{iso}} )$.
Following \citet{Simpson2014PenalisingMC}, we choose the hyperparameters in the three priors so that\begin{equation*}
  \mathbb{P}[\rho<10]=0.05, \quad \mathbb{P}[\sigma_u>3]=0.05,\quad \mathbb{P}\left[\sigma_{\bm{\varepsilon}}>3\right]=0.05.
\end{equation*}
For the anisotropy parameter $\bm{v}$ we choose the hyperparameter $\lambda_{\bm{v}}$ so that $\mathbb{P}[a>10]=0.05$. This choice of hyperparameters imposes that with probability $0.05$, the field has a correlation range that is $10$ times larger in any given direction.

In \Cref{Prior comparison}, we plot the density of the priors on $\rho$. Since the marginal density of $\rho$ is the same for the isotropic PC and EG priors, we do not plot it. We also plot the density of the PC and (anisotropic) EG priors on $r:=\norm{\bm{v}}$.
\begin{figure}[H]
  \centering
  \includegraphics[width=0.9\textwidth]{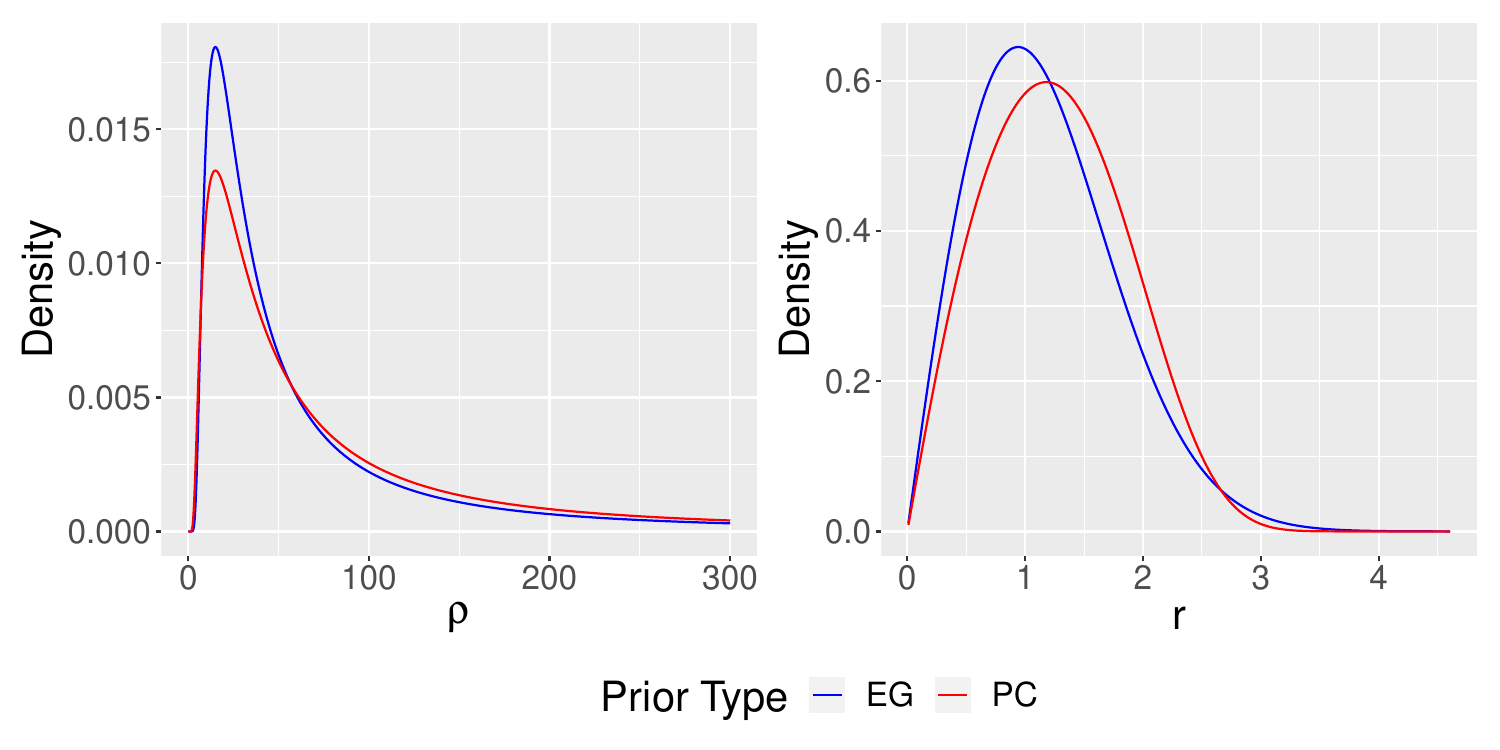}
  \caption{Marginal densities on $\rho$ and $r=\norm{\bm{v}}$ of PC and EG priors}
  \label{Prior comparison}
\end{figure}
The decay of the marginal PC prior on $\kappa $ and of the EG prior are both exponential. The decay rate of the marginal PC prior on $v$ is $\exp(- c_1 \exp(\norm{\bm{v}} ))$, whereas for the EG prior it decays slower, as $\exp(-c_2 \norm{\bm{v}}^2)$ for some constants $c_1,c_2$.

The MAP estimates and symmetric $95\%$ credible intervals for the anisotropic PC, anisotropic EG, and isotropic PC models are shown in \Cref{tab:map_and_ci}.

\begin{table}[H]
  \begin{tabular}{lcccccc}
    \toprule
    \multirow{2}{*}{Parameter}    & \multicolumn{3}{c}{MAP Estimates} & \multicolumn{3}{c}{95\% Credible Intervals}                                 \\
    \cmidrule(lr){2-4} \cmidrule(lr){5-7}
                    & ANISO PC             & ANISO EG                  & ISO PC
                    & ANISO PC             & ANISO EG                  & ISO PC
    \\
    \midrule
    $\wh{\rho}$           & 201                & 201                     & 193  & (132, 310)    & (132, 311)    & (128, 290)   \\
    $\wh{v}_1$            & $-0.45$              & $-0.43$                   & -   & $(-0.81, -0.11)$ & $(-0.78, -0.10)$ & -       \\
    $\wh{v}_2$            & $0.04$              & $0.03$                   & -   & $(-0.28, 0.35)$ & $(-0.28, 0.34)$ & -       \\
    $\wh{\sigma}_u$         & 0.63               & 0.64                    & 0.65  & $(0.46, 0.88)$  & $(0.46, 0.89)$  & $(0.47, 0.90)$ \\
    $\wh{\sigma}_{\bm{\varepsilon}}$ & 0.14               & 0.14                    & 0.13  & $(0.11, 0.18)$  & $(0.11, 0.18)$  & $(0.10, 0.16)$ \\
    \bottomrule
  \end{tabular}
  \caption{\label{tab:map_and_ci}MAP estimates and 95\% credible intervals for the parameters of the three different precipitation models: the anisotropic model with PC priors (ANISO PC), the anisotropic model with EG priors (ANISO EG), and the isotropic model with PC priors (ISO PC).}
\end{table}

The credible intervals for $v_1$ in the anisotropic models do not contain $0$, indicating that with high confidence, anisotropy is present in the precipitation field.

The half angle vector of the MAP for the anisotropic model with PC priors $\wh{\bm{v}}$ is $\tl{\wh{\bm{v}}}=(0.02,0.48)$, and indicates that the precipitation is $a=1.64$ times more correlated in the North-South direction than in the East-West direction.
In \Cref{fig:predPrecPC}, we plot the posterior prediction, latent field, and the covariance function of the anisotropic field $\bm{u}$ with parameters $\wh{\kappa},\wh{\bm{v}}$.
\begin{figure}[H]
  \centering
  \begin{subfigure}[b]{0.28\linewidth}
    \centering
    \includegraphics[width=\linewidth]{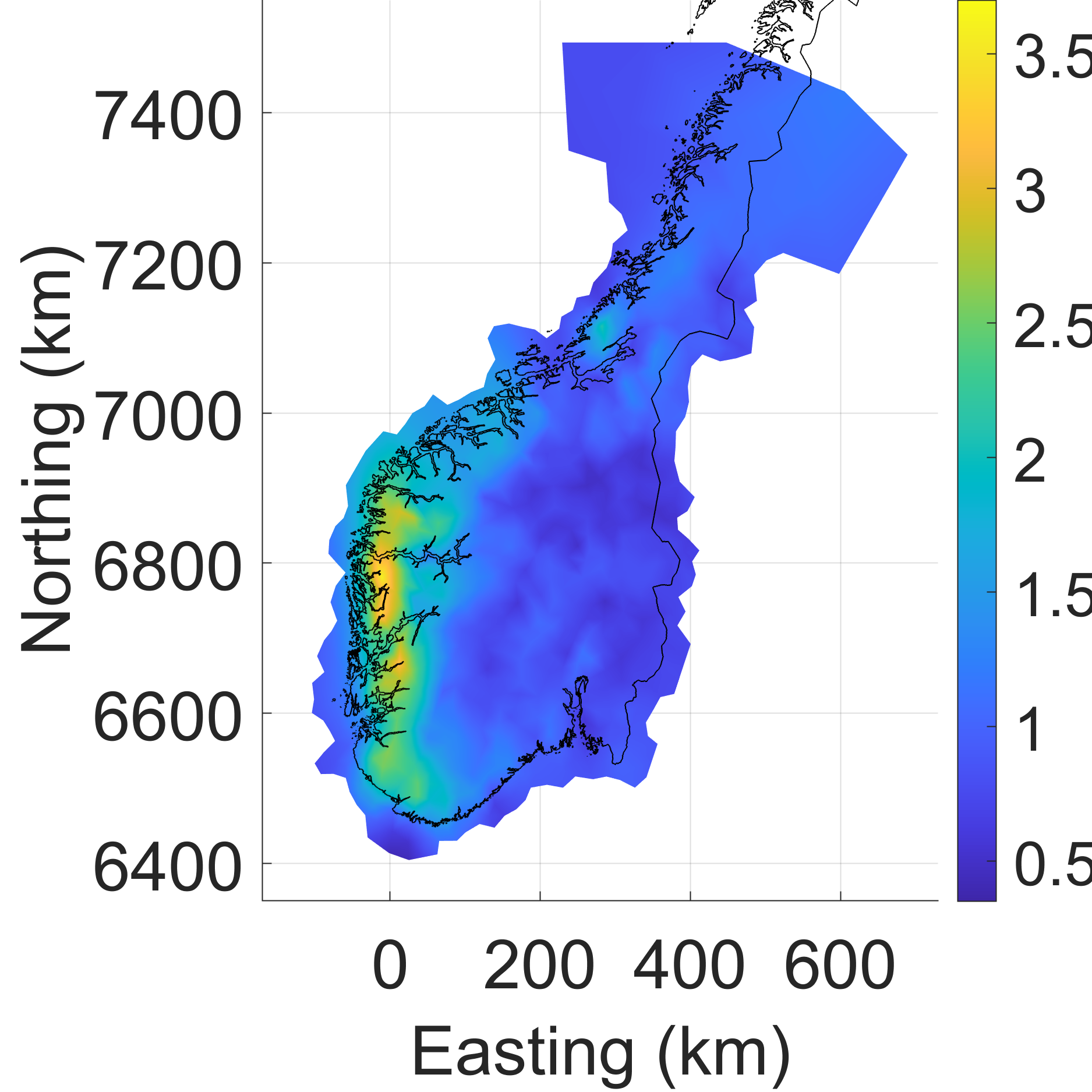}
    \caption{$\E[\beta_0+ \beta_1 \bm{h}+ \bm{u}| \bm{y}]$}
    \label{fig:pred_prec_pc}
  \end{subfigure}
  \hfill
  \begin{subfigure}[b]{0.28\linewidth}
    \centering
    \includegraphics[width=\linewidth]{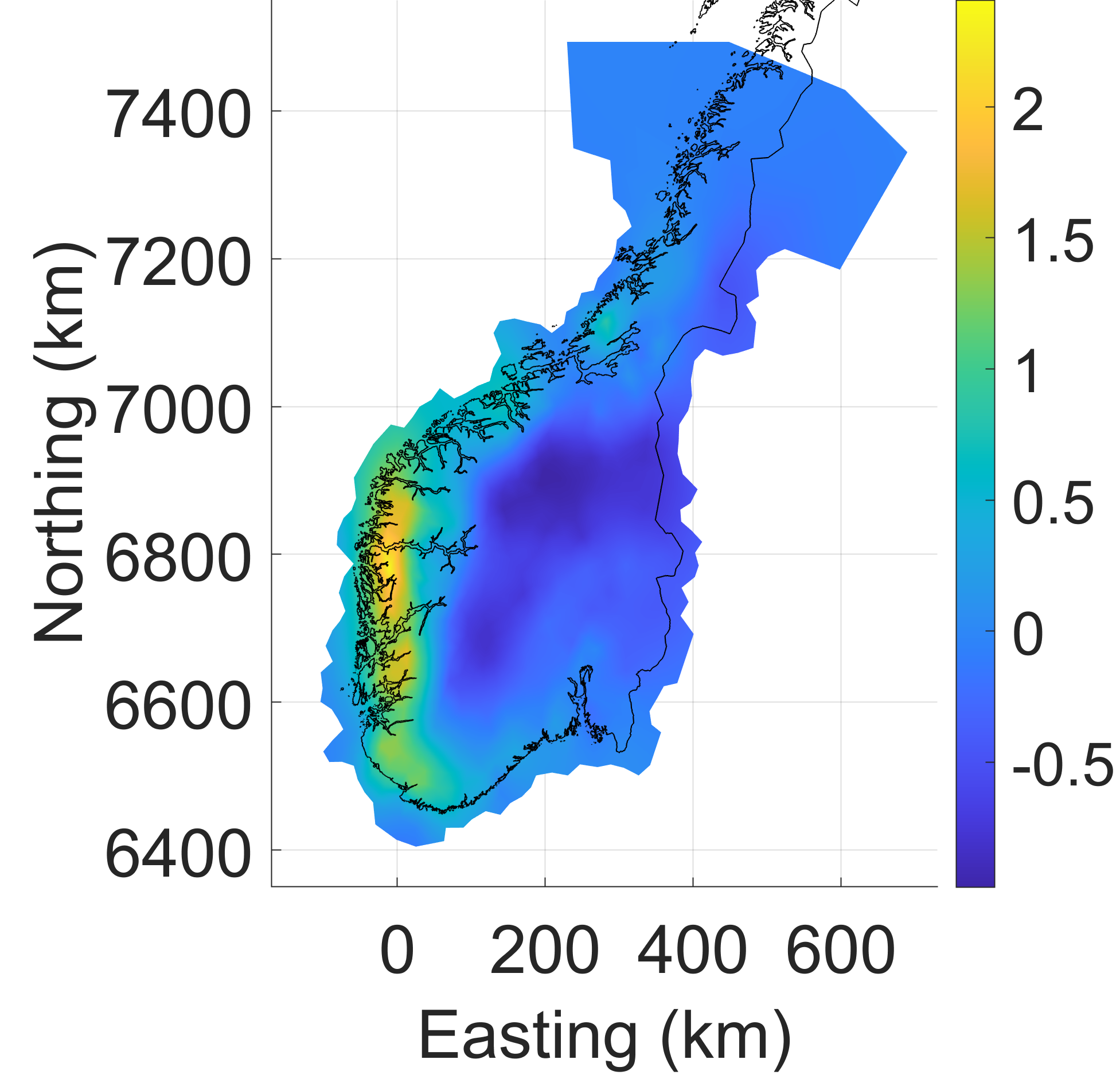}
    \caption{$\E\qb{\bm{u} | \bm{y}}$}
    \label{fig:post_mean_u_pc}
  \end{subfigure}
  \hfill
  \begin{subfigure}[b]{0.40\linewidth}
    \centering
    \includegraphics[width=\linewidth]{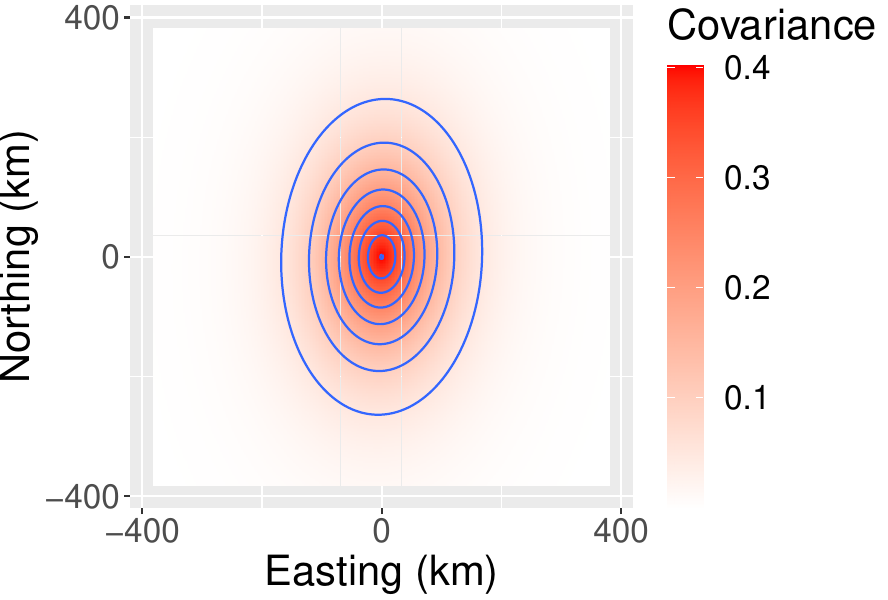}
    \caption{$r(\bm{x})=\mathbb{E}\qty[u_{\wh{\kappa},\wh{\bm{v}}}(\bm{x}+\bm{z})u_{\wh{\kappa},\wh{\bm{v}}}(\bm{z})]$}
    \label{fig:covariance_precip}
  \end{subfigure}
  \caption{Posterior estimates for the precipitation data set under the anisotropic PC prior. (a) Predictive mean of precipitation. (b) The posterior mean of the latent spatial field $\bm{u}$. (c) The anisotropic covariance function for the MAP estimate, showing a stronger correlation (longer range) in the North-South direction.}
  \label{fig:predPrecPC}
\end{figure}

\subsection{Model performance}

To assess the performance of the models, we calculate various scores for each model. We recall that, given a (predictive) distribution $P$ and an (observed) point ${y}$, a score is a function $S({P},{y})$ that measures how closely the prediction ${P}$ matches the observation ${y}$. If ${y}$ follows the distribution ${Q}$, then the expected score is
\begin{align*}
 S({P},{Q}) := \E_{{Q}}[S({P},{y})].
\end{align*}
The score should be minimized when the true distribution ${Q}$ matches the predictive distribution ${P}$. In this case, the score is said to be \emph{proper}. It is \emph{strictly proper} if it is minimized if and only if ${P}={Q}$. We consider the \emph{squared error} ($\mathrm{SE}$), \emph{continuous ranked probability score} ($\rmm{CRPS}$, and \emph{Dawid-Sebastiani score} ($\rmm{DSS}$). These are defined in Appendix \Cref{app:Score definitions}.

Given a sample sample $\bm{y}=(y_1,...,y_n)$, a vector of predictive distributions $\bm{F}=(F_1,...,F_n)$ and a score $S$ , the \emph{mean score} is defined as
\begin{align*}
  \overline{S}(\bm{F}, \bm{y}):= \frac{1}{n}\sum_{i=1}^n S(F_i,y_i).
\end{align*}
By the linearity of the expectations, if $S$ is a (strictly) proper scoring rule, then so is $\overline{S}$.
\subsection{Score results}\label{results precip}
In this section, we calculate the mean scores
\begin{align*}
  \rmm{RMSE}:=\qty(\overline{SE}(\bm{F},\bm{y}))^\frac{1}{2}, \quad \overline{\rmm{CRPS}}(\bm{F},\bm{y}), \quad \overline{\rmm{DSS}}(\bm{F},\bm{y}),
\end{align*}
where $F_i:= \pi(y_i| \bm{y}_{-i})$ is the LOO predictive distribution under each of the three models in Subsection \ref{MAP estimates}, and $\bm{y}$ is the observation.

An initial calculation showed that the scores of the previous sections are very similar under all three priors. From this, we conclude that the model is very informative, and the influence of the priors is limited. As a result, to better compare the priors, we uniformly sub-sample $\bm{y}$ and observe only $\bm{y}' \in \R^{n_y'}$ with $n_y'< n_y=233$. We then calculate the scores from the previous sections. Due to the extra variability introduced by sub-sampling the observations, we repeat this process $10$ times. The resulting mean scores are shown in \Cref{fig: Simulation_images-LO_scores-pdf}. The anisotropic models perform better with less data, with the PC prior performing slightly better than the EG prior. Whereas with more data, the results are almost equivalent.

\begin{figure}[H]
  \centering
  \includegraphics[width=\textwidth]{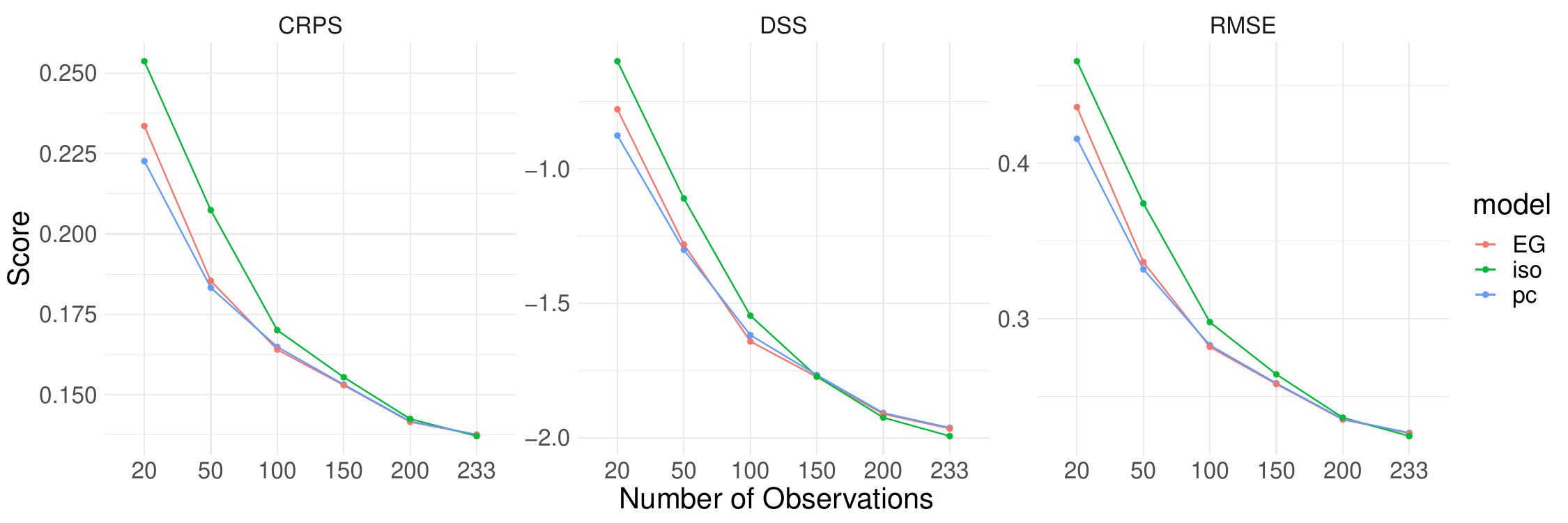}
  \caption{Leave-one-out cross-validation scores as a function of the number of observed data points. Lower scores are better. The anisotropic models (ANISO PC, ANISO EG) show a distinct performance advantage over the isotropic model (ISO PC) in low-data regimes.}
  \label{fig: Simulation_images-LO_scores-pdf}
\end{figure}
Based on the trend of the above results, it is reasonable to check whether the isotropic model outperforms the other two models with a larger number of observations. To test this hypothesis, we conducted a simulation study in which a larger number of observations was generated, up to $n_y=1000$. The results are included in \Cref{simulation study precipitation appendix} in the supplementary material and show that this is not the case. The isotropic model does not outperform the anisotropic models with a larger amount of observations.

\section{Discussion}\label{Discussion section}
In this study, we constructed a smooth, \revisiontwo{invertible}, and geometrically interpretable parameterization for a 2D anisotropic spatial Gaussian model. We developed penalized complexity (PC) priors for these parameters, defining model complexity as a Sobolev seminorm of the correlation function. This distance was calculated in a closed form using the spectral density.

Our performance comparison of the PC prior against other priors demonstrated its effectiveness in penalizing model complexity. We found that priors designed to match the quantiles of the PC prior yielded similar results. Both these priors significantly outperformed the other non-informative priors considered, highlighting the necessity of incorporating penalization in prior information to prevent overfitting.

Applying the anisotropic model to a real dataset of precipitation in Southern Norway, we observed that the anisotropic model outperformed the isotropic model when the number of observations was small, with the PC prior slightly outperforming the other choices. However, with a larger number of observations, the isotropic model performed similarly to the anisotropic models. These results indicate that the anisotropic model is more informative when the data is scarce, but as the number of observations increases, the isotropic model becomes more competitive. The results also suggest that a non-stationary model could be better suited to capture spatially varying patterns in the data.

In conclusion, we advocate for the use of informative priors for the anisotropy parameters in spatial Gaussian models. The PC prior is highly effective, but other priors designed to match desired quantiles can also be effective and are simpler to construct.

\revisiontwo{   Finally, we note that while the Matérn class allows for flexible local behavior, it restricts the field to exponential tail decay. Recent work by \citet{ma2023beyond} introduces the Confluent Hypergeometric (CH) class, which generalizes the Matérn family to allow for polynomial tail decay. While the Matérn class is a limiting case of the CH class, the PC framework presented here relies on the existence of a standard spectral density, which is not guaranteed for the CH class in regimes of strong long-range dependence ($\alpha \le d/2$) where the spectrum exhibits a singularity at the origin. Extending PC priors to such heavy-tailed classes remains a promising avenue for future research, likely requiring a complexity distance adapted to handle such singularities, for instance, via weighted spectral norms.}

Looking forward, we aim to extend this study to the non-stationary setting where the model parameters vary spatially. This presents a challenge as there is no agreed-upon definition of correlation function or spectral density, necessitating a different definition of complexity. We are also interested in extending the parameterization to higher dimensions. The current construction relies on a ``half-angle'' parameterization of the anisotropy vector, which is not easily extendable to higher dimensions. Finally, we are interested in extending the construction to different orders of regularity.


\newpage
\section*{Supplementary Material}
\addcontentsline{toc}{section}{Supplementary Material} %
\setcounter{section}{0}
\renewcommand{\thesection}{S\arabic{section}}
\renewcommand{\thesubsection}{S\arabic{section}.\arabic{subsection}}

\section{Model derivation}\label{app:model derivation}
Here we informally motivate the choice of the SPDE \eqref{SPDE2D}. By definition, a field $u$ on $\R^d$ is \emph{stationary} if its covariance $K$ depends only on the relative position between two points. That is, for some function $r$ called the \emph{covariance function}
\begin{equation}\label{covariance function def}
 K(\bm{x},\bm{y}):={\rmm{Cov}[u(\bm{x}),u(\bm{y})]}=r(\bm{y}-\bm{x}), \quad\forall \bm{x},\bm{y}\in\R^d.
\end{equation}
Equation \eqref{covariance function def}, which is an extension of \eqref{stationary}, requires that the Euclidean translation $\bm{y}-\bm{x}$ captures all relevant information related to the covariance of the field between any two spatial locations $\bm{x}$ and $\bm{y}$. However, Euclidean geometry may not fit with the underlying properties of the field. For example, suppose that our field describes the geological properties of some homogeneous terrain. Even if the field was initially stationary, stationarity is lost if the terrain underwent a geological deformation $\psi$. Instead, the relationship
\begin{equation*}
 K(\bm{x},\bm{y})=r(\psi^{-1}(\bm{x})-\psi^{-1}(\bm{y}))
\end{equation*}
would be more appropriate. This
is the \emph{deformation method} \citep{sampson1992nonparametric} and is, for instance, sometimes used to connect different layers in deep Gaussian processes. In deep Gaussian processes, this transformation is then random and not diffeomorphic \citep{dunlop2018deep}.

The deformation method offers a straightforward approach to constructing non-stationary fields from stationary ones. Furthermore, it can be used in a layered approach, where a stationary field is first built through some appropriate method and then deformed into a non-stationary field.

As discussed previously, SPDEs provide a convenient framework for constructing stationary fields, for example, as seen in \eqref{SPDE2011}. Let us see what happens when we deform such a field. To this aim, consider an (unknown) $d$-dimensional manifold $\tl{\mathcal{D}}$ and our target manifold $\mathcal{D}$ obtained through a diffeomorphic transformation
\begin{equation*}
  \psi:\tl{\mathcal{D}}\to\mathcal{D}.
\end{equation*}
Further consider the solutions $\tl{u}$ to \begin{equation}\label{stationary model}
 (1-\Delta)\tl{u}={\sqrt{4\pi}}\dot{\tl{\Ww}},\quad \text{on }\tl{\mathcal{D}}.
\end{equation}
Here, either $\tl{\Dd}=\R^d$, in which case, to obtain uniqueness, we impose a stationarity condition. Otherwise, $\tl{\Dd }$ is a \revisiontwo{smooth bounded domain in} $\mathbb{R}^d$, in which case Neumann conditions or Dirichlet conditions are imposed on the boundary. \revisiontwo{Here and in all that follows $d$ is restricted to $d=1,2,3$ so that the solution has positive regularity.} Then, a change of variables yields that $u:= \sigma_u \cdot\tl{u} \circ \psi^{-1}$ verifies the \emph{non-stationary} SPDE
Then, a change of variables yields that $u:= \sigma_u \cdot\tl{u} \circ \psi^{-1}$ verifies the \emph{non-stationary} SPDE
\begin{equation}\label{TransformedSPDE}
  \frac{1}{\norm{\bm{\Psi}}}\left[1-\norm{\bm{\Psi}} \nabla \cdot \frac{\bm{\Psi} \bm{\Psi}^T}{\norm{\bm{\Psi}}} \nabla\right] \frac{u}{\sigma_u}=\frac{\sqrt{4\pi}}{\norm{\bm{\Psi}}^{1 / 2}} \dot{\mathcal{W}} \quad \text {on } \Dd,
\end{equation}
where $\bm{\Psi}$ is the Jacobian of $\psi$, and we denote its determinant by $\norm{\bm{\Psi}} $ \cite[Section 3.4]{Lindgren2011AnEL}.

Let us write $\gamma$ for the geometric mean of the eigenvalues of $\bm{\Psi}$. Then, $\norm{\bm{\Psi}}=\gamma^d$ and we obtain from \eqref{TransformedSPDE} that
\begin{equation}\label{TransformedSPDE2}
  \frac{1}{\gamma^d}\left[1-\gamma^d \nabla \cdot \gamma^{2-d}\bm{H}\nabla\right] \frac{u(\bm{x})}{\sigma_u}=\frac{\sqrt{4\pi}}{\gamma^{d / 2}} \dot{\Ww}\quad \text {on } \Dd,
\end{equation}
where we defined $\bm{H}:=\gamma^{-2}\bm{\Psi}\bm{\Psi}^T$. Note that if we write $\kappa:=\gamma^{-1}$ and take the dimension $d$ to be $2$, we recover our model in \eqref{SPDE2D}. For general $d$, it follows from the definition of $\bm{H}$ that
\begin{enumerate}
  \item $\bm{H}$ is symmetric.\label{c1}
  \item $\bm{H}$ has determinant $1$.
  \item $\bm{H}$ is positive definite.\label{c3}
\end{enumerate}
From now on, we will impose these three assumptions on $\bm{H}$. Furthermore, we will restrict ourselves to the stationary case by imposing that $\bm{H}$ is constant in space, or equivalently, we impose that $\psi$ is a linear deformation. That is, $\psi(\bm{x})=\bm{\Psi}\bm{x}$.
By construction, the solution field $u$ is thus \emph{geometrically anisotropic} \citep[ Section~4.1]{cressie2015statistics}. That is, if we replace the Euclidean metric with the \emph{deformed metric} $\norm{\bm{x}}_{\bm{\Psi}^{-1}}:= \norm{\bm{\Psi}^{-1}\bm{x}}$, we obtain analogously to \eqref{stationary} that, for some function $r:\R^d \to \R$ called the \emph{covariance function} of $u$,
\begin{equation*}
 K(\bm{x},\bm{y}):=\rmm{Cov}[u(\bm{x}),u(\bm{y})] =r\left(\norm{\bm{y}-\bm{x}}_{\bm{\Psi}^{-1}}\right), \quad\forall \bm{x},\bm{y}\in\R^d,
\end{equation*}
In the case $\tl{\Dd} =\R^d$, the marginal variance of $\tl{u}$ is $\mathbb{E}[\tl{u} (\tl{ \bm{x}})^2]=1$ for all $\tl{ \bm{x}} \in \tl{\Dd } $ \citep[Section~2.1]{Lindgren2011AnEL} (we recall that the solution to SPDEs of the form \eqref{SPDE2D} have mean $0$). As a result, the marginal variance of $u$ is $\mathbb{E}[u(\bm{x})^2]=\sigma_u^2$ for all $\bm{x} \in \Dd $.

In the case where $\Dd$ is a bounded domain, there is a boundary effect that affects the marginal variance of $u$. However, at a distance larger than twice the correlation length from the boundary, the bounded domain model is almost indistinguishable from the unbounded domain model. As a result, the boundary effect can be made negligible by embedding the region of interest in a sufficiently large domain \cite{khristenko2019analysis}.

\section{Proof of \Cref{parameterizATION THEOREM}}
\begin{proof}
    By definition, $\bm{H}_{\bm{v}}$ is symmetric with determinant $1$. Furthermore, it is positive definite as its eigenvalues are positive. Given a positive definite symmetric matrix $\bm{A}$ with determinant $1$, $\bm{A}$ has an eigensystem of the form
    \begin{align*}
        \set{(\bm{w},\lambda ),(\bm{w}^\perp,\lambda ^{-1})},
    \end{align*}
    where $\lambda\geq 1$.
    By normalizing and taking $-w$ if necessary we may suppose that $\bm{w}=\ex{\alpha i}$ with $\alpha\in [0,\pi)$. Then, we obtain $\bm{A} = \bm{H}_{v}$ for
    \begin{equation*}
        \bm{v}=\log(\lambda)\ex{2\alpha i}.
    \end{equation*}
    This proves that the parameterization is surjective. Suppose now that $\bm{H}_{v}=\bm{H}_{\bm{v}'}$ then their eigenvalues and eigenvectors must be equal so
    \begin{align*}
        \ex{\norm{\bm{v}} }=\ex{\norm{\bm{v}'} }, \quad  \tl{v}=a\tl{\bm{v}'}.
    \end{align*}
    For some $a \in  \R$. From the first condition, we deduce that $\norm{\bm{v}}=\norm{\bm{v}'}$. Taking absolute values in the second condition thus gives $\abs{a}=1$. By construction $\tl{v}\neq -\tl{\bm{v}'}$ for all $v,\bm{v}'$ so necessarily  $\bm{v}=\bm{v}'$. This proves that the parameterization is invertible.

    Next, to derive the second equality in \eqref{cool formula}, we use the half-angle formula
    \begin{equation*}
        \cos(\frac{\alpha}{2})= \sqrt{\frac{1+\cos(\alpha ) }{2} }\mathrm{sign}(\sin(\alpha))  , \quad \sin(\frac{\alpha}{2} )= \sqrt{\frac{1-\cos(\alpha) }{2} }.
    \end{equation*}
    This gives that
    \begin{equation*}
        \tl{\bm{v}}\tl{\bm{v}}^T= \frac{\norm{\bm{v}}^2}{2}\bm{I}+ \frac{\norm{\bm{v}}}{2}\begin{bmatrix}
            v_1 & v_2
            \\
            v_2 & -v_1
        \end{bmatrix}, \quad    \tl{\bm{v}}_\perp\tl{\bm{v}}^T_\perp= \frac{\norm{\bm{v}}^2}{2}\bm{I}-\frac{\norm{\bm{v}}}{2}\begin{bmatrix}
            v_1 & v_2
            \\
            v_2 & -v_1
        \end{bmatrix}    .
    \end{equation*}
    The proof follows immediately by using the definition of $\bm{H}_{\bm{v}}$ (first line of \eqref{cool formula}).
\end{proof}
\section{Selection of a distance}\label{Sobolev distance section}
In this section, we discuss the choice of the Sobolev distance in \Cref{distance metric}. We first examine other possible choices, such as the Kullback-Leibler divergence, the $L^2$ distance, the Wasserstein distance, and the Hellinger distance, and show that they are unsuitable for our purposes. We then derive the exact form of the Sobolev distance in \ref{pseudo metric calculation} and show a possible alternative definition in \ref{alternative distance}.

Without pretending full generality, we will use the following notation. Firstly, \revisiontwo{we denote the torus on $[\b{0}, \b{T}]^d$ by  $\mathbb{T}^d:=\R^d /\b{T} \Z^d$ where $\b{T} \in \R_+^d$ is the side-length. In particular, we assume periodic boundary conditions on $\mathbb{T}^d$ ,} denote the \emph{orthonormal Fourier basis} on \revisiontwo{$\T^d$} by
\begin{align}\label{on Fourier basis}
    e_{\bm{k}}(\bm{x}):= \rmm{vol}(\revisiontwo{\T^d} )^{-1/2}\exp(i\bm{\omega}_{\bm{k}}\cdot {\bm{x}}), \quad \bm{\omega_k}:= \nicefrac{2 \pi  \bm{k}}{\revisiontwo{\b{T}}} \quad \bm{k} \in \Z^d,
\end{align}
where we used the elementwise division $[\bm{k}/\revisiontwo{\b{T}}]_i:= k_i/\revisiontwo{T_i}$. Given a stationary field $u$ on $\revisiontwo{\T^d} $, we denote its \emph{covariance function} $r$ and \emph{spectral density}, if they exist, by
\begin{align*}
    r(\bm{x}):= \rmm{Cov}[u(\bm{x}),u(\bm{0})], \quad S(\bm{k}):=\int_{\T^d}  r(\bm{x})\overline{e_{\bm{k}}(\bm{x})} \d \bm{x}.
\end{align*}
Additionally, we will denote the \emph{covariance operator} of $u$ by \revisiontwo{$\Kk$}. For stationary $u \in L^2(\revisiontwo{\T^d})$ this is given by
\begin{align*}
    \br{\revisiontwo{\Kk} f,g}_{L^2(\revisiontwo{\T^d} )}:= \int_{\revisiontwo{\T^d}}\int_{\revisiontwo{\T^d} } r(\bm{y-x})f(\bm{x})g(\bm{y}) \d \bm{x}\d \bm{y}, \quad f,g \in L^2(\revisiontwo{\T^d}).
\end{align*}
For the calculations, the following lemma will be very useful.
\begin{lemma}\label{diagonal lemma}
    Let $u$ be stationary and periodic on the torus $\mathbb{T}^d$ with continuous covariance $r$ and spectral density $S$. Then, $\Kk$ has eigensystem $\set{\lambda_\b{k}^2, e_\b{k}}_{\b{k} \in \Z^d} $ given by
    \begin{align*}
        \lambda_{\bm{k}}^2= S(\bm{k}),\quad    e_{\bm{k}}(\x) =e^{i \om_{\bm{k}} \cdot \x}, \quad \om_{\bm{k}}:= \nicefrac{2\pi\bm{k}}{\bm{T}}.
    \end{align*}
\end{lemma}
\begin{proof}
    Let $r(\x)$ be the covariance function of $u$. By Bochner's theorem on the torus \cite[p.19]{rudin2017fourier} and Fubini, \revisiontwo{given $f,g \in L^2(\revisiontwo{\T^d})$} 
    \begin{align*}
        \br{\Kk f, g}_{\revisiontwo{L^2(\revisiontwo{\T^d})}}  & = \int_{\revisiontwo{\revisiontwo{\T^d}}} r(\x-\y)f(\y)\overline{g(\x)} \d\x \d\y                                                                                                                                               \\
                    & = \int_{\revisiontwo{\T^d}} \int_{\revisiontwo{\T^d}}  \qt{\sum_{\b{k} \in \Z^d}  S({\b{k}}) e^{ i \om_{\bm{k}} (\x-\y)} } f(\y)\overline{g(\x)}    \d\x \d\y=\sum_{\b{k} \in \Z^d}   S({\b{k}})\hat{f}(\om_{\bm{k}})\overline{\hat{g}(\om_{\bm{k}})}.
    \end{align*}
    This shows that, since \revisiontwo{by orthonormality of the Fourier basis $\wh{e_\b{j}}(\om_{\bm{k}})= \br{e_{\b{j}},e_{\b{k}}}= \delta_{jk}$}
    \begin{align*}
        \br{\Kk e_\b{j}, e_\b{k}} = \sum_{\b{l}\in \Z^d} S({\b{l}}) \wh{e_\b{j}}(\om_{\bm{l}})\wh{e_\b{k}}(\om_{\bm{l}}) \ =\sum_{\b{l}\in \Z^d} S({\b{l}})\delta_{\b{jl}} \delta_{\b{k l}} = S(\b{j}) \delta_{\b{jk}}  .
    \end{align*}
    This concludes the proof.
\end{proof}
We will also repeatedly use the expression of the spectral density of the solution $u$  to \eqref{SPDE2D}. To calculate this, note that the covariance operator of $u$ is given by $\revisiontwo{\Kk}= \Ll ^{-2}$ where
\begin{align*}
    \Ll := \frac{\kappa^2- \nabla \cdot \bm{H}  \nabla}{ \sqrt{4 \pi} \kappa \sigma _u}.
\end{align*}
Using \revisiontwo{the periodicity} and that $\nabla e_{\bm{k}}= i \bm{\omega}_{\bm{k}} e_{\bm{k}}$ we obtain that $\Ll $ diagonalizes in the Fourier basis with
\begin{align*}
    [\Ll ]_{\bm{jk}}:= \br{\Ll e_{\bm{j}},e_{\bm{k}}}_{L^2(\revisiontwo{\T^d} )} = \frac{\kappa^2+ \bm{\omega}_{\bm{k}} \cdot \bm{H} \bm{\omega}_{\bm{k}}}{\sqrt{4 \pi} \kappa \sigma _u}\delta_{\bm{j}\bm{k}}, \quad [\Ll ^{-1}]_{\bm{jk}}= \frac{\sqrt{4 \pi} \kappa \sigma _u}{\kappa^2+ \bm{\omega}_{\bm{k}} \cdot \bm{H} \bm{\omega}_{\bm{k}}}\delta_{\bm{j}\bm{k}}.
\end{align*}
As a result, by \Cref{diagonal lemma}, the spectral density of $u$ is given by
\begin{align}\label{stationary density bounded}
    S(\bm{k})= \frac{4 \pi \kappa ^2 \sigma _u^2}{(\kappa^2+ \bm{\omega}_{\bm{k}}^T \bm{H} \bm{\omega}_{\bm{k}})^2}, \quad \bm{k} \in \Z^d.
\end{align}
If now the domain is $\R^d$, the defining property of the spectral density is that \citep[Volume IV page 264]{gel2014generalized}
\begin{align*}
    \br{\Kk f,g}_{L^2(\R^d )}= \int_{\R^d} S(\bm{\xi })\wh{f}(\bm{\xi })\overline{\wh{g}(\bm{\xi })} \d \bm{\xi }, \quad\forall f,g \in L^2(\R^d).
\end{align*}
Since the Fourier transform is an isometry on $L^2(\R^d)$,
\begin{align*}
    \br{\Kk f,g}_{L^2(\R^d )} & = \br{ \Ll ^{-1} f, \Ll ^{-1}g}_{L^2(\R^d )}=\br{\wh{\Ll ^{-1}f},\wh{\Ll^{-1}g}}_{L^2(\R^d )} \\&=\int_{\R^d} \frac{4\pi\kappa ^2 \sigma _u^2}{(\kappa^2+ 4\pi^2\bm{\xi}^T \bm{H} \bm{\xi})^2} \wh{f}(\bm{\xi })\overline{\wh{g}(\bm{\xi })} \d \bm{\xi }.
\end{align*}
So the stationary solution to \eqref{SPDE2D} on $\R^d$  has spectral density
\begin{align}\label{spectral density R^d}
    S(\bm{\xi })= \frac{4\pi\kappa ^2 \sigma _u^2}{(\kappa^2+ 4\pi^2\bm{\xi}^T \bm{H} \bm{\xi})^2}.
\end{align}

\subsection{Kullback-Leibler divergence}\label{KLD section}
As previously discussed, the KLD is unsuitable as a notion of distance between models as it is infinite for every possible parameter value. As a result, a different method is necessary to measure the distance between models. This section discusses some of the options considered and how \Cref{metric def} was eventually chosen. We work on the torus in this and the next subsection as the Hilbert-Schmidt norm of a stationary field on $\R^d$.

We begin by showing that the KLD is infinite for every pair of parameter sets that do not fulfil certain mutual relationships. To do so, we first give a sufficient condition for two stationary measures to be mutually singular. This is the converse direction of  \citet[Theorem A.1]{stein2004equivalence}. See also \citet[Proposition 3]{dunlop2017hierarchical}.
\begin{lemma}\label{KLD infinite condition}
    Let $u_A,u_B$ be two stationary, \revisiontwo{and in particular periodic}, Gaussian fields in $L^2(\revisiontwo{\T^d} )$ with spectral densities $S_A,S_B$. Suppose that
    \begin{align*}
        \sum_{ \bm{k} \in \Z^d} \qty(\frac{S_A(\bm{k})}{S_B(\bm{k})}-1)^2 \revisiontwo{=} \infty.
    \end{align*}
    Then, the Gaussian measures defined by $u_A,u_B$ are mutually singular.
\end{lemma}
\begin{proof}
    Denote the covariance operators of $u_A,u_B$ by $K_A,K_B$ respectively and let $I$ be the identity on $L^2(\revisiontwo{\T^d} )$. By the Feldman-Hájek theorem, a necessary condition for $u_A,u_B$ to not be mutually singular is that the following operator is Hilbert-Schmidt \revisiontwo{\citep[Theorem 2.25]{da2014stochastic}}, 
    \begin{align*}
        T:=\qty(K_B^{-1/2}K_A^{1/2})\qty(K_B^{-1/2}K_A^{1/2})^*-I,
    \end{align*}
    is a Hilbert-Schmidt operator where $K_A,K_B$ are the covariance operators of $u_A,u_B$ and $I$ is the identity operator on $L^2(\revisiontwo{\T^d} )$. By \Cref{diagonal lemma}, $K_A, K_B$ both diagonalize in the orthonormal Fourier basis with eigenvalues given respectively by $S_A(\bm{k})$ and $S_B(\bm{k})$. As a result, the Hilbert-Schmidt norm of $T$  is given by
    \begin{align*}
        \norm{T}_{\rmm{HS}}^2= \sum_{\bm{k} \in \Z^d} \qty(\frac{S_A(\bm{k})}{S_B(\bm{k})}-1)^2.
    \end{align*}
    The result follows by the necessary condition for non-mutual singularity of $u_A,u_B$.
\end{proof}
Using the just proved \autoref{KLD infinite condition}, we show that the KLD between two different solutions to \eqref{SPDE2D} is infinite if they have the same variance parameter. In particular, if we were to renormalize the spectral densities of the two solutions to have the same variance, the KLD would be infinite for all different parameter values.
\begin{proposition}\label{KLD infinite prop}
    Let $u_{A},u_{B}$ be the solutions to \eqref{SPDE2D} on $\revisiontwo{\T^d}$ with  \revisiontwo{periodic} boundary conditions and parameters $(\kappa _A,\bm{v}_A,\sigma _A),(\kappa _B,\bm{v}_B,\sigma _B)$ respectively. Then, if \revisiontwo{${\kappa _A}{\sigma _A}\neq {\kappa _B}{\sigma _B}$} or $\bm{v}_A\neq\bm{v}_B$ the measures $\mu_A, \mu_B$ defined by $u_{A},u_{B}$ are mutually singular and
    \begin{align*}
        KLD(\mu _A||\mu _B)=\infty.
    \end{align*}

\end{proposition}
\begin{proof}We have by \eqref{stationary density bounded} that
\begin{align}\label{sum1}
    \sum_{\bm{k} \in \Z^2} \qty(\frac{S_A(\bm{k})}{S_B(\bm{k})}-1)^2  = \sum_{\bm{k} \in \Z^2} \qty(\frac{\kappa _A^2 \sigma _A^2}{\kappa _B^2 \sigma _B^2}\frac{\qty(\kappa_B^2+4\pi^2\bm{k} \cdot \bm{H}_{\bm{v}_B}\bm{k})^2}{\qty(\kappa_A^2+4\pi^2\bm{k} \cdot \bm{H}_{\bm{v}_A}\bm{k})^2}-1)^2.
\end{align}
\revisiontwo{
A necessary condition for the convergence of the sum is that the summands tend to zero as $\norm{\bm{k}} \to \infty$. Let us examine the limit along rays. Let $\bm{k}_n$ be a sequence of integer vectors such that $\norm{\bm{k}_n} \to \infty$ and $\bm{k}_n / \norm{\bm{k}_n} \to \bm{u}$ for some unit vector $\bm{u} \in S^{1}$.Dividing numerator and denominator by $\norm{\bm{k}}^4$, the limit of the term inside the square is:
\begin{align}\label{limit ratio}
    \lim_{n \to \infty} \frac{S_A(\bm{k}_n)}{S_B(\bm{k}_n)} = \frac{\kappa _A^2 \sigma _A^2}{\kappa _B^2 \sigma _B^2} \qty( \frac{\bm{u} \cdot \bm{H}_{\bm{v}_B}\bm{u}}{\bm{u} \cdot \bm{H}_{\bm{v}_A}\bm{u}} )^2.
\end{align}
For the sum in \eqref{sum1} to converge, this limit must equal $1$ for all reachable directions $\bm{u}$. Since the directions of integer vectors are dense in $S^{1}$, this equality must hold for all $\bm{u} \in S^{1}$. That is:
\begin{align}\label{quadratic equality}
    \frac{\bm{u} \cdot \bm{H}_{\bm{v}_B}\bm{u}}{\bm{u} \cdot \bm{H}_{\bm{v}_A}\bm{u}} = \frac{\kappa _B \sigma _B}{\kappa _A \sigma _A} := \gamma, \quad \forall \bm{u} \in S^{1}.
\end{align}
Rearranging \eqref{quadratic equality}, we have $\bm{u} \cdot (\bm{H}_{\bm{v}_B} - \gamma \bm{H}_{\bm{v}_A}) \bm{u} = 0$ for all $\bm{u} \in S^1$. Since $\bm{H}_{\bm{v}_B} - \gamma \bm{H}_{\bm{v}_A}$ is a symmetric matrix, it must be zero. Thus, $\bm{H}_{\bm{v}_B} = \gamma \bm{H}_{\bm{v}_A}$.

Taking the determinant of both sides, and using that $\det(\bm{H}_{\bm{v}_A}) = \det(\bm{H}_{\bm{v}_B}) = 1$:
\begin{align*}
    1 = \det(\bm{H}_{\bm{v}_B}) = \det(\gamma \bm{H}_{\bm{v}_A}) = \gamma^2 \det(\bm{H}_{\bm{v}_A}) = \gamma^2.
\end{align*}
Since $\gamma$ is a ratio of real positive parameters, $\gamma > 0$, and thus $\gamma = 1$.

The condition $\gamma = 1$ implies two things. Firstly, $\bm{H}_{\bm{v}_B} = \bm{H}_{\bm{v}_A}$, which implies $\bm{v}_A = \bm{v}_B$ (due to the injectivity of the parameterization). Secondly, $\kappa_A \sigma_A = \kappa_B \sigma_B$.
}
This shows that unless the conditions of the proposition hold the sum diverges. The result follows from \Cref{KLD infinite condition}.
\end{proof}

\subsection{The $L^2$ distance}
One possible option is to consider the $L^2$ distance between  $u,u_{\bm{0}}$. An application of Fubini shows that given two Gaussian fields $u_A\sim \Nn(0,K_A),u_B\sim\Nn(0,K_B)$ on $\Dd $ 
\begin{align*}
    \E\qty[\norm{u_A-u_B}^2_{L^2(\Dd)}]=\mathrm{Tr} (K_A+K_B-2K_{AB}), \end{align*}
where $K_{AB}$ is the covariance between $u_A,u_B$ and $\mathrm{Tr}$ is the trace. It is unclear which choice of $K_{AB}$ would be the most appropriate. For example, choosing  $K_{AB}=0$ would be too coarse a measure as it would not consider any of the non-stationarity that occurs off the diagonal of $K_A, K_B$ (see Lemma \ref{diagonal lemma}). One can also choose $K_{AB}$ as to minimize the $L^2$ norm while keeping $(u_A,u_B)$ jointly Gaussian  by taking
\begin{align*}
    K_{AB}=\qty(\sqrt{K_B}K_A\sqrt{K_B})^{\frac{1}{2}}.
\end{align*}
This leads to the Wasserstein distance, which we discuss in the following subsection.

\subsection{Wasserstein distance between Gaussian measures}\label{wasserstein appendix}
Each solution to SPDE \eqref{SPDE2D} on $\revisiontwo{T}^d$ induces a measure on $L^2(\T^d)$. As a result, one possible way to measure the distance between $u$ and $u_{\bm{0}}$  is to take the Wasserstein distance between the induced Gaussian measures. It is known that given two Gaussian measures $\revisiontwo{\mu}_A\sim \Nn(m_A, K_A),\mu_B\sim\Nn(m_B, K_B)$ on a separable Hilbert space, the Wasserstein distance between them is given by \citet[Theorem 3.5]{gelbrich1990formula}
\begin{align}\label{W gaussian}
    W_2(\mu _A,\mu _B)=\norm{m_A-m_B}_{L^2(\T^d)}+\mathrm{Tr} \qty(K_A+K_B-2\qty(\sqrt{K_B}K_A\sqrt{K_B} )^\frac{1}{2}),
\end{align}
where $\mathrm{Tr} $ is the trace. However, this approach poses some difficulties. Let us write $ K_A, K_B$ for the covariance operators of two stationary solutions $u_A\sim M_A, u_B\sim M_B$ to \eqref{SPDE2D} with parameters $(\kappa _A,\bm{v}_A,\sigma _A),(\kappa _B,\bm{v}_B,\sigma _B)$ respectively. By \Cref{diagonal lemma}, $K_A, K_B$ diagonalize on the same basis, and using the expression of the spectral density in \eqref{stationary density bounded}, we obtain that
\begin{align*}
    W_2(M_A,M_B) & = \mathrm{Tr} \qty(K_A+K_B-2 \sqrt{K_AK_B})= \sum_{\b{j}\in \Z^d} \abs{ \left[\sqrt{K}-\sqrt{K_0}\right]_\b{jj}}^2                                                                                                                                     \\
                 & =\sum_{\b{j} \in \Z^d} \qty(\frac{\sqrt{4 \pi}\kappa_A \sigma_A }{\kappa_A ^2+ 4 \pi^2\bm{\xi}_\b{j}^T \bm{H}_{\bm{v}_A}\bm{\xi}_\b{j}}-\frac{\sqrt{4 \pi}\kappa_B \sigma_B }{\kappa_B ^2+ 4 \pi^2\bm{\xi}_\b{j}^T \bm{H}_{\bm{v}_B}\bm{\xi}_\b{j}})^2 .
\end{align*}
Making the side-length $\b{T}$ of the torus go to infinity \revisiontwo{heuristically} shows that, by definition of $\bm{\xi}_\b{j}$ our discrete sum becomes an integral with
\begin{align*}
    W_2(M _A,M_B)=\frac{\rmm{vol}(\T^d) }{\pi} \int_{\R^2} \qty(\frac{\kappa_A }{\kappa_A ^2+ \bm{\xi}^T \bm{H}_{\bm{v}_A}\bm{\xi}}-\frac{\kappa _B}{\kappa _B^2+ \bm{\xi}^T \bm{H}_{\bm{v}_B}\bm{\xi} })^2\d\bm{\xi} +O(1).
\end{align*}
The Wasserstein distance between $\mu _0,\mu _1$ scales as $\rmm{vol}(\T^d) $ times the Hellinger distance of the spectral densities of $u_A,u_B$. Thus, a reasonable option could be to define the scaling as the distance
\begin{align*}
    D_{W_2} (M_A,M_B):= \lim_{T \to \infty} \frac{W_2(M_A,M_B)}{\rmm{vol}(\T^d) }=\mathrm{Hell}(S_A,S_B).
\end{align*}
However, a calculation in this simplified case shows that this will be bounded. For example, in the case where, as in the base model, $\bm{v}_A=\bm{v}_B=\bm{0}$, and $\sigma_A = \sigma _B =1$, we have the following.
\begin{align*}
    D_{W_2} (u_A,u_B)=2\qty(1-\frac{\kappa_A\kappa_B \left(\log
            \left(\kappa _A^2\right)-\log
            \left(\kappa_B ^2\right)\right)}{\kappa _A^2-\kappa_B ^2}).
\end{align*}
The above expression takes a maximum of $2$. This makes putting an exponential distribution on the Wasserstein distance impossible, and as a result, the Wasserstein distance cannot be used to define PC priors for our model \eqref{SPDE2D}.

\subsection{Sobolev distance between fields}
Given a stationary field $u$ on $\revisiontwo{\R^2} $ with $s$-times differentiable covariance function $r(\bm{x}):= \rmm{Cov}[u(\bm{x},u(\bm{0}))]$, we define the seminorm $\abs{\cdot }_s$ as the Sobolev seminorm of order $s$ of $r$. That is, \revisiontwo{for $s \in \R$ we define }
\begin{align}
    \abs{u}_s:= \norm{\nabla^s r}_{L^2(\revisiontwo{\R^2})}=\norm{\norm{2\pi\bm{\xi}}^{s}S(\bm{\xi})}_{L^2(\revisiontwo{\R^2})}:=\qt{\int_{\revisiontwo{\R^2}} \norm{2\pi\bm{\xi}}^{2s}S(\bm{\xi})^2 \d\bm{\xi}}^{\frac{1}{2}},
\end{align}
If we think of the $L^2(\revisiontwo{\R^2})$ norm of $r$ as a measure of the total ``random energy'' present in $u$. Then, $\abs{\cdot }_s$ measures the oscillation size in this energy to order $s$.

We now discuss which value of $s$ is appropriate. In our case, as we saw in \eqref{spectral density R^d}, the solution $u$ to \eqref{SPDE2D} with parameters $\kappa, \bm{v}, \sigma_u$  has spectral density
\begin{align*}
    S_{\kappa, \bm{v}, \sigma_u}(\bm{\xi})= \frac{4\pi\kappa ^2\sigma_u^2}{\qty(\kappa ^2+4\pi^2\bm{\xi} \cdot \bm{H_v}\bm{\xi} )^2} .
\end{align*}
A change of variables $\xi \to \kappa \xi $ shows that the behavior in $\kappa $ is
\begin{align*}
    \abs{u}_s\sim\kappa ^{s-1} .
\end{align*}
If we take $s=0$, then  $\abs{u}_0$ becomes infinite, whereas if we choose $s=1$, we obtain that  $\abs{u}_1$ is bounded in $\kappa $ and thus cannot be exponentially distributed. \revisiontwo{It is also not possible to choose $s\geq 3$ as then the integrand has decay rate smaller than $d=2$, leading to an infinite integral.} This leads to the choice $s=2$ used in Definition \ref{metric def}. \revisiontwo{However, fractional choices $ s \in (2,3)$ are also possible. Deriving PC priors for these exponents and comparing their performance to the ones in this work would be an interesting avenue of future research.}

\revisiontwo{For choices of model smoothness and dimension different than the ones considered here, the distance could be adapted by modifying the value of $s$.}

\revisionone{Write $M_{\kappa,\bm{v}}$ for the model given by SPDE \eqref{SPDE2D} with parameters $(\kappa,\bm{v})$ and with variance fixed to $\sigma _u=1$. Due to the rescaling of the spectral densities in \Cref{metric def}, the choice of $\sigma _u$  does not affect the distance between models and can be set to any other positive value. Then the following holds:}
\revisionone{\begin{lemma}\label{seminorm is norm}
The distance $D_2$  in Definition \ref{metric def} distinguishes between different models. That is, for any two stationary models $M_{\kappa_A,\bm{v}_A},M_{\kappa_B,\bm{v}_B}$, we have that
\begin{align*}
    D_2(M_{\kappa_A,\bm{v}_A},M_{\kappa_B,\bm{v}_B})=0 \iff (\kappa_A,\bm{v}_A)=(\kappa_B,\bm{v}_B).
\end{align*}
\end{lemma}
\begin{proof}
    By  definition of $D_2$, we have that 
    \begin{align*}
        D_2(M_{\kappa_A,\bm{v}_A},M_{\kappa_B,\bm{v}_B})^2=\int_{\revisiontwo{\R^2}} \norm{2\pi\bm{\xi} }^4\abs {S_{\kappa_A,\bm{v}_A}(\bm{\xi})-S_{\kappa_B,\bm{v}_B}(\bm{\xi})}^2 \d \bm{\xi }.
    \end{align*} 
    By basic measure theory, this is zero if and only if the integrand is zero almost everywhere. This is if and only if the spectral densities are equal almost everywhere. By \Cref{spectral density R^d}, and since $\bm{H}_{\bm{v}_A}\neq \bm{H}_{\bm{v}_B}$ for $\bm{v}_A \neq \bm{v}_B$  this is equivalent to $(\kappa_A,\bm{v}_A)=(\kappa_B,\bm{v}_B)$. This concludes the proof.
\end{proof}}

\begin{observation}
    The seminorm $\abs{u}_s$ is different to the Sobolev  seminorm of $u$ as we have that, given a multi-index $\alpha \in \Z^d$, the derivative  $D^\alpha u$ of a stationary field $u$  exists as long as
    \begin{align*}
        \int_{\revisiontwo{\R^2}} \abs{(2\pi\bm{\xi})^{2\alpha}} S(\bm{\xi})  \d \bm{\xi} < \infty .
    \end{align*}
\end{observation}
Thus, one possible choice of metric, corresponding to $s=1 /2$, could also be to set
\begin{align*}
    \abs{u}_*:=\int_{\revisiontwo{\R^2}} \norm{2\pi\bm{\xi}} S(\bm{\xi})  \d \bm{\xi}   .
\end{align*}
This distance is unbounded in $\kappa $ and finite at the base model. A calculation gives a distance of a similar form to the one derived for the distance used in this paper in \Cref{distance lemma}.
\begin{aligneq}\label{alternative distance}
    d_*(\kappa , \bm{v})^2 & :=\lim_{\kappa_0,\bm{v} \to 0 }\int_{\revisiontwo{\R^2}} \norm{2\pi\bm{\xi}} \qty(\frac{S_{\kappa, \bm{v}, \sigma_u}}{\sigma_u^2}-\frac{S_{\kappa, \bm{v}, \sigma_0}}{\sigma _0^2})  \d \bm{\xi}
    \\
    & =2\pi{ E\left(1-\ex{2\norm{\bm{v}} }\right)\ex{-\frac{\norm{\bm{v}} }{2}}}\kappa,
\end{aligneq}
where $E$ is the complete elliptic integral
\begin{align*}
    E(m):=\int_{0}^{\pi /2}(1-m\sin^2(\alpha ))^\frac{1}{2} \d \alpha .
\end{align*}
This distance is unbounded in $\kappa $ and $\bm{v}$ and $0$ at the base model. As a result, using $d_*$ instead of $d$  could also provide a valid alternative.
\section{Derivation of PC priors}\label{distance calc stationary sec}
\begin{lemma}\label{distance lemma}
    Let $M_{\kappa ,\bm{v},\sigma_u},M_{\kappa _0,\bm{0},\sigma_0}$ be the stationary models given by \eqref{SPDE2D} with parameters  $(\kappa,\bm{v},\sigma_u)$ and  $(\kappa_0,\bm{0},\sigma_0)$ respectively. Then,
    \begin{align*}
        d(\kappa ,\bm{v}):=\lim_{\kappa_0 \to 0}D_2(M_{\kappa ,\bm{v},\sigma_u},M_{\kappa_0,\bm{0},\sigma_0})^2=\frac{\pi}{3} \kappa ^2 (3 \cosh (2 \norm{\bm{v}} )+1) .
    \end{align*}
\end{lemma}
\begin{proof}
    For any $(\kappa ,\bm{v},\sigma_u)$ the spectral density of $u_{\kappa ,\bm{v},\sigma_u}$ is
    \begin{align*}
        S_{\kappa ,\bm{v},\sigma_u}(\bm{\xi})=\frac{4\pi\kappa^2 \sigma _u^2}{\qty(\kappa ^2+4\pi^2\bm{\xi}^T \bm{H} \bm{\xi})^2} .
    \end{align*}
    As a result,  in a distributional sense
    \begin{align*}
        \lim_{\kappa _0 \to 0} \frac{S_{\kappa_0, \bm{0},\sigma_0}}{\sigma_0^2}(\bm{\xi})= \lim_{\kappa _0 \to 0} \frac{4\pi\kappa_0^2 }{\qty(\kappa_0 ^2+4\pi^2\norm{\bm{\xi }}^2 )^2}=\delta_{\bm{0}} (\bm{\xi}),
    \end{align*}
    where $\delta_{\bm{0}}$ is the Dirac delta at $\bm{0}$.
    In consequence, expanding the square in \eqref{pseudo metric} and using the change of variable $\bm{\xi } \to \kappa \bm{\xi }/(2\pi)$, we obtain that
    \begin{equation}\label{distance calc exp}
        \lim_{\kappa_0 \to 0}D_2(M_{\kappa ,\bm{v},\sigma_u},M_{\kappa_0,\bm{0},\sigma_0})^2=\int_{\R^2} \norm{2\pi\bm{\xi }}^4 \frac{S_{\kappa ,\bm{v}, \sigma _u}(\bm{\xi })^2}{\sigma _u^4} \d \bm{\xi }  =4\kappa^{2} \int_{\R^2}\frac{\norm{\bm{\xi}}^{4}}{(1+\bm{\xi}^T \bm{H}_{\bm{v}} \bm{\xi})^4} \d\bm{\xi}.
    \end{equation}
    Let $\alpha:= \arg(\bm{v})$ and $B_{\alpha /2}$ be a rotation by $\alpha /2$ radians. Then
    \begin{align*}
        \bm{B}_{\alpha /2}    \bm{H}_{\bm{v}} \bm{B}_{ \alpha /2 }^{-1}= \begin{bmatrix}
                                                                             \ex{-\norm{\bm{v}} } & 0
                                                                             \\
                                                                             0                    & \ex{\norm{\bm{v}}}
                                                                         \end{bmatrix} .
    \end{align*}
    Thus, by rotating and then changing to polar coordinates, we obtain
    \begin{align*}
         & \lim_{\kappa_0 \to 0}D_2(M_{\kappa ,\bm{v},\sigma_u},M_{\kappa_0,\bm{0},\sigma_0})^2=4\kappa^2 \int_{0}^{2\pi } \int_{0}^\infty \frac{r^5}{\left(1+r^2 \left(\ex{\norm{\bm{v}} } \sin ^2(\phi )+\ex{-\norm{\bm{v}} }
        \cos ^2(\phi )\right)\right)^4} \d r \d \phi                                                                                                                                                                    \\
         & =\int_{0}^{2\pi } \frac{4\kappa ^2}{6 \left(\ex{\norm{\bm{v}}}  \sin ^2(\phi )+\ex{-\norm{\bm{v}} } \cos
            ^2(\phi )\right)^3}\d  \phi=\frac{ \pi }{3 } (3 \cosh (2 \norm{\bm{v}} )+1)\kappa ^2 .\end{align*}
    This concludes the proof.
\end{proof}
\begin{observation}[\revisionone{Selection of Sobolev exponent}]
\revisionone{The same calculation to the one that gave \Cref{distance calc exp}  shows that for $s \geq 1$ }
\begin{equation*}
    \lim_{\kappa_0 \to 0}D_s(M_{\kappa ,\bm{v},\sigma_u},M_{\kappa_0,\bm{0},\sigma_0})^2=\int_{\R^2} \norm{2\pi\bm{\xi }}^{2s} \frac{S_{\kappa ,\bm{v}, \sigma _u}(\bm{\xi })^2}{\sigma _u^4} \d \bm{\xi }  =4\kappa^{2(s-1)} \int_{\R^2}\frac{\norm{\bm{\xi}}^{2s}}{(1+\bm{\xi}^T \bm{H}_{\bm{v}} \bm{\xi})^{4}} \d\bm{\xi}.
\end{equation*}
This shows that the distance $d_s$ is constant in $\kappa $ for $s=1$ and becomes infinite for \revisiontwo{$s\geq 3$ due to the decay rate of the integrand being too low}. For this reason, the value $s=2$ was chosen.

\end{observation}
As we see from \Cref{distance lemma}, $d(\kappa ,\bm{v})$ is independent of the argument $\alpha$ of $\bm{v}$  and depends only on $\kappa $ and $r:=\norm{\bm{v}}$. For this reason, we consider $\alpha $ to be independent of $\kappa,\bm{v}$, and, since a priori, there is no preferred direction for the anisotropy, we set a uniform prior on $\alpha$
\begin{equation}\label{alf uniform}
    \pi_\alpha (\alpha)=\frac{1}{2\pi }1_{[0,2\pi ]}(\alpha)
\end{equation}
We will now define priors on $\kappa,r$ following a sequential construction. To do so, we will use the fact that the distance to the base model $\rho (\kappa,\bm{v} )$ decomposes as a product
\begin{align}\label{pseudo metric calculation}
    d(\bm{v},\kappa)  =f(r)g(\kappa),\quad
    f(r)            :=\qty(\frac{ \pi }{3} (3 \cosh (2 \norm{\bm{v}} )+1) )^\frac{1}{2},\quad g(\kappa):=\kappa.
\end{align}
Together with the following lemma,
\begin{lemma}[Distribution of product]\label{product lemma}
    Let $X,Y$ be two positive random variables such that $Y|X$ has density
    \begin{equation*}
        f_{Y|X}(y,x)=\lambda x \exp(-\lambda x y )1_{[0,+\infty)}(y).
    \end{equation*}
    Then, the product $Z:=XY$ follows an exponential distribution with parameter $\lambda $.
\end{lemma}
\begin{proof}
    By the law of total expectation, the cumulative density function of $Z$ verifies
    \begin{align*}
        F_Z(z) & =\mathbb{P}[XY\leq z]=\E[\E[1_{XY\leq z}|X]]=\E\qty[\int_{0}^{z/X}f_{Y|X}(y|X) dy] \\&=\int_{0}^\infty\int_{0}^{z /x} f_{Y|X}(y|x)\d y\d\mathbb{P}_{X}(x),
    \end{align*}
    where $\mathbb{P}_{X}$ is the push-forward of $\mathbb{P}$ by $X$.
    Applying the fundamental theorem of calculus shows that the density of  $Z$ is
    \begin{equation*}
        F'_Z(z)=\int_{0}^\infty \frac{1}{x}f_{Y|X}(z /x|x)d\mathbb{P}_{X}(x)=\lambda \ex{-\lambda z} \int_{0}^\infty d\mathbb{P}_{X}(x)=\lambda \ex{-\lambda z}.
    \end{equation*}
    This concludes the proof.
\end{proof}
The above lemma states that if we want the distance $d=f g$ to be exponentially distributed, it suffices to let $f$ follow  \emph{any} distribution and have $g$ conditional on  $f$ be exponentially distributed with parameter proportional to  $f$. Since $f$ takes a minimum of
\begin{equation*}
    f(0)=\sqrt{4\pi/3 },
\end{equation*}
the conditional distribution of $f$ given $g$ cannot follow an exponential distribution. Therefore, we will instead select the conditional distribution of $g$ given $f$ to be exponential.

While \Cref{product lemma} allows complete freedom in choosing a distribution for $f$, the symmetry in the distance $fg$ implies there is no intrinsic reason to prefer penalizing an increase in $f$ either more or less than an increase in $g$. Consequently, we will also impose an exponential penalty on $f$ by setting its marginal density to be equal to the following,
\begin{equation*}
    \pi_{f}(f)= \lambda_{\bm{v}} \ex{-\lambda_{\bm{v}} (f- f(0))}1_{[f(0),\infty)}.
\end{equation*}

We can set priors in the stationary setting in the same way as in the previous section.

\subsection{Proof of \Cref{Prior r kappa stationary theorem}} %
\begin{proof}
    Let us write
    \begin{align*}
        r=\norm{\bm{v}}, \quad  \alpha =\arg(\bm{v}) .
    \end{align*}
    We have already determined the prior distributions
    \begin{equation}\label{prior f,g}
        \pi_{f}(f)= \lambda_{\bm{v}} \ex{-\lambda_{\bm{v}} \qty(f- f(0))}1_{[f(0),\infty)}(f),\quad \pi_{\kappa |f}(\kappa |f)=\lambda_{\bm{\theta}} f \ex{-\lambda_{\bm{\theta}} f \kappa }1_{[0,\infty)}(\kappa )
    \end{equation}
    Since $f:[0,\infty) \to [f(0), \infty)$ is bijective, we may apply a change of variables to obtain that the marginal prior for $r$ is
    \begin{equation}\label{prior r}
        \pi _r(r)=\pi_{f}(f(r))\abs{f'(r)}= \lambda_{\bm{v}} \ex{-\lambda_{\bm{v}}\qty( f(r)-f(0))}\abs{f'(r)}  .
    \end{equation}
    Again, by the bijectivity of $f$, we obtain that
    \begin{equation}\label{prior equation 2}
        \pi_{\kappa|r}=\pi_{\kappa|f=f(r)}=\lambda_{\bm{\theta}} f(r) \ex{-\lambda_{\bm{\theta}} f(r) \kappa }1_{[0,\infty)}(\kappa ).
    \end{equation}
    Since $\alpha$ is uniformly distributed on $[0,2\pi ]$ independently of $r,\kappa $ the joint prior for $(\kappa,\alpha,r )$ is
    \begin{align*}
        \pi_{\kappa,\alpha,r }(\kappa,\alpha,r )=\frac{1_{[0,2\pi]}(\alpha)}{2\pi } \pi_r (r)\pi_{\kappa|r } (\kappa|r ) .
    \end{align*}
    Changing variables from polar to Cartesian coordinates gives
    \begin{align}\label{prior equation 1}
        \pi_{\kappa,\bm{v}}(\kappa,\bm{v} )=\frac{1}{2\pi \norm{\bm{v}}}  \pi_r (\norm{\bm{v}} )\pi_{\kappa|r } (\kappa\mid\norm{\bm{v}}) .
    \end{align}
    Using equations \eqref{prior r} and \eqref{prior equation 2} in \eqref{prior equation 1} and using the expression for $f$ derived in \eqref{pseudo metric calculation}  concludes the proof. \end{proof}

\subsection{Proof of \Cref{par r theorem}}
\begin{proof}
    The prior for $f$ is defined in \eqref{prior f,g}. Since $f$ is increasing, we obtain that the cumulative distribution function of $r$ is for all $r_0\geq 0$
    \begin{align}\label{cdf r2}
        F_r(r_0)=\mathbb{P}(r\leq r_0)=\mathbb{P}(f(r)\leq f(r_0))=1-\ex{-\lambda_{\bm{\theta}} (f(r_0)-f(0))}.
    \end{align}
    The theorem follows by a straightforward algebraic manipulation.
\end{proof}

\subsection{Proof of \Cref{par kappa theorem}}%
\begin{proof}
    We begin by calculating the prior $\pi_{\kappa }$  for $\kappa $. A calculation shows that for all $\kappa  >0$
    \begin{align*}
        \pi_{\kappa }(\kappa ) & =\int_{\R}\pi_{f}(f)\pi_{\kappa |f}(\kappa  |f)df=\int_{f(0)}^\infty\lambda_{\bm{v}} \lambda_{\bm{\theta}} f  \exp(-\lambda_{\bm{v}} \qty(f- f(0))-\lambda_{\bm{\theta}}  f \kappa  )df
        \\&= \frac{\lambda_{\bm{\theta}}  \lambda_{\bm{v}} \ex{-f(0) \kappa  \lambda_{\bm{\theta}} } \left(f(0) \kappa  \lambda_{\bm{\theta}} +f(0) \lambda_{\bm{v}}+1\right)}{\left(\kappa  \lambda_{\bm{\theta}} +\lambda_{\bm{v}}\right)^2},\end{align*}
    whereas for $\kappa <0$ we have $\pi_\kappa (\kappa )=0$.
    A further calculation shows that the CDF of $\kappa  $ is
    \begin{align}\label{cdf kappa2}
        F_\kappa  (\kappa _0 )=\int_{-\infty}^{\kappa_0}\pi_\kappa (\kappa)  \d \kappa = \qt{1-\frac{\lambda_{\bm{v}} \ex{-f(0) \lambda_{\bm{\theta}} \kappa _0}}{ \lambda_{\bm{\theta}} \kappa _0+\lambda_{\bm{v}}}}1_{[0,\infty]}(\kappa _0).
    \end{align}
    Now solving for $\lambda_{\bm{\theta}}  $ gives that for any $\kappa _0>0$
    \begin{equation}\label{sg}
        \mathbb{P}[\kappa >\kappa _0]=\alpha \Longleftrightarrow \lambda_{\bm{\theta}} = \frac{1}{f(0)\kappa _0} W_0\qty(\frac{\ex{\lambda_{\bm{v}} f(0)}\lambda_{\bm{v}}f(0)}{\alpha} ) -\frac{\lambda_{\bm{v}} }{\kappa _0}
    \end{equation}
    This concludes the proof.
\end{proof}

\section{Gaussian reparameterization of the stationary PC priors}\label{reparameterization section}
We show how to write $(\kappa,\bm{v}) $ jointly as a function of a three-dimensional standard Gaussian $\vb{Y}\stackrel{d}{=} \Nn(\vb{0}, \mathbf{I}_3)$. Our idea is to generalize the method of inverse sampling. This is done by building an invertible generalization of the cumulative distribution function.
\begin{definition}
    Given a random variable $\bm{X}=(X_1,X_2) \in \R^2$ we define the \emph{generalized CDF of $\bm{X}$} as
    \begin{align*}
        \varphi(\bm{X},\bm{x}):=\varphi_{\bm{X}}(\bm{x}):=(\mathbb{P}[X_1\leq x_1],\mathbb{P}[X_2\leq x_2|X_1=x_1]).
    \end{align*}
\end{definition}
\begin{observation}\label{obs indep}
    In the case where $X_1, X_2$ are independent we obtain that
    \begin{align*}
        \varphi_{\bm{X}}(\bm{x})=(F_{X_1}(x_1),F_{X_2}(x_2)).
    \end{align*}

\end{observation}

By construction $\varphi_{\bm{X}}:\R^2\to [0,1]^2$. The definition can be extended without difficulty to dimensions larger than $2$. However, the notation becomes more cumbersome. The motivation for the above definition is given by the following properties, which mimic that of the $1$-dimensional CDF.
\begin{lemma}\label{generalized CDF lemma}
    Let $\bm{X}$ be a \revisiontwo{continuous} two-dimensional random variable \revisiontwo{with strictly increasing CDFs}, then the following hold
    \begin{enumerate}
        \item \label{p1} The distribution of $\bm{X}$ is determined by  $\varphi$. That is if $\bm{Y}$ is a two dimensional random variable such that $\varphi_{\bm{X}}=\varphi_{\bm{Y}}$, then $\bm{X}\stackrel{d}{=} \bm{Y}$.
        \item \label{p2} Consider a function $\phi:A\to\R^2$ defined on some subset $A \subset \R^2$ of the form
              \begin{align}\label{form}
                  \phi(x_1,x_2)=(\phi_1(x_1),\phi_2(x_1,x_2)),
              \end{align}
              where $\phi_1$ is invertible and  $\phi_1,\phi_2(x_1,\cdot)$ are monotone functions.
              Then,
              \begin{align*}
                  \varphi(\phi(\bm{X}),\phi(\bm{x}))=\varphi(\bm{X},\bm{x}), \quad\forall \bm{x} \in A.
              \end{align*}
        \item Let $\bm{X}$ be a  $\R^2$ valued random variable. Then $\varphi_{\bm{X}}$ has an inverse and $\varphi_{\bm{X}}, \varphi_{\bm{X}}^{-1}$ are both of the form  \eqref{form}.
    \end{enumerate}
\end{lemma}
\begin{proof}
    We prove the above points in order.
    \begin{enumerate}
        \item Suppose  $\varphi_{\bm{X}}=\varphi_{\bm{Y}}$, then by definition of $\varphi$
              we have that for all $\bm{x}=(x_1,x_2) \in \R^2$
              \begin{align*}
                  \mathbb{P}[X_1\leq x_1]=\mathbb{P}[Y_1\leq x_1],\quad\mathbb{P}[X_2\leq x_2|X_1=x_1]=\mathbb{P}[Y_2\leq x_2|Y_1=x_1] .
              \end{align*}
              From the basic theory of random variables, the left-hand side of the above implies that $X_1\stackrel{d}{=}Y_1$. As a result, we obtain that
              \begin{align*}
                  \mathbb{P}[X_2\leq x_2] & =\int_{\R}\mathbb{P}[X_2 \leq x_2|X_1=x_1]\d\mathbb{P}_{X_1}(x_1) \\&=\int_{\R}\mathbb{P}[Y_2 \leq x_2|Y_1=x_1]\d\mathbb{P}_{Y_1}(x_1)=\mathbb{P}[Y_2\leq x_2].
              \end{align*}
              Since this holds for all $x_2 \in \R$ we deduce as previously that $X_2\stackrel{d}{=} Y_2$. Thus we obtain that $\bm{X}\stackrel{d}{=} \bm{Y}$, as desired.

        \item To prove   \cref{p2} we work with  $\varphi$ componentwise. By the monotonicity of $\phi_1$, it is clear that for the first component,               \begin{align*}
                  \varphi_1(\phi(\bm{X}),\phi(\bm{x})):=\mathbb{P}[\phi_1(X_1)\leq \phi_1(x_1)]=\mathbb{P}[X_1\leq x_1]=\varphi_1(\bm{X},\bm{x}).
              \end{align*}
              The second part is proved similarly. We have that
              \begin{align*}
                   & \varphi_2(\phi(\bm{X}),\phi(\bm{x})):=\mathbb{P}[\phi_2(X_1,X_2)\leq \phi_2(x_1,x_2)|\phi_1(X_1)=\phi_1(x_1)] \\&=\mathbb{P}[\phi_2(X_1,X_2)\leq \phi_2(x_1,x_2)|X_1=x_1]=\mathbb{P}[\phi_2(x_1,X_2)\leq \phi_2(x_1,x_2)|X_1=x_1]
                  \\&= \mathbb{P}[X_2\leq x_2|X_1=x_1]=\varphi_2(\bm{X},\bm{x}),
              \end{align*}
              where in the second equality, we used that $\phi_1$ is invertible, so the $\sigma$-algebra generated by $X_1$ is the same as that generated by  $\phi(X_1)$ and $\phi(X_1)=\phi(x_1)$ if and only if $X_1=x_1$. And in the fourth equality, we used the fact that $\phi_2(x_1,\cdot)$ is monotone.
        \item We first prove $\varphi_{\bm{X}}$ is invertible. Let us take $\bm{p}=(p_1,p_2) \in [0,1]^2$. Then, the first component of $\varphi_{\bm{X}}$ is the univariate CDF $F_{X_1}$, and we can define
              \begin{align*}
                  x_1:=F_{X_1}^{-1}(p_1).
              \end{align*}
              Now, since $\varphi_2(\bm{X},(x_1,\cdot ))$ is increasing (as a function of the dot) for each $x_1$, it has an inverse
              \begin{align*}
                  x_2:=\varphi_2(\bm{X},(x_1,\cdot ))^{-1}(p_2).
              \end{align*}
              An algebraic verification now shows that $\varphi_{\bm{X}}$ has inverse
              \begin{align}\label{inverse gen cdf}
                  \varphi_{\bm{X}}^{-1}(\bm{p})=(x_1,x_2).
              \end{align}
              Furthermore, since $\varphi_{\bm{X}}$ is monotone in each component, so is $\varphi^{-1}_{\bm{X}}$.
    \end{enumerate}
    This completes the proof.
\end{proof}

\begin{lemma}\label{multivariate inverse sampling}
    Given \revisiontwo{continuous} random variables $\bm{X},\bm{Y}$ valued in $\R^2$ it holds that
    \begin{align*}
        \varphi_{\bm{X}}^{-1}\circ\varphi_{\bm{Y}}(\bm{Y})\stackrel{d}{=} \bm{X}
    \end{align*}
\end{lemma}
\begin{proof}
    By \cref{p1} of \Cref{generalized CDF lemma}, it suffices to prove that the generalized CDFs of $\bm{X}$ and  $\bm{Y}$ coincide. To this aim, we have that, by \cref{p2} of the previous lemma
    \begin{align*}
        \varphi(\varphi_{\bm{X}}^{-1}\circ\varphi_{\bm{Y}}(\bm{Y}),\bm{x})=\varphi(\bm{Y},\varphi_{\bm{Y}}^{-1}\circ\varphi_{\bm{X}}(\bm{x}))=\varphi_{\bm{Y}}(\varphi_{\bm{Y}}^{-1}(\varphi_{\bm{X}}(\bm{x})))=\varphi_{\bm{X}}(\bm{x}).
    \end{align*}
    This concludes the proof.
\end{proof}

We now calculate the transformation that makes $(\kappa,r)$ Gaussian.
\begin{lemma}\label{Gaussian r,k}
    Let $(\bm{v},\kappa)$ be distributed according to the PC priors $\pi (\bm{v},\kappa )$ and define $r:=\norm{\bm{v}} $, \revisiontwo{and let $f$ be as in \Cref{Prior r kappa stationary theorem}}.Then, given $\bm{U} \sim \Nn(\bm{0},\mathbf{I}_2)$ it holds that
    \begin{align*}
        (r,\kappa)\stackrel{d}{=}\qty( f^{-1}\qty(f(0)-\frac{\log(1-\Phi(U_1))}{\lambda_{\bm{v}} } ), \qty(\frac{\log(1-\Phi(U_1))}{\lambda_{\bm{v}} }-f(0)  )^{-1}\frac{\log(1-\Phi(U_2))}{\lambda_{\bm{\theta}} }
        ),
    \end{align*}
    where $\Phi$ is the CDF of a univariate standard Gaussian and $r:=\norm{\bm{v}}$.
\end{lemma}
\begin{proof}
    The CDF $F_r$ of $r$ was calculated in \eqref{cdf r2}. An inversion gives
    \begin{align*}
        x_1:=F_r^{-1}(p_1)             & = f^{-1}\qty(f(0)-\frac{\log(1-p_1)}{\lambda_{\bm{v}} }  ).    \end{align*}
    To calculate the CDF of $\kappa |r$, we work with the density in \eqref{prior f,g}. This gives
    \begin{align*}
        \varphi_2((r ,\kappa),(x_1,x_2))= 1-\ex{-\lambda_{\bm{\theta}} f(x_1)x_2}.
    \end{align*}
    For each fixed $x_1$, the above has the inverse
    \begin{align*}
        x_2:=\varphi_2((r,\kappa),(x_1,\cdot ))^{-1}(p_2)=- \frac{\log(1-p_2)}{\lambda_{\bm{\theta}} f(x_1)}.
    \end{align*}
    This, combined with \eqref{inverse gen cdf}, the preceding \Cref{multivariate inverse sampling} (with $\bm{Y}= \bm{U}$), and \Cref{obs indep} (applied to $\bm{U}$), completes the proof.
\end{proof}
\subsection{Proof of \Cref{reparameterization theorem stationary}}
\begin{proof}
    Let $\alpha\sim \Uu_{[0,2\pi ]} $ be independent from $r$. Then, by the definition $r= \norm{\bm{v}} $,
    \begin{align*}
        \bm{v}\stackrel{d}{=}(r\cos(\alpha ),r\sin(\alpha )),
    \end{align*}
    The proof then follows from Lemma \ref{Gaussian r,k} and by observing that $ R\qty(\sqrt{Y_1^2+Y_2^2} ),\Phi(Y_3)$ are i.i.d. uniformly distributed on  $[0,1]$, and that
    \begin{align*}
        \frac{Y_1}{\sqrt{Y_1^2+Y_2^2} } \stackrel{d}{=} \cos(\alpha),\quad \frac{Y_1}{\sqrt{Y_1^2+Y_2^2} } \stackrel{d}{=} \sin(\alpha)  .
    \end{align*}
\end{proof}
\subsection{Proof of \Cref{relationship prop}}

\begin{proof}
    The determinant of $\bm{H}_{\bm{v}(\alpha)}$ is $(\beta+\gamma)\gamma$. Setting this equal to $1$ gives $\beta=(1-\gamma^2)/\gamma$. Next, we note that given any unit vector $\bm{e}$, it holds that
    \begin{equation*}
        \mathbf{I}=\bm{e} \bm{e}^T+\bm{e}_\perp \bm{e}^T_\perp.
    \end{equation*}
    Substituting this into \eqref{fuglstadtparam} with $\bm{e}=\bm{v}(\alpha)$ gives
    \begin{equation*}
        \bm{H}_{\bm{v}(\alpha)}=(\gamma+\beta)\bm{v}(\alpha)\bm{v}(\alpha)^T+\gamma \bm{v}(\alpha)_\perp \bm{v}(\alpha)_\perp ^T=\frac{1}{\gamma}\bm{v}(\alpha)\bm{v}(\alpha)^T+\gamma \bm{v}(\alpha)_\perp \bm{v}(\alpha)_\perp ^T.
    \end{equation*}
    Setting $\bm{H}_{\bm{v}}=\bm{H}_{\bm{v}(\alpha)}$, we obtain the condition
    \begin{equation*}
        \frac{1}{\gamma}\bm{v}(\alpha)\bm{v}(\alpha)^T+{\gamma} \bm{v}(\alpha)_\perp \bm{v}(\alpha)^T_\perp=\frac{\ex{\norm{\bm{v}} }}{\norm{\bm{v}}^2 }\tl{\bm{v}}\tl{\bm{v}}^T+\frac{\ex{-\norm{\bm{v}} }}{\norm{\bm{v}}^2 }\tl{\bm{v}}_\perp\tl{\bm{v}}_\perp^T.
    \end{equation*}
    Thus, equality holds in the conditions of the theorem.
\end{proof}
\section{Simulation study details}\label{simulation details}
In this section, we expand on the simulation study of \Cref{simulation section}. We discuss the details of the simulation of the SPDE, compare the priors, and present some additional results.
\subsection{Sampling from the SPDE }
To simulate the Gaussian field resulting from \eqref{SPDE2D}, we use the \texttt{fmesher} package. This package uses a FEM method to approximate the precision matrix of $u$. That is, we approximate $u$ to be in the Hilbert space $H$ spanned by our finite element basis $\set{\phi_i}_{i=1}^n$ and impose
\begin{align}\label{FEM equation}
    \qty[\br{\phi_i, \Ll u}, i\in \{1,...,n\}]= \qty[\br{\phi_i, \kappa \dot{\Ww}}, i\in\{1,\ldots,n\}]
\end{align}
The condition that $u\in H$ leaves us with
\begin{align}\label{approx u}
    u(\bm{x})= \sum_{i=1}^n \phi_i(\bm{x}) u_i,
\end{align}
where one can show that the vector of weights $\bm{u}=(u_1,\ldots,u_n)$ is Gaussian. Let us write $\bm{Q}_u$ for its precision matrix, $\bm{u} \sim \Nn(\bm{0},\bm{Q}_u^{-1})$. Then, from
\eqref{approx u}, knowing $\bm{Q}_u$ determines  $\bm{u}$. The value of $\bm{Q}_u$ is determined in turn by \eqref{FEM equation}, as substituting in  \eqref{approx u} and applying integration by parts gives
\begin{align}\label{SPDE discrete}
    \bm{L} \bm{u}\sim \Nn(\bm{0},\bm{C}_\kappa),
\end{align}
where
\begin{align}\label{matrices}
    \bm{L}=\bm{C}_{\kappa }+\bm{G}_{\bm{H}},\quad [\bm{C}_\kappa]_{ij}=\br{\phi_i,\kappa ^2\phi_j}_{L^2(\Dd )}, \quad [\bm{G}_{\bm{H}}]_{ij}=\br{\nabla \phi_i, \bm{H}\nabla \phi_j}_{L^2(\Dd )} .
\end{align}
From \eqref{SPDE discrete} we deduce that the precision $Q$ of the weight vector $w$ is
\begin{align}\label{precision weights}
    \bm{Q}_u=(\bm{C}_\kappa +\bm{G}_{\bm{H}})^\dagger \bm{C}_\kappa^{-1} (\bm{C}_\kappa +\bm{G}_{\bm{H}})=\bm{C}_\kappa +2\bm{G}_{\bm{H}}+\bm{G}_{\bm{H}}\bm{C}_\kappa ^{-1} \bm{G}_{\bm{H}},
\end{align}
where we used that since $\bm{H}$ is symmetric, and we are considering real-valued functions, $\bm{C}_\kappa,\bm{G}_{\bm{H}}$ are symmetric.

The finite element basis is chosen so that $\bm{C}_\kappa, \bm{G}_{\bm{H}}$ are sparse. Since $\bm{C}_\kappa$ is not sparse, it is approximated by the mass lumped version
\begin{equation*}
    [\tl{\bm{C}}_\kappa]_{ij}:=\delta_{ij}\sum_{k=1}^n [\bm{C}_\kappa]_{ik}.
\end{equation*}
After this approximation we obtain a sparse precision $\tl{\bm{Q}}_u$ which allows us to sample efficiently from $\bm{u}$ (or rather, its approximation $\tl{u}\sim \Nn(\bm{0},\tl{\bm{Q}}_u)$) by using a Cholesky decomposition $\tl{\bm{Q}}_u=\bm{L}\bm{L}^T$ and solving $\bm{L}^Tz=u$ where $z\sim \Nn(\bm{0},\mathbf{I})$.

It remains to discuss how the integrals defining $\bm{C}_\kappa,\bm{G}_{\bm{H}}$ are calculated. Let $\Tt$ be the mesh of the domain. For each triangle $T\in \Tt$ we denote the average value of $\kappa , \bm{H}$ on the nodes of $T$ as $\kappa^2(T),\bm{H}(T)$ and approximate the integrals in \eqref{matrices} by
\begin{equation}\label{matrices calculation}
    [\bm{C}_\kappa]_{ij} \approx\sum_{T\in \Tt  }\kappa ^2(T)\int_{T}\phi_i \phi_j \,\mathrm{d}\bm{x}, \quad
    [\bm{G}_{\bm{H}}]_{ij}\approx\sum_{\substack{T\in \Tt  \\ k,l=1,2}}[\bm{H}(T)]_{kl}\int_{T}\partial_k\phi_i \partial_l\phi_j \,\mathrm{d}\bm{x}.
\end{equation}
Finally, we take $\phi_i$ to be piecewise linear, equal to $1$ on node $i$ and equal to $0$ on every other node. An explicit formula for the integrals in \eqref{matrices calculation} can be found in \citet[Appendix A2]{Lindgren2011AnEL}.
\subsection{Comparison of priors for the simulation study}\label{app:prior comparison}
For our simulation study in Section \ref{simulation section}, we compare four priors for the parameters $\kappa,\bm{v}$ of the SPDE \eqref{SPDE2D}.
\begin{enumerate}
  \item \label{option1} The PC priors in \eqref{pc1} where the anisotropy hyperparameters $\lambda_{\bm{\theta}} , \lambda_{\bm{v}}$ are chosen so that the anisotropy ratio $a=\exp(|\bm{v}|)$ and the correlation range $\rho = \sqrt{8}\kappa ^{-1} $ satisfy
     \begin{align}\label{quantiles simulation}
       \mathbb{P}[a> a_0]= 0.01, \quad \mathbb{P}[\rho < \rho_0]=0.01,
     \end{align}
 where we take $a_0=10, \rho_0=1$. This choice of hyperparameters corresponds to allowing with probability $0.01$ that the field is $10$ times more correlated in any given direction and, with the same probability, that the field has a correlation range smaller than $1$.  The values of $\lambda_{\bm{\theta}},\lambda_{\bm{v}}$ can then be calculated using \Cref{quantiles theorem}.
  \item\label{option2} Independent priors $\kappa \sim \textrm{Exp}(\lambda_\kappa )$ and $\bm{v} \sim \Nn(\bm{0},\sigma_{\bm{v}}^2 \bm{I}_2 ) $. Under these priors, $\kappa,\bm{v}$ have the same mode as the PC priors ($0$ in each case). Additionally, $\lambda_\kappa ,\sigma_{\bm{v}} ^2$ are chosen such that, under these priors, \eqref{quantiles simulation} also holds. We denote this prior by $\pi_{\rmm{EG}}$.
  \item Independent (improper) uniform priors for $\log(\kappa),\bm{v}$ with infinite support. We denote this prior by $\pi_{\rmm{U}}$.
  \item Independent linear transformations of beta priors on $\log(\kappa),v_1,v_2$ with shape parameters $1.1$ such that the correlation range is supported in $[\rho_0/w, wL]$ and $v_1,v_2$ are supported in $[-wa_0,w a_0]$. The shape parameter is chosen so that the distribution is approximately uniform while having a smooth density, which is relevant for the optimization. The purpose of $w>1$ is to extend the support of the parameters past $\rho_0,a_0$. The same value of $\rho_0=1$ is taken, $L=10$ is taken to be the length of $\Dd$, and $w$ is set to $20$. We denote this prior by $\pi_\beta$.
\end{enumerate}
The remaining hyperparameters $\lambda_{\sigma_u}, \lambda_{\bm{\varepsilon }}$ are chosen so that
\begin{align*}
  \mathbb{P}[\sigma_u>\sigma_{0}]=0.01, \quad\mathbb{P}[\sigma_{\bm{\varepsilon }}>\sigma_1]=0.01,
\end{align*}
where we take $\sigma_0=10,\sigma_1=1.5$.

To better understand the four priors $\pi_\rmm{PC},\pi_\rmm{EG},\pi_\rmm{U},\pi_\beta$ and their corresponding posteriors, in  \Cref{fig:prior_posterior_plots_log_kappa} and \Cref{fig:prior_posterior_plots_v}, we fix $\sigma_u, \sigma_\varepsilon $ and plot each of the marginal prior and posterior densities mentioned above (up to a multiplicative constant) in $\log(\kappa)$ and $v$ respectively, for a given observation $\bm{y}$. As we can see, $\pi_\rmm{PC}$ and $\pi_\rmm{EG}$ give similar prior and (as a result) posterior densities, to the point where it is impossible to distinguish them in the plots. The uniform and beta priors and posteriors are similar to each other, differing mainly in the size of their support. The beta priors and posteriors resemble a cutoff version of their uniform analogs. This indicates that choosing $\pi_\rmm{PC}$ and $\pi_\rmm{EG}$ as priors will give similar results. Likewise, $\pi_\rmm{U}, \pi_\beta $ will also give similar results to each other. This is borne out in the results.
\begin{figure}[H]
    \centering
    \includegraphics[width=\textwidth]{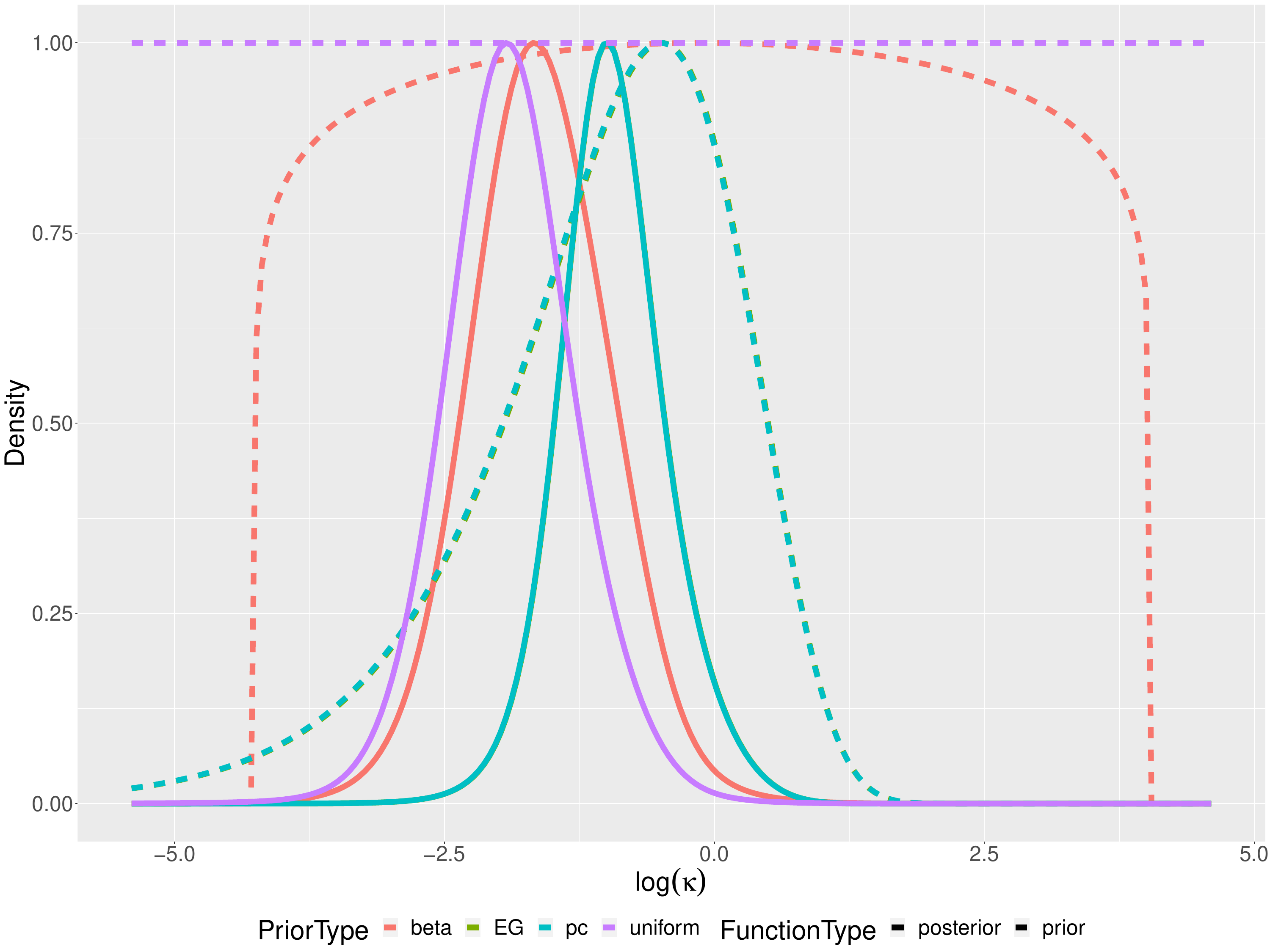}
    \caption{Renormalized marginal prior and posterior densities of $\log(\kappa)$ for each of the four priors $\pi_\rmm{PC},\pi_\rmm{EG},\pi_\rmm{U},\pi_\beta$ and for a given observation $\bm{y}$.}
    \label{fig:prior_posterior_plots_log_kappa}
\end{figure}
\begin{figure}[H]
    \centering
    \includegraphics[width=\textwidth]{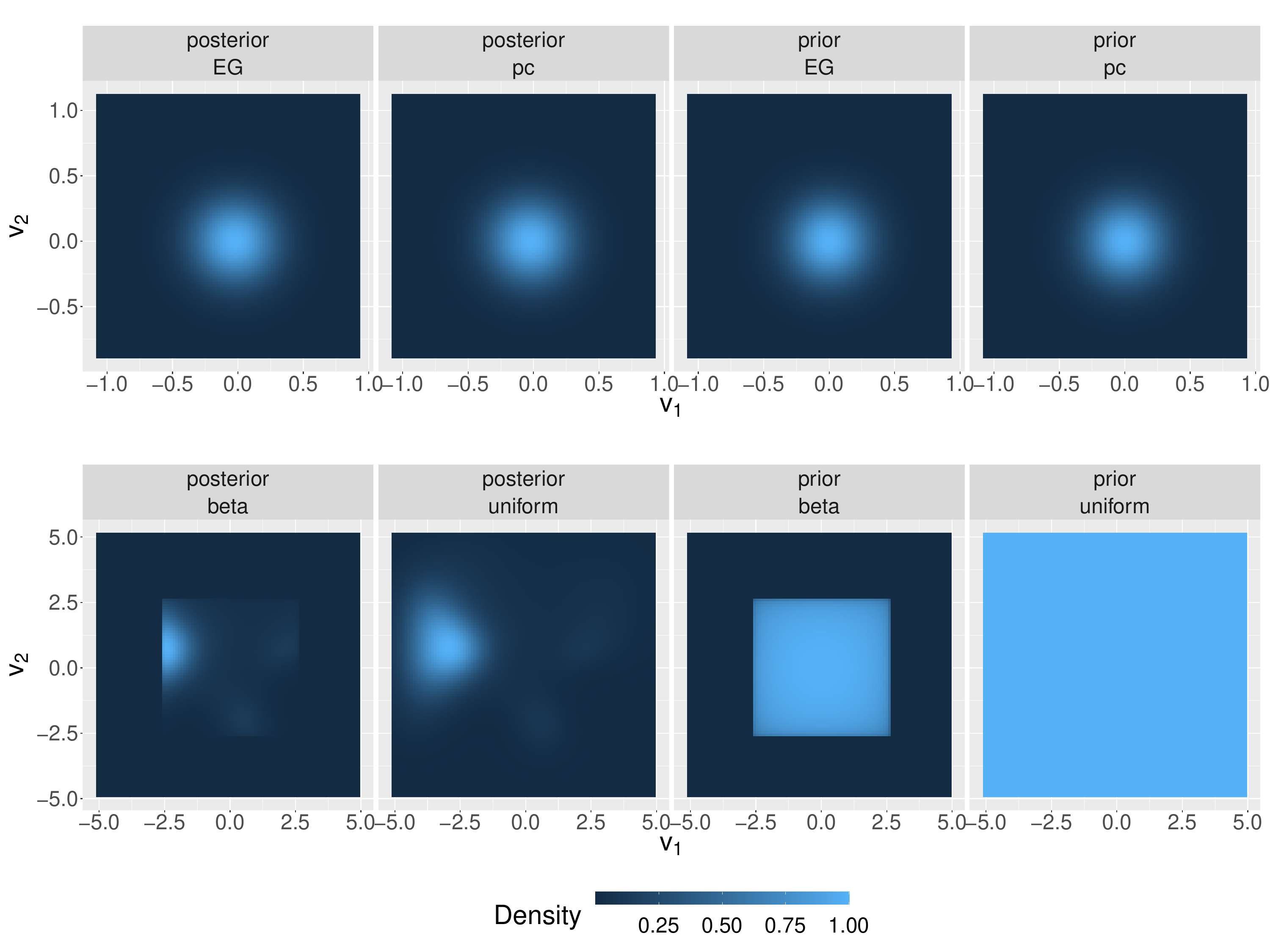}
    \caption{Marginal prior and posterior densities of $v$ for each of the four priors $\pi_\rmm{PC},\pi_\rmm{EG},\pi_\rmm{U},\pi_\beta$ and for a given observation $\bm{y}$}
    \label{fig:prior_posterior_plots_v}
\end{figure}

\subsection{Posterior sampling}\label{Posterior sampling appendix}
Returning to the simulation, our parameters are  $\bm{\theta} =(\log\kappa,\bm{v}, \log(\sigma_u),\log(\sigma_{\bm{\varepsilon}}))\in\mathbb{R}^5$. Bayes formula gives for any prior $\pi({\bm{\theta}})$ on $\bm{\theta}$ that
\begin{align}\label{posterior}
    \pi(\bm{\theta}  |\bm{y}) & =\frac{\pi(\bm{\theta} , \bm{y})}{\pi(\bm{y})}\frac{\pi\left(\bm{u} | \bm{\theta} ,\bm{y}\right)}{\pi\left(\bm{u} | \bm{\theta} ,\bm{y}\right)}=\frac{\pi(\bm{\theta} ) \pi(\bm{u} | \bm{\theta} ) \pi(\bm{y} | \bm{u},\bm{\theta})}{\pi(\bm{y}) \pi\left(\bm{u} | \bm{\theta}  ,\bm{y}\right)}.
\end{align}
We note that the right-hand side of \eqref{log posterior} involves $\bm{u}$. However, the resulting expression is independent of the value of $\bm{u}$ chosen.
As a result, the log posterior $\ell (\bm{\theta}  |\bm{y})$ is up to an additive constant given by
\begin{align}\label{log posterior}
    \tl{ \ell} (\bm{\theta}  |\bm{y}) = \ell (\bm{\theta}  )+\ell (\bm{u}|\bm{\theta}  )+\ell(\bm{y}|\bm{u},\bm{\theta})- \ell (\bm{u}|\bm{\theta}  ,\bm{y})
\end{align}
\begin{observation}
    The maximum of the right-hand side of \eqref{log posterior} is independent of $\bm{u}$. However, in the case where $\ell (\bm{u}|\bm{\theta},\bm{y})$ were not known, the value of $\bm{u}$ would affect the result. If $\bm{u}|\bm{\theta},\bm{y}$ were approximated to be Gaussian, it would make the most sense to take $\bm{u}$ as the mode $\bm{m_u}$ 	of this Gaussian approximation.
\end{observation}
In our case, all terms in \eqref{log posterior} are known exactly. The log prior $\ell(\bm{\theta})$ is determined by \eqref{PC priors} and the remaining terms are given by
\begin{align*}
    \ell(\bm{u}|\bm{\theta}  )        & =\frac{1}{2}\qty(\log\norm{\bm{Q}_u}-\norm{\bm{u}-\bm{m}_u}_{\bm{Q}_u}^2-n\log(2\pi)) \notag                                                                                \\
    \ell(\bm{y}|\bm{u},\bm{\theta})   & =\frac{1}{2}\qty(\log\norm{\bm{Q}_{\bm{\varepsilon}}} -\norm{\bm{y}-\bm{A}{\bm{u}}}_{\bm{Q}_{\bm{\varepsilon}}}^2-m\log(2\pi))\notag                                        \\
    \ell(\bm{u}|\bm{\theta}  ,\bm{y}) & =\frac{1}{2}\qty(\log\norm{\bm{Q}_{\bm{u}|\bm{y},\bm{\theta}  }} -\norm{\bm{u}- \bm{m}_{\bm{u}|\bm{y},\bm{\theta }}}_{\bm{Q}_{\bm{u}|\bm{y},\bm{\theta}  }}^2-n\log(2\pi)),
\end{align*}
where we defined $\norm{\bm{x}}_{\bm{Q}}^2:= \bm{x}^T\bm{Q} \bm{x}$ and $m$ is the dimension of $\bm{y}$. The precision  $\bm{Q}_u$ of $\bm{u}$ is calculated by FEM and is given by  \eqref{precision weights} divided by $4\pi \sigma_u^2$. The remaining precision matrices are
\begin{equation}\label{posterior precision u}
    \bm{Q}_{\bm{u}|\bm{y},\bm{\theta}  }:= \bm{Q}_u+ \bm{A}^T \bm{Q}_\varepsilon \bm{A},\quad \bm{m}_{\bm{u}|\bm{y},\bm{\theta} }:=\bm{Q}_{\bm{u} | \bm{y},\bm{\theta}}^{-1}\left(\bm{Q}_u \bm{m}_u+\bm{A}^T \bm{Q}_{\bm{\varepsilon}} \bm{y}\right).
\end{equation}
Maximizing the expression in \eqref{log posterior} is made possible by the fact that the precision matrices are sparse. We obtain the MAP estimate
\begin{align}\label{MAP def}
    \wh{\bm{\theta}  }:= \arg\max_{\bm{\theta}  \in \R^3} \ell (\bm{\theta}  |\bm{y})=\arg\max_{\bm{\theta}  \in \R^3} \tl{\ell} (\bm{\theta}  |\bm{y}).
\end{align}
Using this, we do the following simulations. We fix as our domain the square $[-1,1]^2$ from which we sample uniformly and independently $m=15$ locations $\set{\bm{x}_j}_{j=1}^m$ and use a mesh with $n=2062$ degrees of freedom. Then, we do the following.

For $\pi_\rmm{sim} \in \{\pi_\rmm{PC}, \pi_\rmm{EG}, \pi_\rmm{U}, \pi_\beta\}$:
\begin{enumerate}
    \item For $j= 1,...,N=200$:
          \begin{enumerate}
              \item \label{generation point} We simulate  ${\bm{\theta }^{(j)}}^{\rmm{true}}$ from  $\pi_\rmm{sim}$, use the FEM to simulate $\bm{u}^{(j)}$ from \eqref{SPDE2D} and then simulate $\bm{y}^{(j)}$ from \eqref{obs}.
              \item \label{prior point}For prior $\pi_{\mathrm{est}} \in \{\pi_\rmm{PC}, \pi_\rmm{EG}, \pi_\rmm{U}, \pi_\beta\}$ on $\bm{\theta }$:
                    \begin{enumerate}
                        \item We calculate a maximum a posteriori estimate $\wh{\bm{\theta}}^{(j)}$ using \eqref{MAP def}.
                        \item We calculate $\abs{{\theta^{(j)}}^{\rmm{true}}_i -\wh{\theta}_i }$.
                        \item We calculate a Gaussian approximation to the posterior distribution. To do so, we  write
                              \begin{align}\label{gaussian approx}
                                  \bm{Z} \sim\Nn \qty(\wh{\bm{\theta }}^{(j)}, \bm{\Sigma }_{\wh{\theta }|\bm{y}}),
                              \end{align}
                              where the precision, covariance, and marginal standard deviations of $\bm{Z}$ are, respectively
                              \begin{equation*}
                                  \qty[\bm{M}_{\wh{\bm{\theta }}^{(j)}|\bm{y}^{(j)}}]_{ij}=-\qt{\frac{\partial^2\ell\qty(\wh{\bm{\theta}}|\bm{y}^{(j)})}{\partial\theta_i\partial\theta_j}},\quad
                                  \bm{\Sigma }_{\wh{\bm{\theta }}^{(j)}|\bm{y}^{(j)}}=\bm{M}_{\wh{\bm{\theta }}^{(j)}|\bm{y}^{(j)}}^{-1}.
                              \end{equation*}
                              We write $\pi_{\bm{Z}}(\bm{\theta })$ for the density of $\bm{Z}$.
                              If the posterior distribution is Gaussian, then $\pi_{\bm{Z}}=\pi (\bm{\theta }|\bm{y}^{(j)})$.
                        \item We approximate the posterior measure. This step is necessary as we only have access to the unnormalized $\tl{\ell }$ (see \eqref{log posterior}). Three options were considered: approximating the posterior by the Gaussian $Z$ with the same median and whose precision is the negative of the Hessian of the posterior at its median, using importance sampling with $Z$ as the proposal distribution, and using smoothed importance sampling with $Z$ as the proposal distribution\citep{vehtari2015pareto}. Of these three approximations, the best performing method was \emph{smoothed importance sampling}. This can be measured by comparing the frequency of times the true parameter is within the $0.95$ confidence intervals. As a result, smoothed importance sampling is the method used. 

                            For the importance sampling, we use $\bm{Z}$ as our proposal distribution. Let $\{\bm{\theta}^{(s)} \sim \bm{Z},s=1,..., S\}$ be i.i.d. We denote the normalized, unnormalized, and self-normalized importance ratios by
                              \begin{align*}
                                  r(\bm{\theta })= \frac{\pi(\bm{\theta }|\bm{y}^{(j)})}{\pi_{\vb Z}(\bm{\theta })}, \quad  \tl{r}(\bm{\theta })= \frac{\exp(\tl{\ell} (\bm{\theta }))}{\pi_{\vb Z}(\bm{\theta })}, \quad \overline{r}(\bm{\theta}^{(s)}):= \frac{\tl{r} (\bm{\theta}^{(s)})}{\sum_{s=1}^S \tl{r} (\bm{\theta}^{(s)})}.
                              \end{align*}
                              Since we do not have access to the normalizing constant of $\pi (\bm{\theta }|\bm{y}^{(j)})$, it is the last two ratios that we use. Using $\overline{r}(\bm{\theta}^{(s)})$ calculate the normalized smoothed weights $w(\bm{\theta}^{(s)})$ and approximate
                              \begin{align*}
                                  \pi (\bm{\theta }|\bm{y}^{(j)}) \approx \pi_{\rmm{IS}} \qty(\bm{\theta }|\bm{y}^{(j)}):=  \sum_{s=1}^S w\qty(\bm{\theta}^{(s)}) \delta_{\bm{\theta}^{(s)}}(\bm{\theta }) .
                              \end{align*}
                              In the case where Pareto smoothing is not used, we directly take $w(\bm{\theta}^{(s)})=\overline{r}(\bm{\theta}^{(s)})$.
                              With this framework, the expect value of a function $f(\bm{\theta })$ can be approximated as
                              \begin{align*}
                                  \mathbb{E}\qty[f(\bm{\theta })|\bm{y}^{(j)}] \approx \mathbb{E}_{\rmm{IS}}\qty[f(\bm{\theta })|\bm{y}^{(j)}]:= \sum_{s=1}^S w\qty(\bm{\theta}^{(s)}) f(\bm{\theta}^{(s)}).
                              \end{align*}
                              Additionally, we will be interested in the Kullback-Leibler divergence between the true posterior and the Gaussian approximation to the posterior. Here, $f(\bm{\theta })=\log(r\qty(\bm{\theta }|\bm{y}^{(j)}))$ is known up to a normalizing constant, but we can approximate it using the self-normalized importance ratios as follows
                              \begin{align*}
                                  \mathrm{KLD}\qty(\pi(\cdot|\bm{y}^{(j)})||\pi_{\bm{Z }}) & = \int_{\R^5 } \pi(\bm{\theta }|\bm{y}^{(j)}) \log({r} (\bm{\theta }))\d \bm{\theta } \sim \sum_{s=1}^S \overline{r}(\bm{\theta}^{(s)})  \log(S\overline{r} (\bm{\theta}^{(s)})) \\
                                                                                           & \sim \sum_{s=1}^{S} w(\bm{\theta}^{(s)}) \log(Sw (\bm{\theta}^{(s)})),
                              \end{align*}
                              where the first equality is the definition of the Kullback-Leibler divergence and the first and second approximations are the self-normalized importance sampling step ($r \sim S \overline{r} $ ) and smoothing step, respectively. This can be seen to converge to the true Kullback-Leibler divergence as $S\to \infty$ by using that $\tl{r}=\E_{Z}\qb{\tl{r} } r$ and the central limit theorem.
                        \item Using $\pi_{\rmm{IS}} (\bm{\theta }|\bm{y}^{(j)}) $ we approximate the mean, covariance, and KL divergence and build a $0.95$ credible interval $I$ for $\bm{\theta }$. We also check whether the true parameter ${\bm{\theta }^{(j)}}^\rmm{true} $ is in $I$ and compare these results against the Gaussian approximation to the posterior in \eqref{gaussian approx}.
                              \label{recovery point}
                        \item We calculate $a_i=\mathbb{P}_{\mathrm{est}}\qty[\theta _i\leq{\theta_i^{(j)}}^{\rmm{true}}|\bm{y}^{(j)}]$ for $i=1,\ldots,5$ both in the case where $\bm{\theta }$ is sampled from $\pi_{\rmm{IS}} (\bm{\theta }|\bm{y}^{(j)}) $ and from the Gaussian approximation to the posterior using prior $\pi_{\mathrm{est}}$. We then plot the empirical cumulative distribution function of $a_i$ and compare it to the uniform distribution. If the implementation is well calibrated, then $a_i\sim U(0,1)$ \citep{talts2018validating,modrak2023simulation}.\label{predictive point}
                    \end{enumerate}
          \end{enumerate}
\end{enumerate}
The above was attempted using the Gaussian distribution $Z$ to replace the posterior rather than employing importance sampling. The method that uses no smoothing in the importance sampling was also tested. However, both methods yielded worse results than Pareto smoothed importance sampling, and we only display this approximation in the results.
Using \Cref{recovery point}, we can check whether the posterior distribution for $\bm{\theta }$  is accurate, and using \Cref{predictive point}, we check whether the implementation is well-calibrated.

\begin{observation}
    Since the uniform priors give support to extreme values of $\kappa,\bm{v}$, sampling the true parameter $\bm{\theta}^{\rmm{true}}$ from the uniform prior often leads to singularities in the calculation of the importance samples. These cases made up a high percentage of the simulations (in one case, $90\%$), for which reason the improper uniform priors were not used to sample from $\bm{\theta}^{\rmm{true}}$ in the simulations below. Rather, the width of the beta priors was also reduced when simulating these parameters to avoid singularities.
\end{observation}
The results for the MAP estimate and the CI lengths were already presented for the anisotropy parameters $\log(\kappa),\bm{v}$ in \Cref{simulation section} \Cref{fig: MAP distances} and \Cref{fig: CI lengths}. Here, we begin by presenting these results for the remaining parameters $\log(\sigma_u),\log(\sigma_{\bm{\varepsilon}})$ in \Cref{fig: MAP distances sigma} and \Cref{fig: CI lengths sigma}.
\begin{figure}[H]
    \centering
    \includegraphics[width=\textwidth]{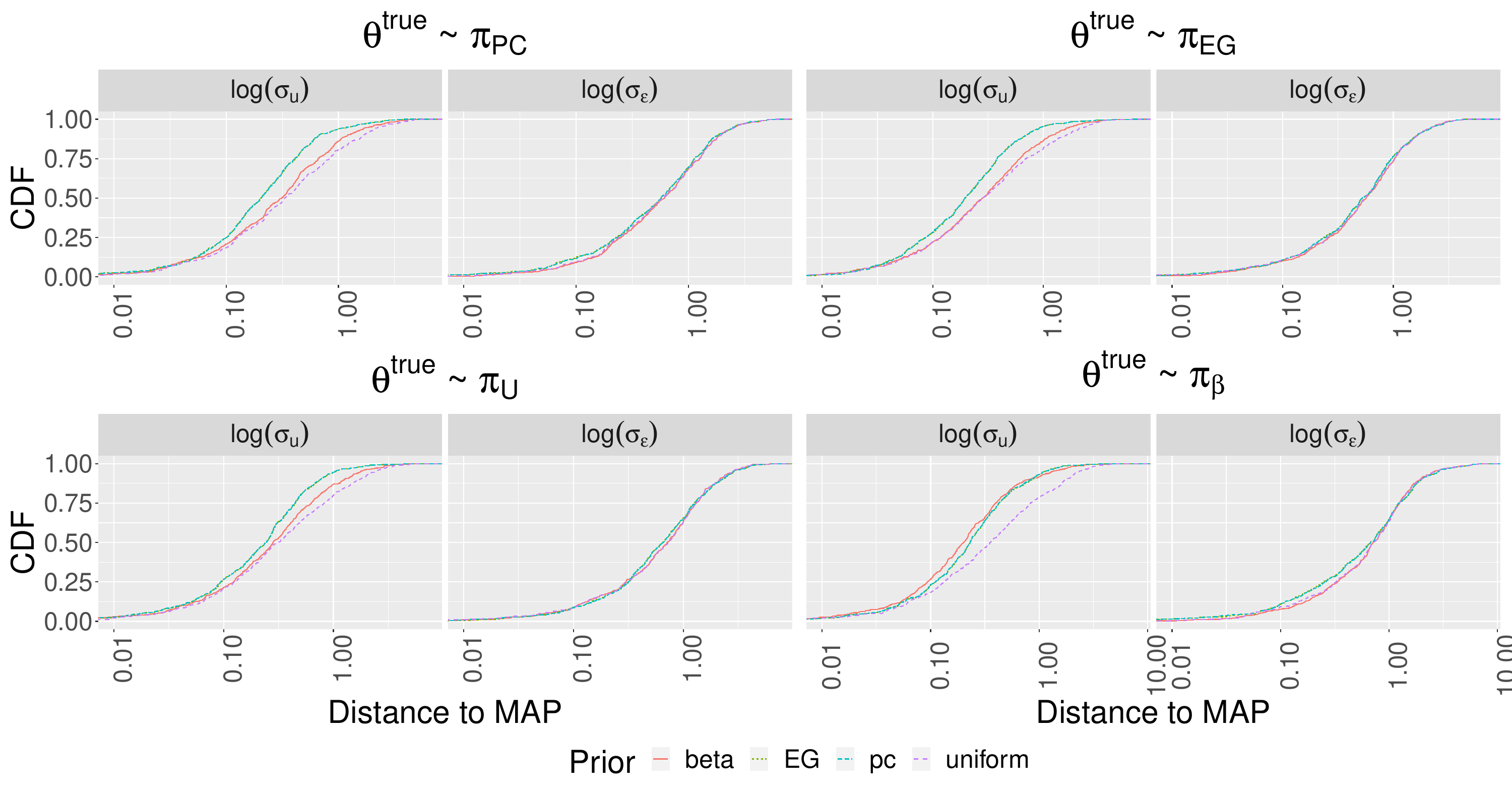}
    \caption{Empirical cumulative distribution function (eCDF) of the absolute distances between true parameter value ($\theta_i^{\text{true}}$) and their estimated values ($\widehat{\theta}_i$) for each of $\log(\sigma_u)$ and $\log(\sigma_\epsilon)$. In the plots, arranged from left to right and top to bottom, $\bm{\theta}^{\text{true}}$ is simulated from the four distributions—$\pi_{\text{PC}}$, $\pi_{\text{EG}}$, $\pi_{\text{U}}$, and $\pi_{\beta}$. Then, $\bm{y} = \bm{A} \bm{u} + \bm{\varepsilon}$ is observed with true parameter value $\bm{\theta}^{\text{true}}$. Finally, for each prior (red, green, teal, purple), the MAP estimate $\widehat{\bm{\theta}}$ is computed, and the eCDF over 600 simulations of the distances between $\bm{\theta}^{\text{true}}$ and $\widehat{\bm{\theta}}$ is plotted.}
    \label{fig: MAP distances sigma}
\end{figure}

\begin{figure}[H]
    \centering
    \includegraphics[width=\textwidth]{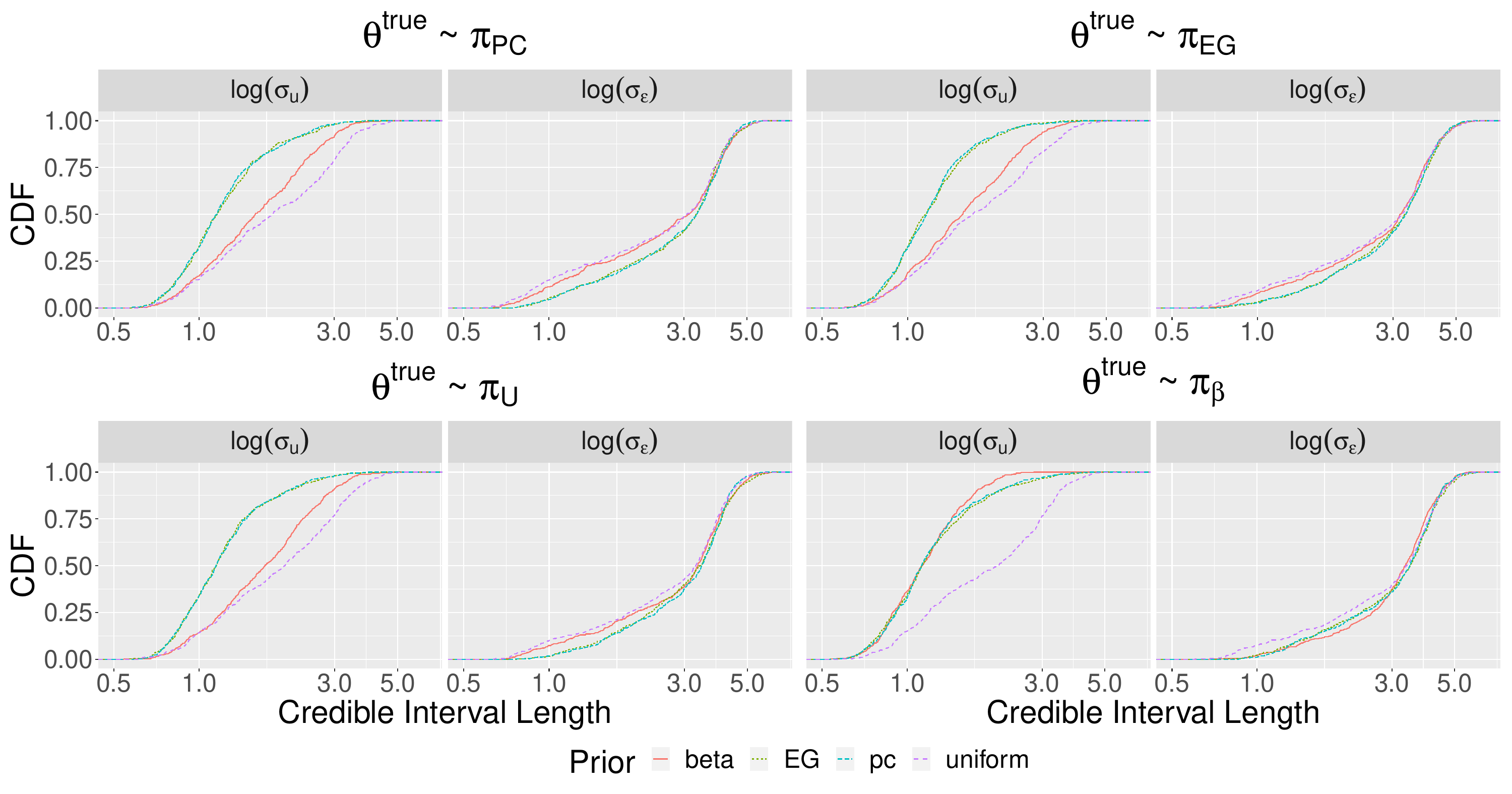}
    \caption{Empirical cumulative distribution function (eCDF) of the length of symmetric equal-tailed $0.95$ credible intervals (CIs) for the posterior of the parameters $\log(\sigma_u)$ and $\log(\sigma_\epsilon)$. In the plots, arranged from left to right and top to bottom, $\bm{\theta}^{\text{true}}$ is simulated from the four distributions—$\pi_{\text{PC}}$, $\pi_{\text{EG}}$, $\pi_{\text{U}}$, and $\pi_{\beta}$. Then, $\bm{y} = \bm{A} \bm{u} + \bm{\varepsilon}$ is observed with true parameter value $\bm{\theta}^{\text{true}}$.}
    \label{fig: CI lengths sigma}
\end{figure}

In \Cref{fig: within CI}, we show the frequency of the true parameter being in the $0.95$ credible interval. As can be seen in the graphs, the performance depends on the distribution from which $\bm{\theta }^{\mathrm{true}}$ is sampled. If $\bm{\theta }^{\mathrm{true}}$ is sampled from $\pi_{\mathrm{PC}}$ or $\pi_{\mathrm{EG}}$ then these two priors perform better, whereas otherwise $\pi_{\mathrm{U}},\pi_\beta $ perform better. We also observe that the confidence intervals for $\sigma_{\bm{\varepsilon}}$ are not wide enough. This may be due to the importance sampling.
The results for the remaining parameters $\log(\sigma_u),\log(\sigma_{\bm{\varepsilon}})$ are more similar, which is explained by the fact that the same priors for these parameters are used in all cases. This is a common theme in the results, and we will not comment on it further.
\begin{figure}[H]
    \centering
    \includegraphics[width=\textwidth]{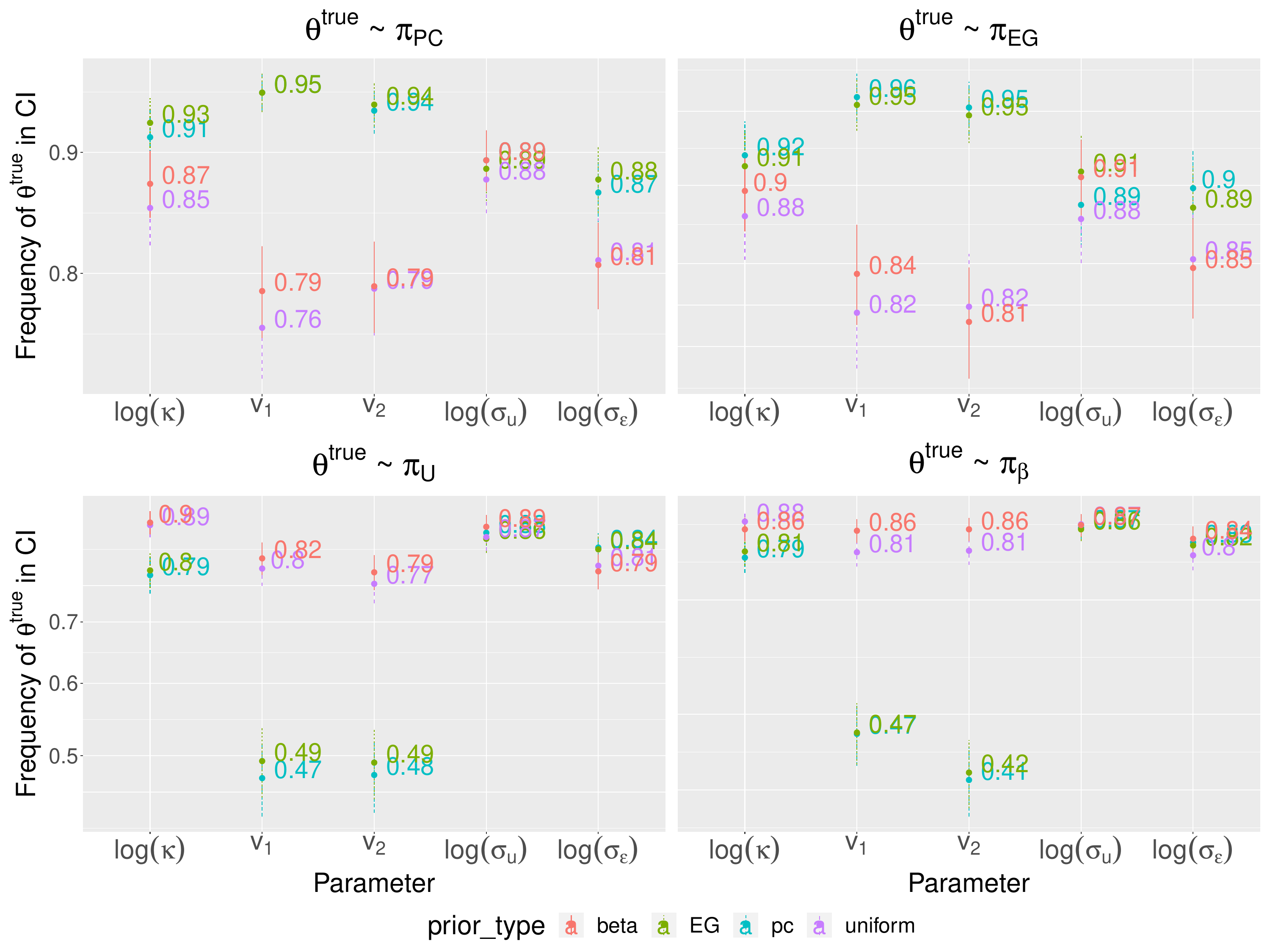}
    \caption{Frequency of true parameter being in the $0.95$ credible interval}
    \label{fig: within CI}
\end{figure}
In \Cref{fig: probabilities}, we show the eCDF of probabilities $a_i=\mathbb{P}[\theta_i\leq\theta_i^{\rmm{true}}|\bm{y}]$ for each of the five parameters. If the implementation is well calibrated, then $a_i$ should be uniform when $\pi_{\rmm{sim}}=\pi$. Once more, the performance is divided into two cases. If $\bm{\theta }^\rmm{true}$ is sampled from $\pi_{\rmm{PC}}$ or $\pi_{\rmm{EG}}$ then these two priors perform better, whereas otherwise $\pi_{\rmm{U}},\pi_\beta $ perform better.

\begin{figure}[H]
    \centering
    \includegraphics[width=\textwidth]{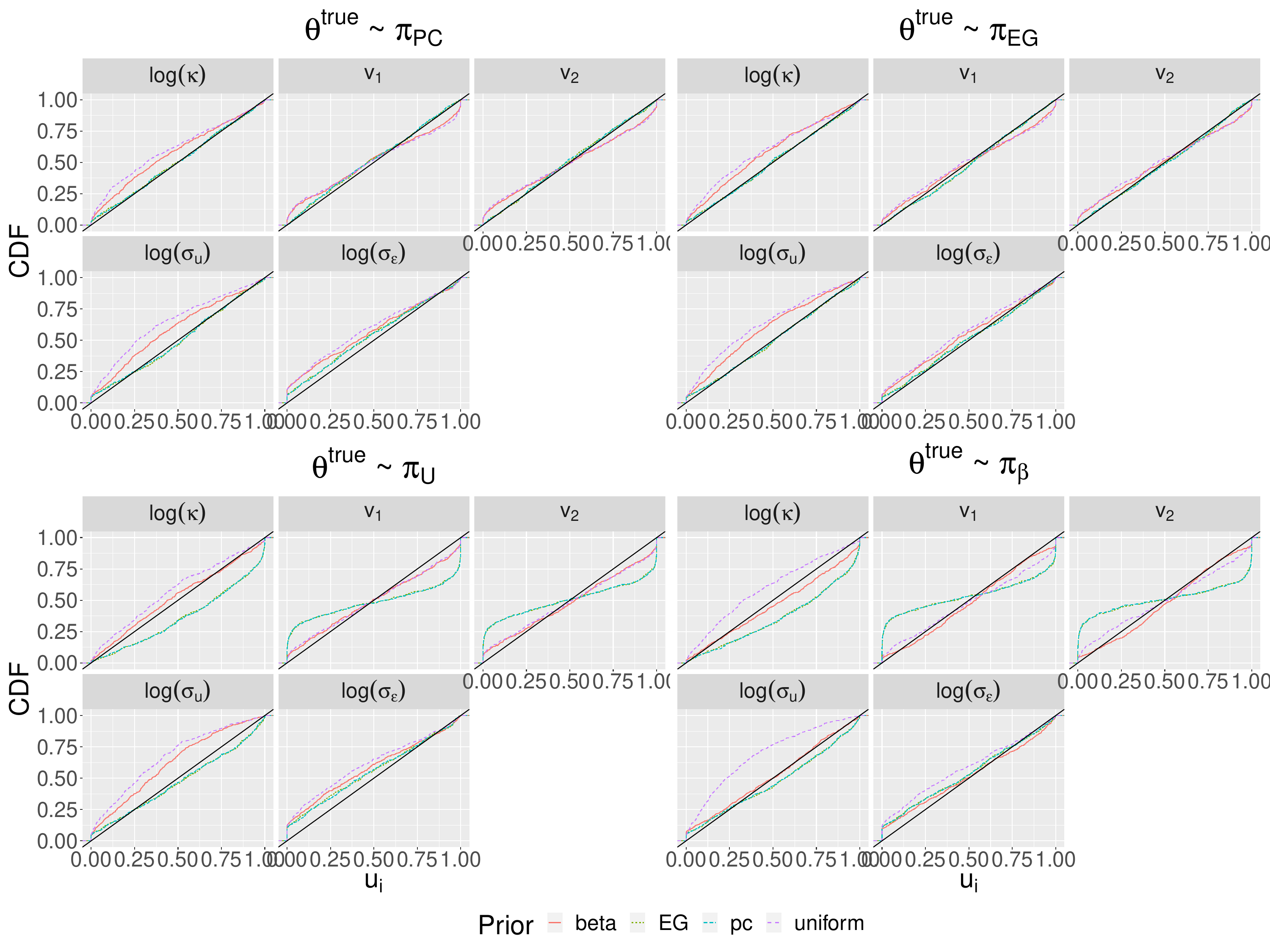}
    \caption{eCDF of probabilities $a_i=\mathbb{P}_{\mathrm{est}}[\theta_i\leq\theta_i^{\rmm{true}}|\bm{y}]$ where $\bm{\theta}^{\rm{true}}$ and $\bm{y}$ come from the true model. }
    \label{fig: probabilities}
\end{figure}

In  \Cref{tab:p_values_combined}, we show the $p$-values for the KS test between $a_i=\mathbb{P}_{\rm{est}}[\theta_i\leq\theta_i^{\rmm{true}}|\bm{y}]$ and the CDF of the uniform distribution.

\begin{table}[H]
    \centering
    \begin{minipage}{.49\textwidth}
        \centering
        \resizebox{\textwidth}{!}{%
            \begin{tabular}{lrrrrr}
                \hline
                Prior               & $\log(\kappa)$ & $v_1$    & $v_2$    & $\log(\sigma_u)$ & $\log(\sigma_{\epsilon})$ \\
                \hline
                $\pi_{\mathrm{PC}}$ & 3.72e-01       & 2.34e-01 & 6.60e-01 & 1.17e-01         & 6.240e-03                 \\
                $\pi_{\mathrm{EG}}$ & 5.55e-01       & 2.55e-01 & 8.37e-01 & 1.05e-01         & 2.720e-03                 \\
                $\pi_{\mathrm{U}}$  & 0.00e+00       & 9.42e-08 & 3.07e-06 & 0.00e+00         & 1.020e-12                 \\
                $\pi_{\beta}$       & 1.14e-08       & 1.44e-05 & 4.77e-05 & 1.71e-10         & 9.070e-10                 \\
                \hline
            \end{tabular}
        }
        \caption{$\theta^{\mathrm{true}} \sim \pi_{\mathrm{PC}}$}
        \label{tab:p_values_pc}
    \end{minipage}
    \hfill
    \begin{minipage}{.49\textwidth}
        \centering
        \resizebox{\textwidth}{!}{%
            \begin{tabular}{lrrrrr}
                \hline
                Prior               & $\log(\kappa)$ & $v_1$    & $v_2$    & $\log(\sigma_u)$ & $\log(\sigma_{\epsilon})$ \\
                \hline
                $\pi_{\mathrm{PC}}$ & 7.46e-01       & 1.96e-01 & 5.57e-01 & 1.47e-01         & 2.060e-01                 \\
                $\pi_{\mathrm{EG}}$ & 8.42e-01       & 2.05e-01 & 6.62e-01 & 3.17e-01         & 1.380e-01                 \\
                $\pi_{\mathrm{U}}$  & 3.77e-15       & 1.63e-05 & 6.85e-04 & 0.00e+00         & 1.100e-06                 \\
                $\pi_{\beta}$       & 2.25e-09       & 4.78e-04 & 1.60e-03 & 1.88e-12         & 8.440e-04                 \\
                \hline
            \end{tabular}
        }
        \caption{$\theta^{\mathrm{true}} \sim \pi_{\mathrm{EG}}$}
        \label{tab:p_values_eg}
    \end{minipage}
    \vfill
    \begin{minipage}{.49\textwidth}
        \centering
        \resizebox{\textwidth}{!}{%
            \begin{tabular}{lrrrrr}
                \hline
                Prior               & $\log(\kappa)$ & $v_1$    & $v_2$    & $\log(\sigma_u)$ & $\log(\sigma_{\epsilon})$ \\
                \hline
                $\pi_{\mathrm{PC}}$ & 0.00e+00       & 0.00e+00 & 0.00e+00 & 2.55e-04         & 1.73e-06                  \\
                $\pi_{\mathrm{EG}}$ & 0.00e+00       & 0.00e+00 & 0.00e+00 & 5.69e-05         & 3.35e-07                  \\
                $\pi_{\mathrm{U}}$  & 2.70e-09       & 1.51e-03 & 6.95e-06 & 0.00e+00         & 0.00e+00                  \\
                $\pi_{\beta}$       & 2.78e-02       & 1.25e-04 & 3.59e-05 & 0.00e+00         & 2.86e-12                  \\
                \hline
            \end{tabular}
        }
        \caption{$\theta^{\mathrm{true}} \sim \pi_{\mathrm{U}}$}
        \label{tab:p_values_u}
    \end{minipage}
    \hfill
    \begin{minipage}{.49\textwidth}
        \centering
        \resizebox{\textwidth}{!}{%
            \begin{tabular}{lrrrrr}
                \hline
                Prior               & $\log(\kappa)$ & $v_1$    & $v_2$    & $\log(\sigma_u)$ & $\log(\sigma_{\epsilon})$ \\
                \hline
                $\pi_{\mathrm{PC}}$ & 0.00e+00       & 0.00e+00 & 0.00e+00 & 7.26e-04         & 3.15e-06                  \\
                $\pi_{\mathrm{EG}}$ & 0.00e+00       & 0.00e+00 & 0.00e+00 & 1.54e-03         & 6.40e-07                  \\
                $\pi_{\mathrm{U}}$  & 7.88e-15       & 8.50e-06 & 3.35e-04 & 0.00e+00         & 3.73e-13                  \\
                $\pi_{\beta}$       & 1.06e-04       & 8.80e-03 & 2.57e-02 & 4.48e-03         & 5.06e-05                  \\
                \hline
            \end{tabular}
        }
        \caption{$\theta^{\mathrm{true}} \sim \pi_{\beta}$}
        \label{tab:p_values_beta}
    \end{minipage}
    \caption{KS statistic between $a_i = \mathbb{P}_{\rm{est}}[\theta_i \leq \theta_i^{\mathrm{true}} | \bm{y}]$ and the CDF of the uniform distribution}
    \label{tab:p_values_combined}
\end{table}

In \Cref{fig: KS bridge}, we show the empirical difference against $F(a_i)$ together with a 95\% confidence interval. As can be seen from the plots, the implementation is well-calibrated.
\begin{figure}[H]
    \centering
    \includegraphics[width=\textwidth]{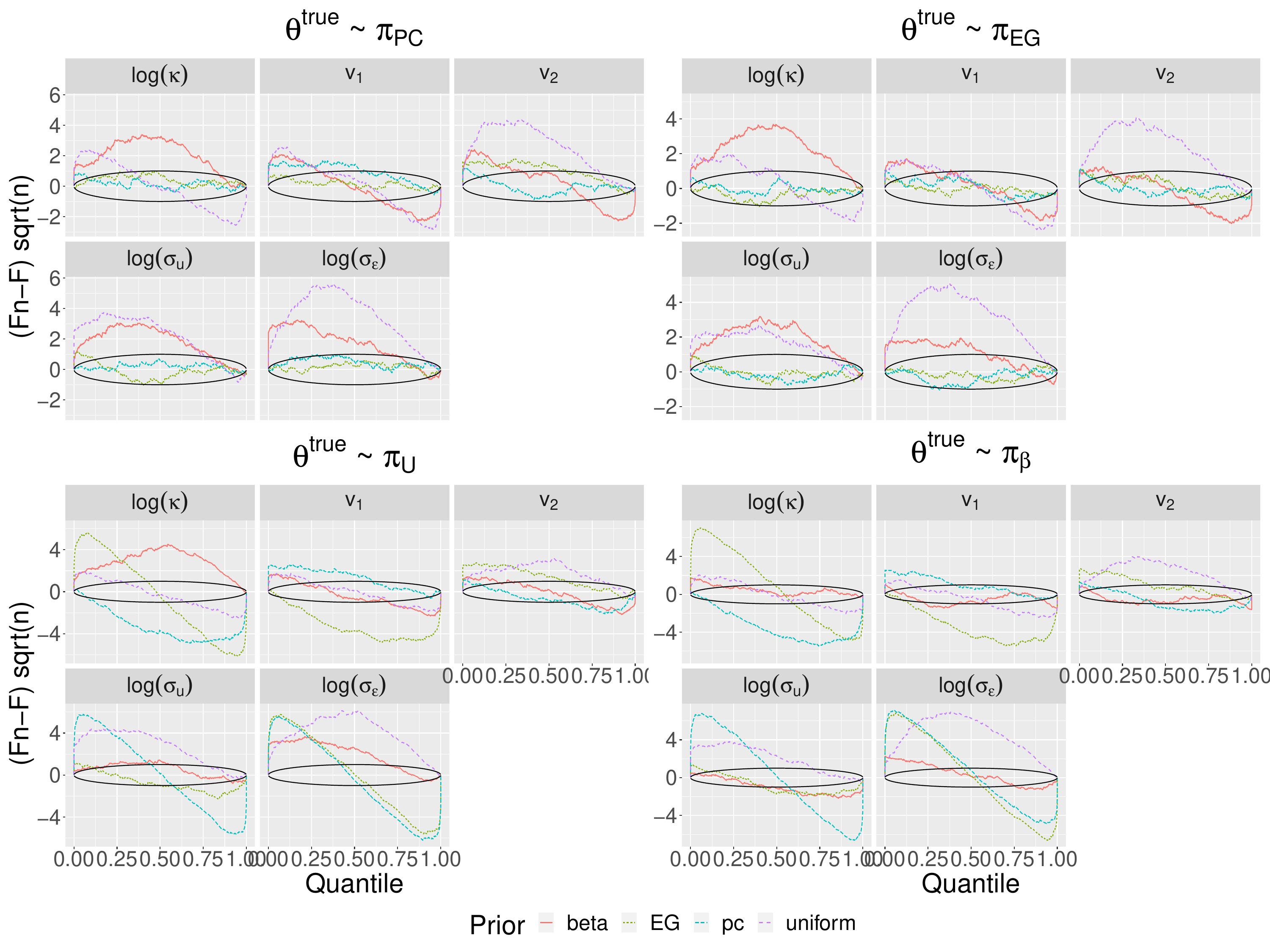}
    \caption{Plot of the empirical difference against $F(a_i)$ together with a 95\% confidence interval}
    \label{fig: KS bridge}
\end{figure}
In \Cref{fig: KL divergence}, we show the eCDF of the Kullback-Leibler divergence between the posterior and the Gaussian approximation to the posterior around the median. As can be seen from the plots, in all cases $\pi_\rmm{PC}(\cdot |\bm{y})$ and $\pi_\rmm{EG}(\cdot |\bm{y})$ are closer to their respective Gaussian approximations than $\pi_\rmm{U}(\cdot |\bm{y})$ and $\pi_\beta(\cdot |\bm{y})$. This is because the posteriors using $\pi_{\rmm{U}},\pi_\beta$ are very much not Gaussian.
\begin{figure}[H]
    \centering
    \includegraphics[width=\textwidth]{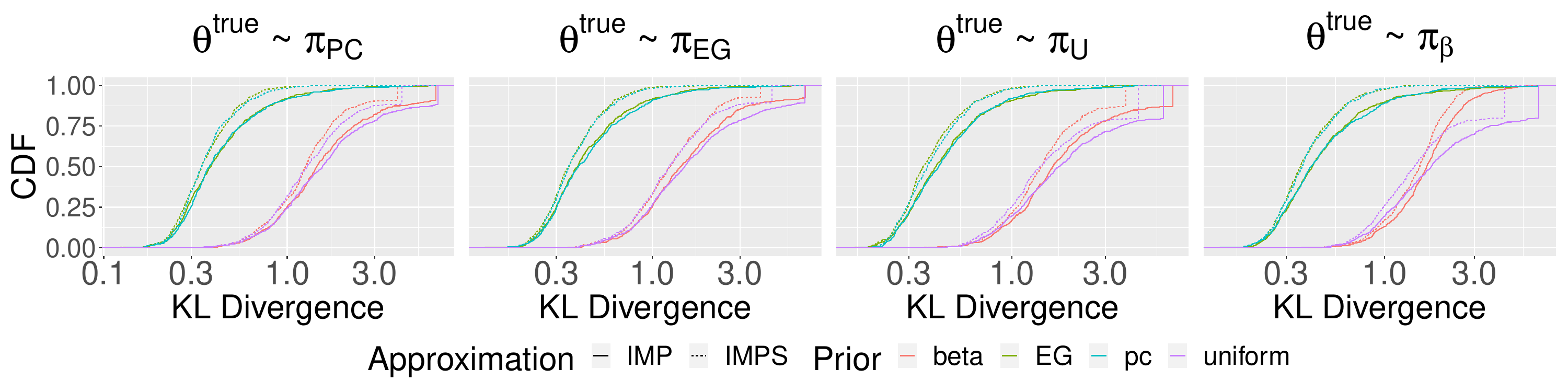}
    \caption{Kullback-Leibler divergence between the posterior and the Gaussian approximation to the posterior around the median.}
    \label{fig: KL divergence}
\end{figure}

To better compare PC and exponential-Gaussian priors, we also plot these separately. As can be seen, they behave quite similarly once more, showing that both formulations lead to practically identical penalization of complexity.
\begin{figure}[H]
    \centering
    \includegraphics[width=\textwidth]{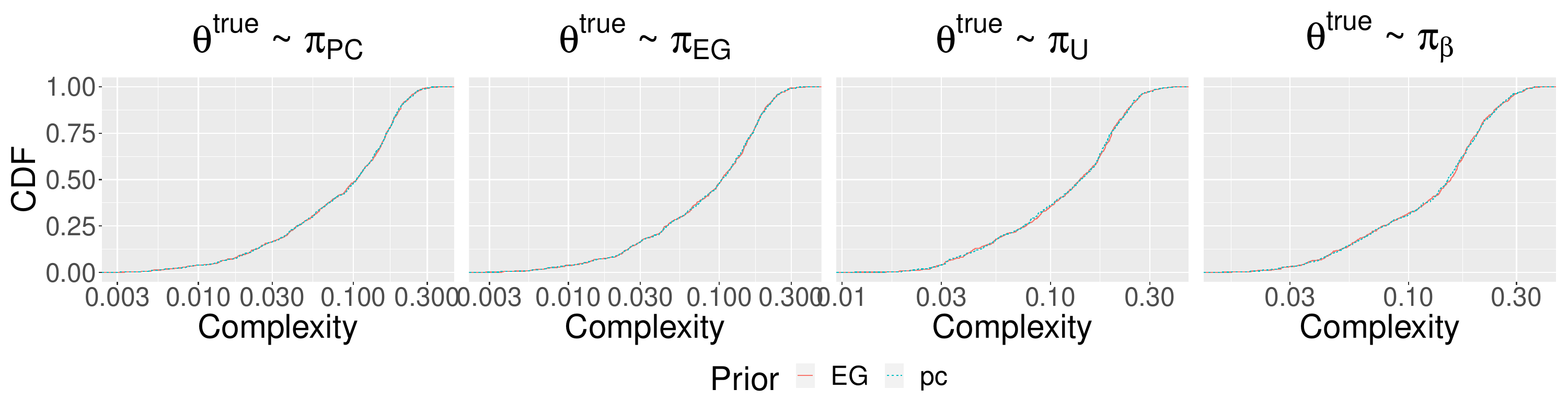}
    \caption{CDF of the posterior complexity $d(\kappa,\bm{v})$ for PC and not exponential-Gaussian priors}
    \label{fig: complexity12}
\end{figure}
\newpage
\section{Precipitation study details}\label{precipitation details appendix}
\subsection{Score definitions}\label{app:Score definitions}
The scores considered in \Cref{precipitation section} are
\begin{itemize}
  \item Given a distribution $\d F$ and an observation $x \in \R$ the \emph{squared error} (SE) is defined as
     \begin{align*}
       \mathrm{SE}(F,x) & := \left(x-\int_{-\infty}^\infty t\d F(t)\right)^2.
     \end{align*}
 The expression above defines a proper scoring rule (its value is minimized when $x$ follows the predictive distribution). However, it is not strictly proper (the minimum in $F$ is not unique).
  \item Given a CDF $F$ and an observation $x \in \R$, the \emph{continuous ranked probability score} (CRPS) is defined as \citep{gneiting2007strictly}
     \begin{align*}
       \rmm{CRPS}(F,x) & := \int_{\R} (F(t)-\bm{1}\{x \leq t\})^2 d t.
     \end{align*}
 The CRPS is a strictly proper scoring rule on distributions with finite expectations.
  \item Given a random variable $\bm{X}$ with mean $\bm{\mu }$ and precision $\bm{Q}$, the \emph{Dawid-Sebastiani score} (DSS) is defined as  \citep{dawid1999coherent}
     \begin{align*}
       \rmm{DSS}(\bm{X};\bm{x}) & := -\log(\norm{\bm{Q}} ) + (\bm{x} - \bm{\mu })^T\bm{Q}(\bm{x} - \bm{\mu }).
     \end{align*}
 The DSS is similar to a log-Gaussian density score and is a strictly proper scoring rule for Gaussian random variables. Still, it also defines a proper score for non-Gaussian random variables.
\end{itemize}

\subsection{Calculations}\label{calculations precipitation appendix}
In this section, we show how to calculate scores of \Cref{app:Score definitions}. For what follows, we first define the leave-one-out mean of $y_i$ having observed all locations except $y_i$
\begin{align}\label{posterior mean LOO prediction}
    m_{y_i|\bm{y}_{-i}} & :=  \int_{\R^5} \qty(  \int_{\R} y_i \pi (y_{i}|\bm{y_{-i}},\bm{\theta })\d y_{i}) \pi(\bm{\theta }| \bm{y}) \d \bm{\theta} = \int_{\R^5} m_{y_i|\bm{y}_{-i},\bm{\theta } }\pi (\bm{\theta }| \bm{y}) \d \bm{\theta }.
\end{align}
The expression of $m_{y_i|\bm{y}_{-i}}$ is approximately that of $ \E\qb{y_i|\bm{y}_{-i}}$, with the exception that the integral in \eqref{posterior mean LOO prediction} is against $\pi(\bm{\theta }| \bm{y})$ as opposed to $\pi(\bm{\theta }| \bm{y}_{-i})$. This is a common approach to calculating the LOO predictions as it avoids the need to approximate $\pi (\bm{\theta }| \bm{y}_{-i})$ for each $i$. The expression is justified by the fact that the posterior distribution of $\bm{\theta }$ does not depend much on any single observation. That is, $\pi (\bm{\theta }| \bm{y})\approx \pi (\bm{\theta }| \bm{y}_{-i})$. Similarly, we will need the posterior leave out one variance
\begin{align}\label{posterior variance LOO prediction}
    \sigma^2_{y_i|\bm{y}_{-i}} := \int_{\R^5} \sigma^2_{y_i|\bm{y}_{-i},\bm{\theta } }\pi (\bm{\theta }| \bm{y}) \d \bm{\theta },
\end{align}
where $\sigma^2_{y_i|\bm{y}_{-i},\bm{\theta }}$ is the posterior variance of $y_i$ having observed all locations except $y_i$. This will be used later in the computations.

Because of the independence of $\bm{u}, \bm{\beta}$ from $\bm{\theta }$  we have that
\begin{align*}
    \left. \bm{u}_{\bm{\beta }}\right| \bm{\theta } & \sim \Nn(\bm{0},\bm{Q}_{\bm{u}_{\bm{\beta }}}^{-1}), \quad \text{where} \quad  \bm{Q}_{\bm{u}_{\bm{\beta }}}:=\begin{bmatrix}
                                                                                                                                                                        \bm{Q}_{\bm{\beta }} & \bm{0}_{2 \times n} \\
                                                                                                                                                                        \\
                                                                                                                                                                        \bm{0}_{n \times 2}  & \bm{Q}_u
                                                                                                                                                                    \end{bmatrix}.
\end{align*}
We can then calculate $\wh{\bm{\theta }}$ to maximize the (unnormalized) log-posterior as in \eqref{log posterior}
\begin{align}\label{log posterior Norway}
    \tl{ \ell} (\bm{\theta}   |\bm{y}) = \ell (\bm{\theta}  )+\ell (\bm{u}_\beta |\bm{\theta}   )+\ell(\bm{y}|\bm{u}_\beta ,\bm{\theta} )- \ell (\bm{u}_\beta |\bm{\theta}  ,\bm{y}).
\end{align}
To calculate the RMSE, CRPS, and DSS, it is sufficient to calculate the posterior mean and variance of $y_i$ given $\bm{y}_{-i}$ and $\bm{\theta }$. We show how to do this efficiently.

Due to the independence of $\bm{u}$ and $\bm{\varepsilon }$ and of $\varepsilon _i$ and $\bm{\varepsilon }_{-i}$ knowing $\bm{\theta }$,  we deduce that $\bm{y}_{-i} = \bm{A}_{-i} \bm{u}+ \varepsilon_{-i}$ and $\varepsilon _i$ are independent knowing $\theta$. As a result,
\begin{aligneq}\label{y knowing traning}
    y_i | \bm{y}_{-i},\bm{\theta } & = (\bm{A}_i  \bm{u}) | \bm{y}_{-i},\bm{\theta } + \varepsilon _i | \bm{\theta } \sim \Nn (\bm{A}_i \bm{m}_{\bm{u}|\bm{y}_{-i},\bm{\theta }}, \bm{A}_i \bm{\Sigma }_{\bm{u}|\bm{y}_{-i},\bm{\theta }} \bm{A}_i^T + \sigma_{\bm{\varepsilon }}^2) \\
                                   & = \Nn (m_{y_i|\bm{y}_{-i},\bm{\theta }}, \sigma_{y_i|\bm{y}_{-i},\bm{\theta }}^2),
\end{aligneq}
where $\bm{A}_i \in \R^{1 \times N}$ is the $i$-th row of $\bm{A}$. Furthermore, we have that, see \eqref{posterior precision u}
\begin{aligneq}\label{posterior precision u training}
    \bm{Q}_{\bm{u}|\bm{y}_{-i}, \bm{\theta }} & = \bm{Q}_u + \sigma_\varepsilon ^{-2} \bm{A}_{-i}^T\bm{A}_{-i}= \bm{Q}_{\bm{u}|\bm{y}, \bm{\theta }} - \sigma_\varepsilon ^{-2} \bm{A}_{i}^T \bm{A}_i  \\
    \bm{m}_{\bm{u}|\bm{y}_{-i},\bm{\theta} }  & =\sigma_{\bm{\varepsilon }}^{-2} \bm{\Sigma }_{\bm{u} | \bm{y}_{-i},\bm{\theta}}\bm{A}_{-i}^T \bm{y}_{-i},
\end{aligneq}
where we wrote $\bm{A}_{-i} \in \R^{n-1 \times N}$ for the matrix $\bm{A}$ with the $i$-th row removed. To avoid calculating an inverse for each $i$, we can calculate  $\bm{\Sigma}_{u| \bm{y}, \bm{\theta }}$ once and then use the rank $1$ correction provided by the Sherman-Morrison formula to obtain
\begin{align}\label{sherman}
    \bm{\Sigma }_{\bm{u}|\bm{y}_{-i},\bm{\theta }} & = \bm{\Sigma }_{\bm{u}|\bm{y},\bm{\theta }} + \frac{\bm{\Sigma }_{\bm{u}|\bm{y},\bm{\theta }}\bm{A}_{i}^T\bm{A}_{i}
    \bm{\Sigma }_{\bm{u}|\bm{y},\bm{\theta }}}{\sigma_{\bm{\varepsilon }}^2 -\bm{A}_{i}\bm{\Sigma }_{\bm{u}|\bm{y},\bm{\theta }}\bm{A}_{i}^T}.
\end{align}

Using \eqref{sherman} in \eqref{y knowing traning} and writing $V_i :=\bm{A}_i \bm{\Sigma }_{\bm{u}|\bm{y},\bm{\theta }} \bm{A}_i^T$ and $q_{\bm{\varepsilon }}:= \sigma_{\bm{\varepsilon }}^{-2}$,
gives
\begin{align}\label{leave out variance}
    \sigma_{y_i|\bm{y}_{-i},\bm{\theta }}^2 & = V_i  + \frac{V_i^2}{\sigma_{\bm{\varepsilon }}^2 - V_i} + \sigma_{\bm{\varepsilon }}^2   = \frac{V_i}{1-q_{\bm{\varepsilon }}V_i}    +\sigma_{\bm{\varepsilon }}^2  =   \frac{\sigma_{\bm{\varepsilon }}^2 }{1-q_{\bm{\varepsilon }}V_i}.
\end{align}
Furthermore, for the mean, we have from \eqref{posterior precision u training} that
\begin{align}\label{start}
    m_{y_i|\bm{y}_{-i},\bm{\theta }} & =q_{\bm{\varepsilon }}  \bm{A}_i \bm{\Sigma }_{\bm{u} | \bm{y}_{-i},\bm{\theta}}\bm{A}_{-i}^T \bm{y}_{-i}= q_{\bm{\varepsilon }}  \bm{A}_i \bm{\Sigma }_{\bm{u} | \bm{y}_{-i},\bm{\theta}}\bm{A}^T \bm{y} - q_{\bm{\varepsilon }}  \bm{A}_i \bm{\Sigma }_{\bm{u} | \bm{y}_{-i},\bm{\theta}}\bm{A}_{i}^T {y}_i
\end{align}
Let us write
\begin{align*}
    \eta_i := \E\qb{\bm{A}_i \bm{u}| \bm{y},\bm{\theta }}= \bm{A}_i \bm{m}_{\bm{u}|\bm{y},\bm{\theta }}= q_{\bm{\varepsilon }}\bm{A}_i \bm{\Sigma }_{\bm{u}|\bm{y},\bm{\theta }}\bm{A}^T \bm{y},
\end{align*}
Then, using the rank $1$ correction in \eqref{sherman} we obtain for the first term in \eqref{start} that
\begin{aligneq}\label{start1}
    q_{\bm{\varepsilon }}\bm{A}_i \bm{\Sigma }_{\bm{u} | \bm{y}_{-i},\bm{\theta}}\bm{A}^T \bm{y} & =  q_{\bm{\varepsilon }}\bm{A}_i \bm{\Sigma }_{\bm{u} | \bm{y},\bm{\theta}}\bm{A}^T \bm{y} +q_{\bm{\varepsilon }} \bm{A}_i\frac{\bm{\Sigma }_{\bm{u}|\bm{y},\bm{\theta }}\bm{A}_{i}^T\bm{A}_{i}
    \bm{\Sigma }_{\bm{u}|\bm{y},\bm{\theta }}}{\sigma_{\bm{\varepsilon }}^2 -\bm{A}_{i}\bm{\Sigma }_{\bm{u}|\bm{y},\bm{\theta }}\bm{A}_{i}^T}\bm{A}^T\bm{y}                                                                                                                              \\& = \eta_i +  \frac{V_i}{\sigma_{\bm{\varepsilon }}^2 -V_i}\eta_i= \frac{\eta_i }{1-q_{\bm{\varepsilon }}V_i}.
\end{aligneq}
Whereas, for the second term in \eqref{start}, we have that
\begin{aligneq}\label{start2}
    q_{\bm{\varepsilon }}\bm{A}_i \bm{\Sigma }_{\bm{u} | \bm{y}_{-i},\bm{\theta}}\bm{A}_{i}^T {y}_i & =  q_{\bm{\varepsilon }}\bm{A}_i \bm{\Sigma }_{\bm{u} | \bm{y},\bm{\theta}}\bm{A}_{i}^T {y}_i +q_{\bm{\varepsilon }} \bm{A}_i\frac{\bm{\Sigma }_{\bm{u}|\bm{y},\bm{\theta }}\bm{A}_{i}^T\bm{A}_{i}
    \bm{\Sigma }_{\bm{u}|\bm{y},\bm{\theta }}}{\sigma_{\bm{\varepsilon }}^2 -\bm{A}_{i}\bm{\Sigma }_{\bm{u}|\bm{y},\bm{\theta }}\bm{A}_{i}^T}\bm{A}_{i}^T{y}_i                                                               \\& = q_{\bm{\varepsilon }} V_i y_i + q_{\bm{\varepsilon }} \frac{V_i^2}{\sigma_{\bm{\varepsilon }}^2 - V_i}y_i = q_{\bm{\varepsilon }}  V_i y_i\qty(1+ q_{\bm{\varepsilon }} \frac{V_i}{1-q_{\bm{\varepsilon }}V_i})= \frac{q_{\bm{\varepsilon }}  V_i y_i}{1-q_{\bm{\varepsilon }}V_i}
\end{aligneq}
Using \eqref{start1} and \eqref{start2} in \eqref{start} we obtain
\begin{align}\label{leave out mean}
    m_{y_i|\bm{y}_{-i},\bm{\theta }} & = \frac{\eta_i }{1-q_{\bm{\varepsilon }}V_i} - \frac{q_{\bm{\varepsilon }}  V_i y_i}{1-q_{\bm{\varepsilon }}V_i} = \frac{\eta_i - q_{\bm{\varepsilon }}  V_i y_i}{1-q_{\bm{\varepsilon }}V_i}= y_i + \frac{\eta_i-y_i}{1- q_{\bm{\varepsilon }}V_i}.
\end{align}

The term $V_i$ can be calculated efficiently using a Takahashi recursion on the Cholesky factor of the posterior precision $\bm{Q}_{\bm{u}|\bm{y,\theta }}$ without the need to calculate a dense matrix inverse, as implemented by $\texttt{inla.qinv}$. To calculate $\eta_i$, we use a matrix-vector solve.

Next, since the expression of $\pi (\bm{\theta }| \bm{y}_{i})$ is known up to a normalizing constant, we can use importance sampling and Riemann integration together with \eqref{leave out variance}, \eqref{leave out mean} to calculate \eqref{posterior mean LOO prediction}.
To estimate the CRPS, we use importance sampling and the expression for the posterior predictive distribution \eqref{posterior mean LOO prediction}. We have
\begin{align*}
    F_i(t) := \int_{\R^5} \qty(\int_{-\infty}^t \pi(y_i | \bm{y}_{-i}, \bm{\theta })\d y_i)p(\bm{\theta }| \bm{y})\d \bm{\theta } \approx  \sum_{j=1}^J   \Phi \qty(\frac{t-m_{y_i|\bm{y}_{-i},\bm{\theta }_j}}{\sigma_{y_i|\bm{y}_{-i},\bm{\theta_j }}}) w_j,
\end{align*}
where $\Phi$ is the cumulative distribution function of a standard Gaussian variable, and $\bm{\theta }_j,w_j$ are the importance samples and smoothed self-normalized weights.
Using this expression and the fact that there is an exact expression for the CRPS of a Gaussian mixture \citep{grimit2006continuous}, we obtain the CRPS.
\subsection{Spatial analysis of the scores}\label{spatial analysis precipitation appendix}
We hypothesize that some models may perform better than others in different spatial areas. To confirm this hypothesis, we begin by plotting the difference in the scores between the isotropic and anisotropic PC models. We plot the difference of the scores for the RMSE, CRPS, and DSS in \Cref{fig:diff_scores}.

\begin{figure}[H]
    \centering
    \begin{subfigure}[b]{0.32\linewidth}
        \centering
        \includegraphics[width=\linewidth]{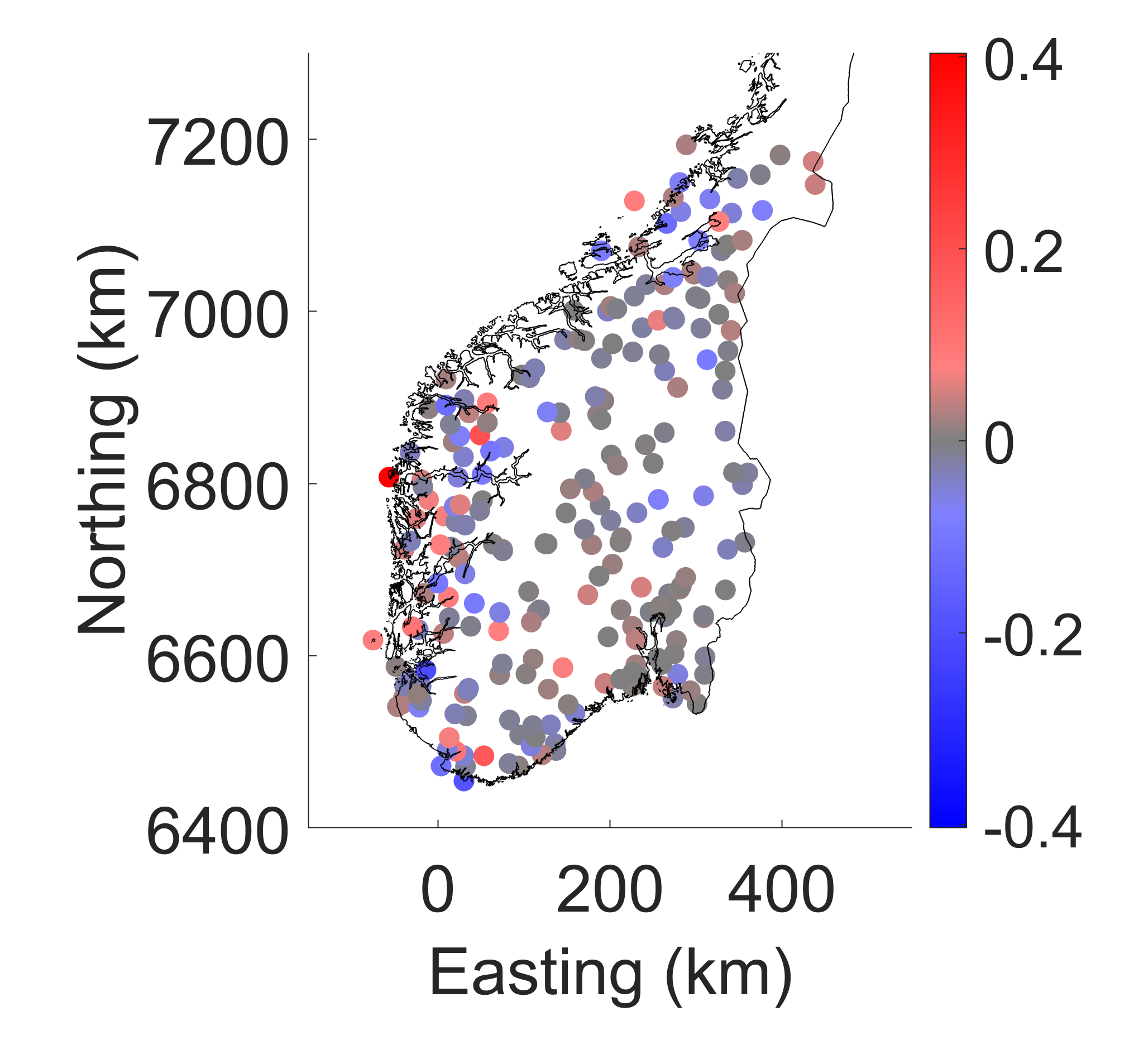}
        \caption{RMSE}
        \label{fig:rmse}
    \end{subfigure}
    \hfill
    \begin{subfigure}[b]{0.32\linewidth}
        \centering
        \includegraphics[width=\linewidth]{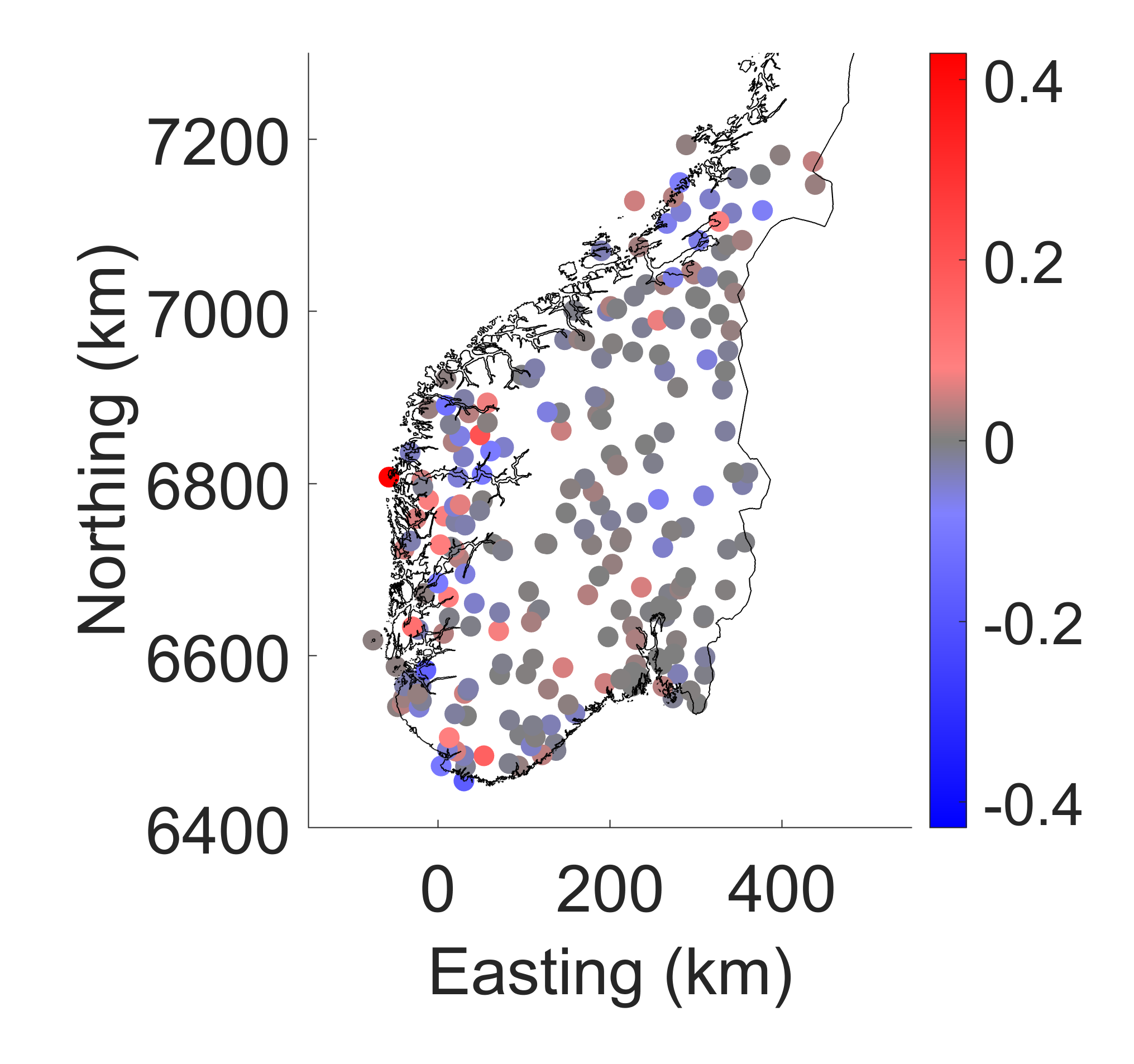}
        \caption{CRPS}
        \label{fig:crps}
    \end{subfigure}
    \hfill
    \begin{subfigure}[b]{0.32\linewidth}
        \centering
        \includegraphics[width=\linewidth]{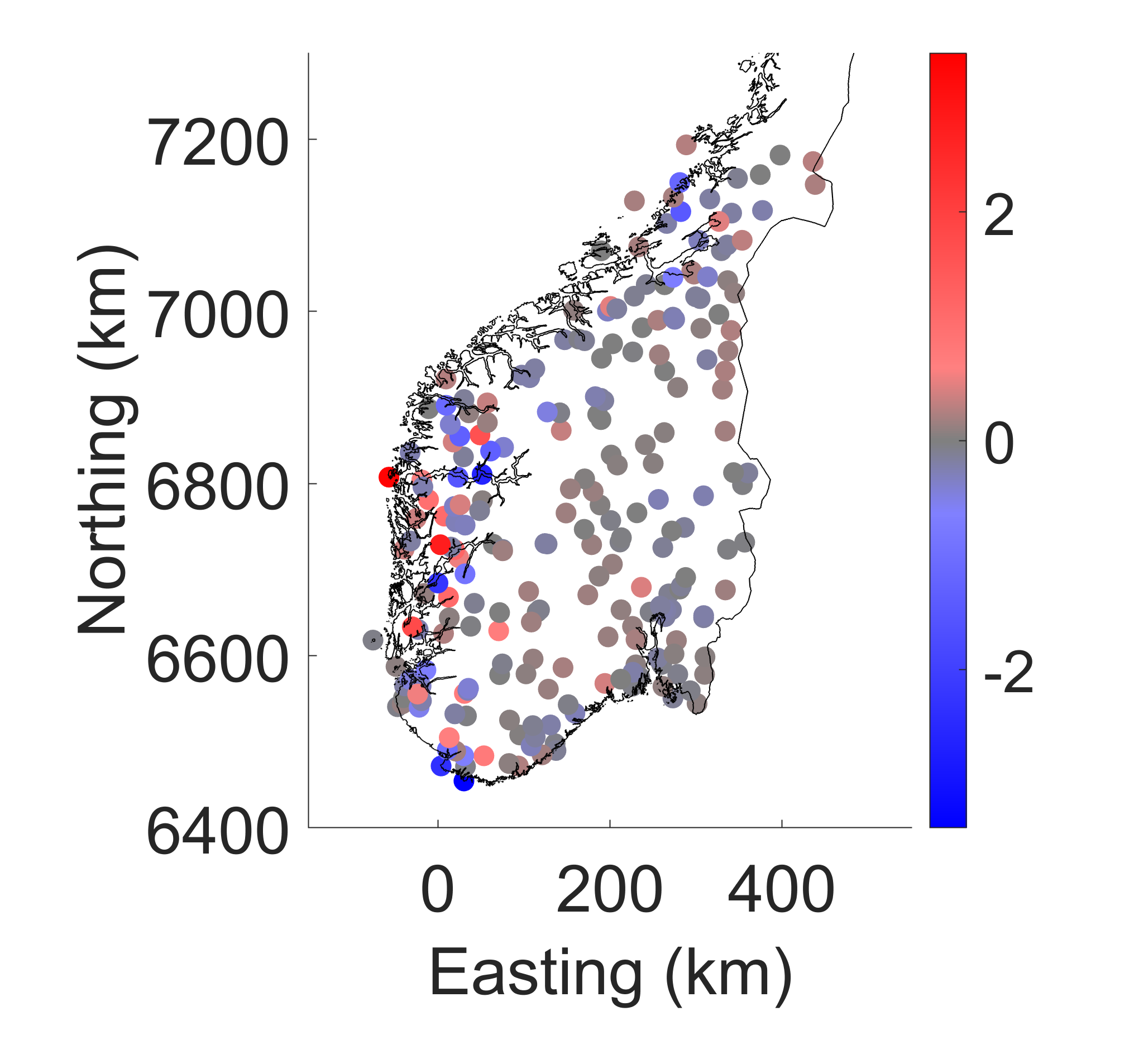}
        \caption{DSS}
        \label{fig:dss}
    \end{subfigure}
    \caption{Difference of scores between the isotropic and anisotropic PC models}
    \label{fig:diff_scores}
\end{figure}

We also do the same for the difference in the scores between the isotropic and anisotropic EG models in \Cref{fig:diff_scores2}. The results are very similar. We observe that the anisotropic models perform better towards the western coast and worse in the remaining region. This is to be expected, as the correlation structure differs between regions with high elevation and those with lower elevation, indicating that a non-stationary model may be more appropriate, as the anisotropy is spatially varying
\begin{figure}[H]
    \centering
    \begin{subfigure}[b]{0.32\linewidth}
        \centering
        \includegraphics[width=\linewidth]{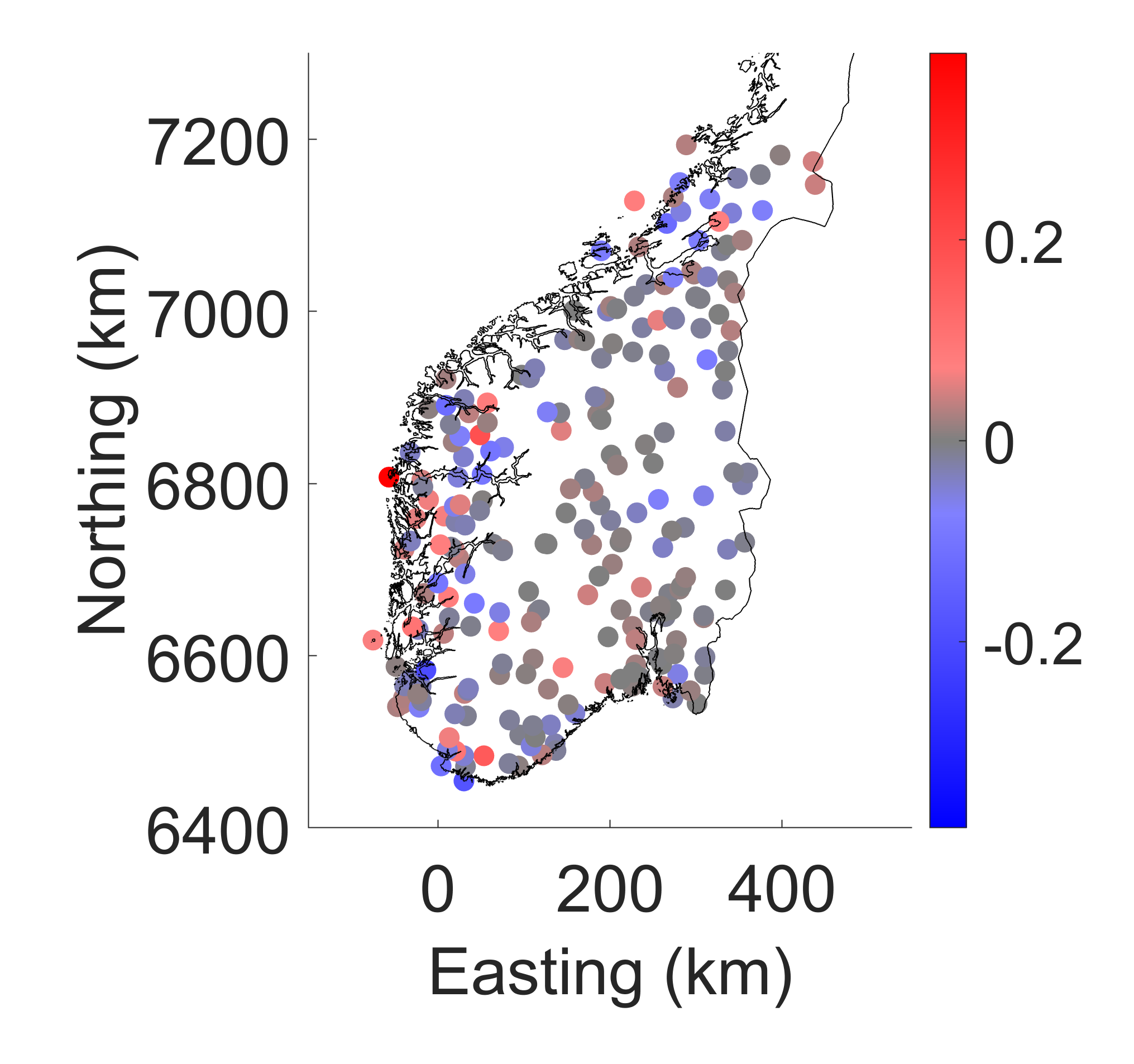}
        \caption{RMSE}
        \label{fig:rmse2}
    \end{subfigure}
    \begin{subfigure}[b]{0.32\linewidth}
        \centering
        \includegraphics[width=\linewidth]{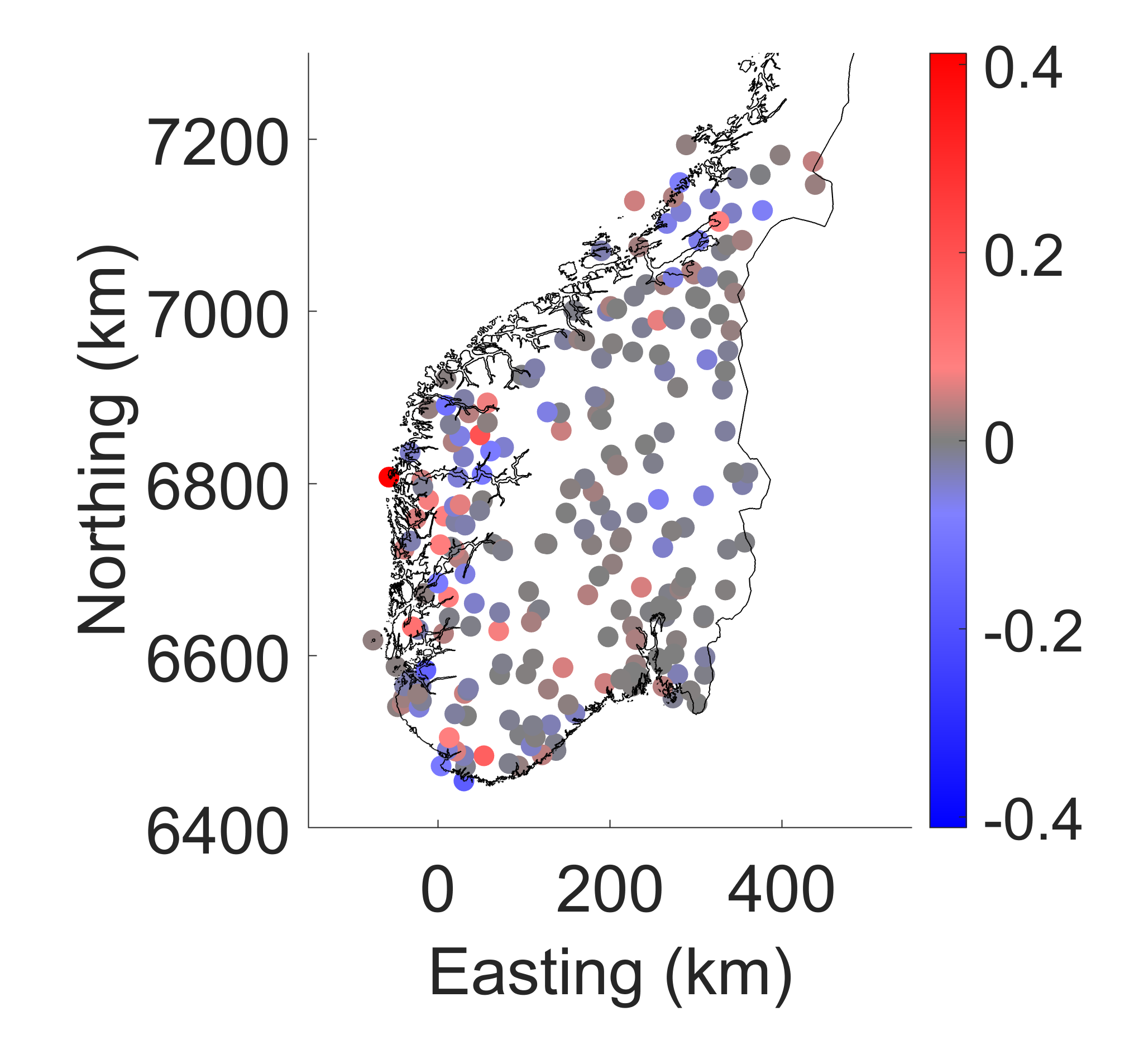}
        \caption{CRPS}
        \label{fig:crps2}
    \end{subfigure}
    \begin{subfigure}[b]{0.32\linewidth}
        \centering
        \includegraphics[width=\linewidth]{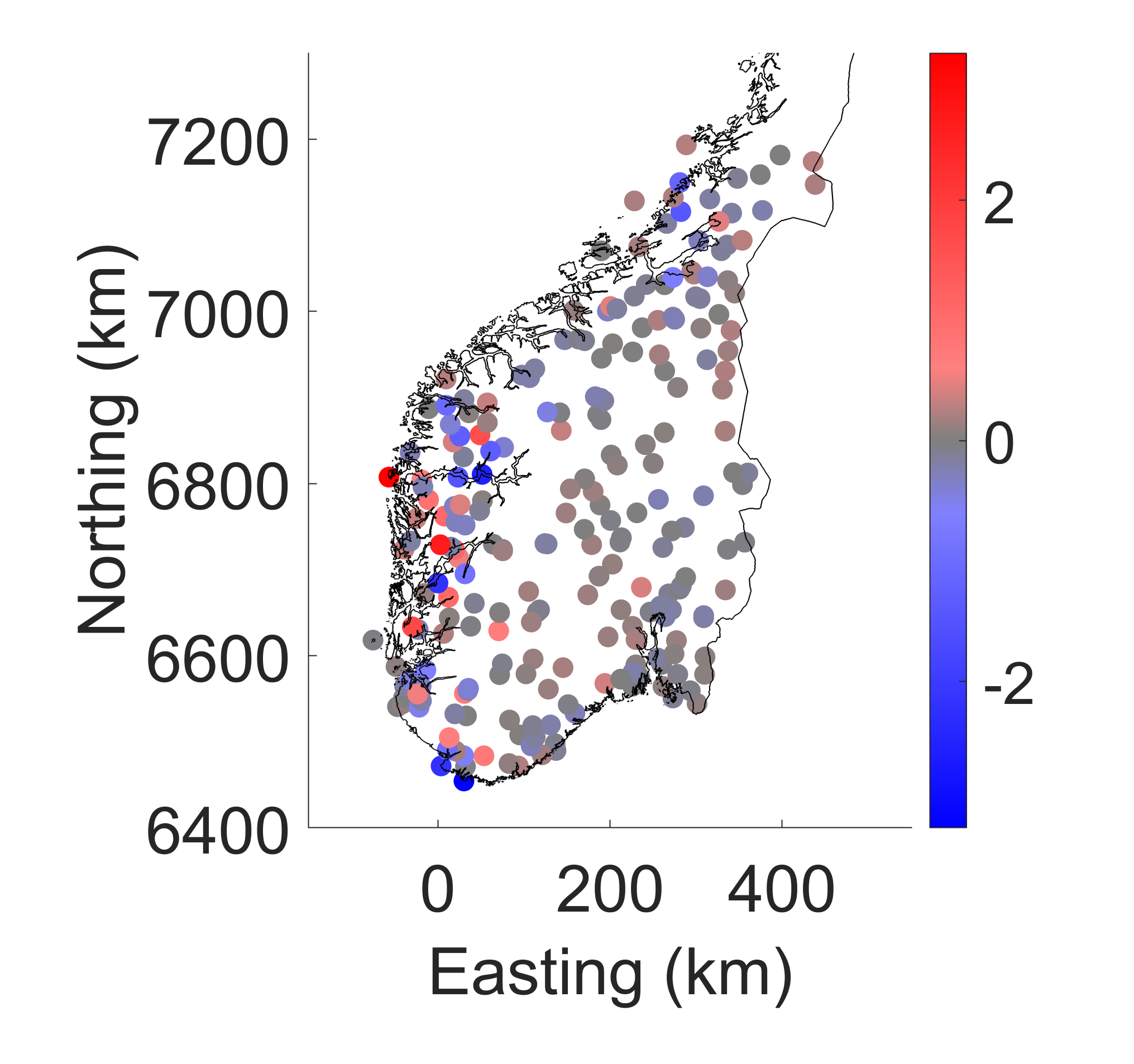}
        \caption{DSS}
        \label{fig:dss2}
    \end{subfigure}
    \caption{Difference of scores between the isotropic PC and anisotropic EG models}
    \label{fig:diff_scores2}
\end{figure}

\subsection{Precipitation simulation study}\label{simulation study precipitation appendix}
In this section, we repeat the structure of the simulations in \Cref{precipitation section}, but now, to investigate how the quality of the prediction changes with increased data, we use synthetic data. To do so, we simulate precipitation data on 4000 uniformly distributed locations in Norway from the model in \eqref{height model} with parameters
\begin{align*}
    \wh{\rho} & = 201, & \wh{\bm{v}} & = (-0.45,0.03), & \wh{\sigma}_u & =0.63, & \wh{\sigma}_{\bm{\varepsilon}} & = 0.14, & \beta_0 & = 0.96, & \beta_1 & =0.67  \\
    \wh{\rho} & = 132, & \wh{\bm{v}} & = (0.00,0.00),  & \wh{\sigma}_u & =0.65, & \wh{\sigma}_{\bm{\varepsilon}} & = 0.13, & \beta_0 & = 0.93, & \beta_1 & =0.66.
\end{align*}
These correspond to data generated from anisotropic and isotropic models, respectively, using the MAP obtained for the precipitation data in \Cref{results precip} with anisotropic and isotropic PC priors respectively. The data is shown in \Cref{fig: simulated data}.
\begin{figure}[H]
    \centering
    \begin{subfigure}[b]{0.45\linewidth}
        \centering
        \includegraphics[width=\linewidth]{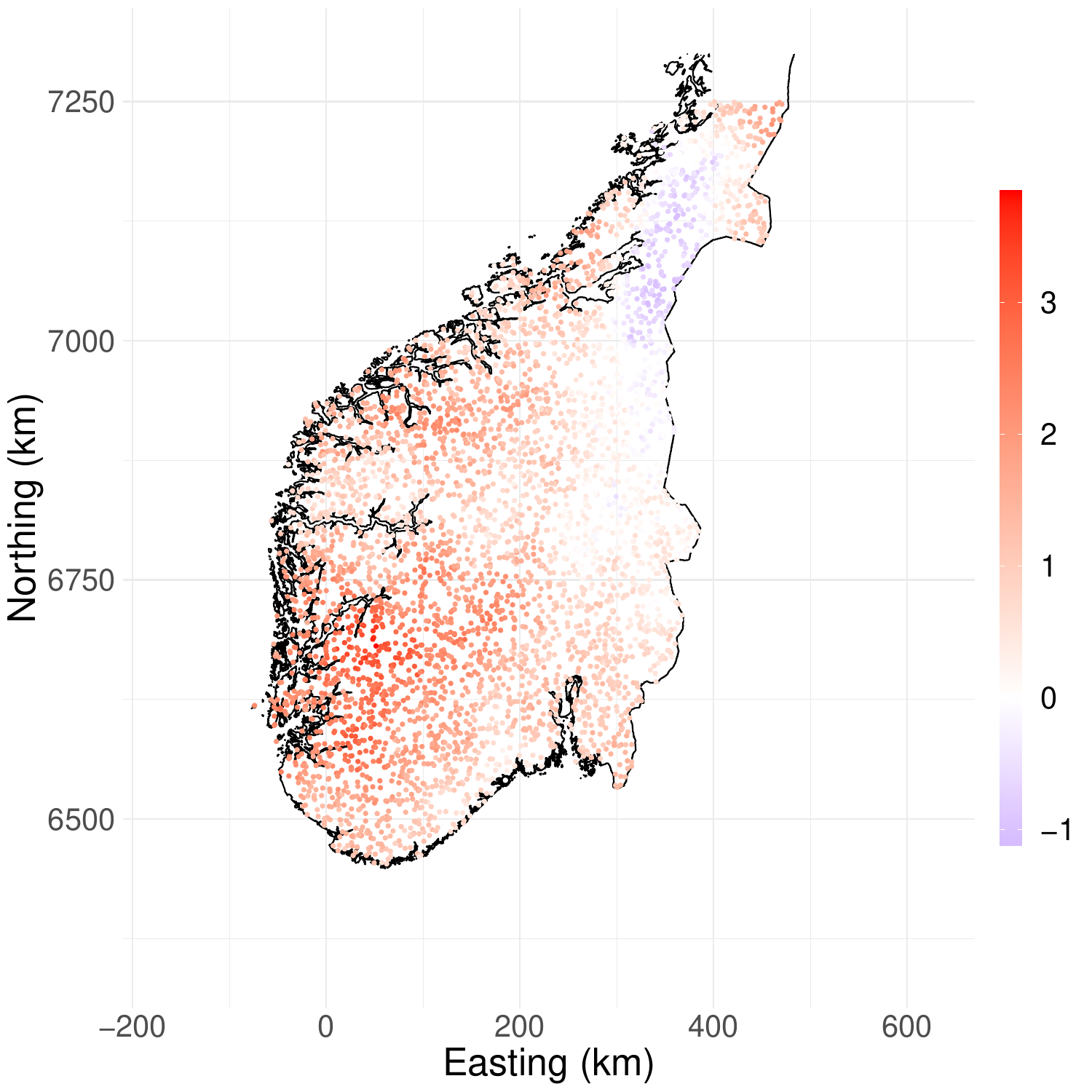}
        \caption{Anisotropic data}
    \end{subfigure}
    \begin{subfigure}[b]{0.45\linewidth}
        \centering
        \includegraphics[width=\linewidth]{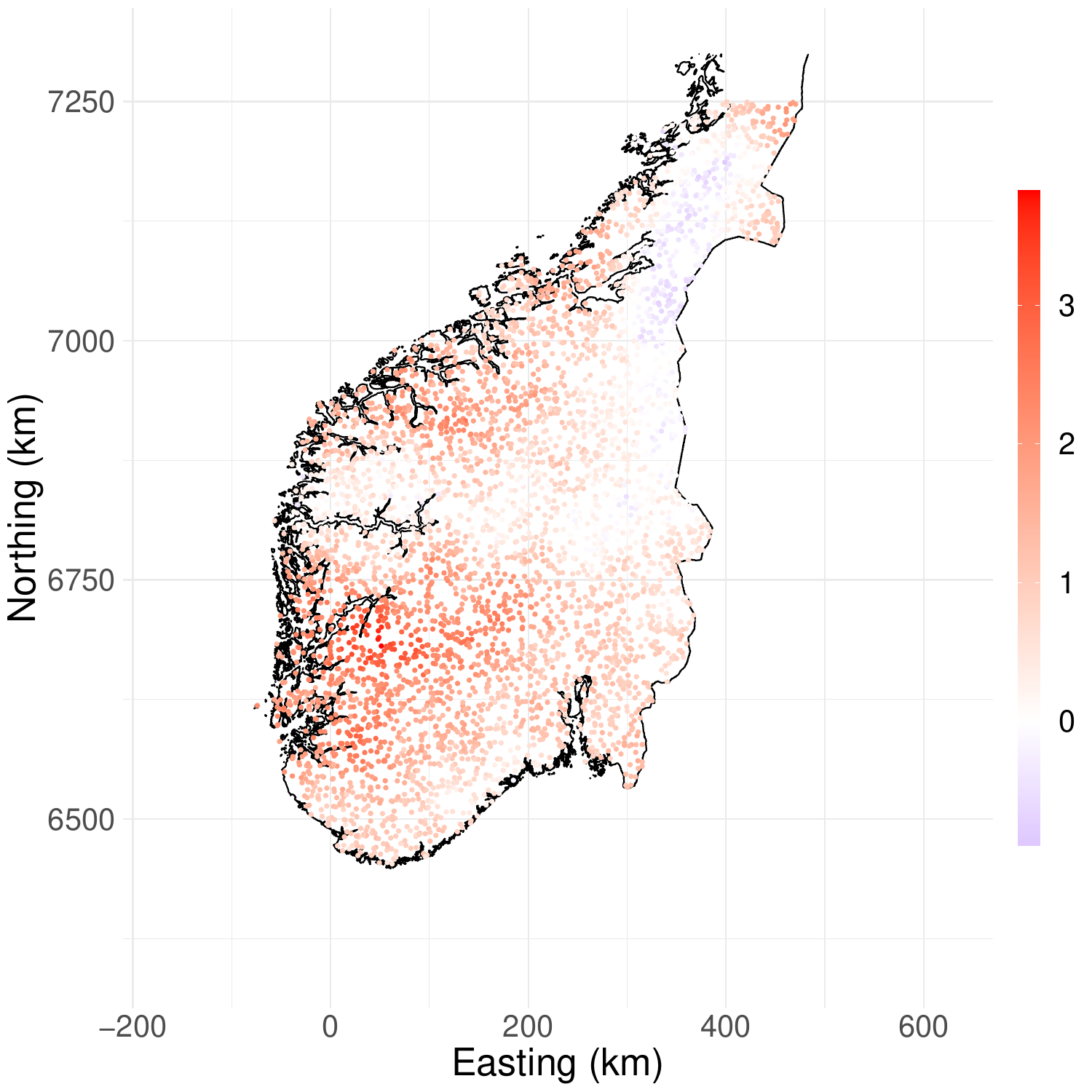}
        \caption{Isotropic data}
    \end{subfigure}
    \caption{Precipitation data simulated using model \eqref{height model} using an anisotropic and isotropic field $u$ respectively}
    \label{fig: simulated data}
\end{figure}

We then fit the model to the data and calculate the RMSE, CRPS, and DSS scores. We repeat this process for $n_y=25,50,100,150,200,400, 600,800,1000$ observations sampled uniformly from the simulated data. We repeat this process $100$ times for each $n_y$ and each prior. The results are shown in \Cref{fig: Simulation_images-LO_scores-pdf}. As can be seen from the plots, the anisotropic model outperforms the isotropic model for smaller datasets. However, as the number of observations increases, the scores become almost equal. The difference between the anisotropic PC and EG models is less pronounced than in the previous section.
\begin{figure}[H]
    \centering
    \begin{subfigure}[b]{0.48\linewidth}
        \centering
        \includegraphics[width=\linewidth]{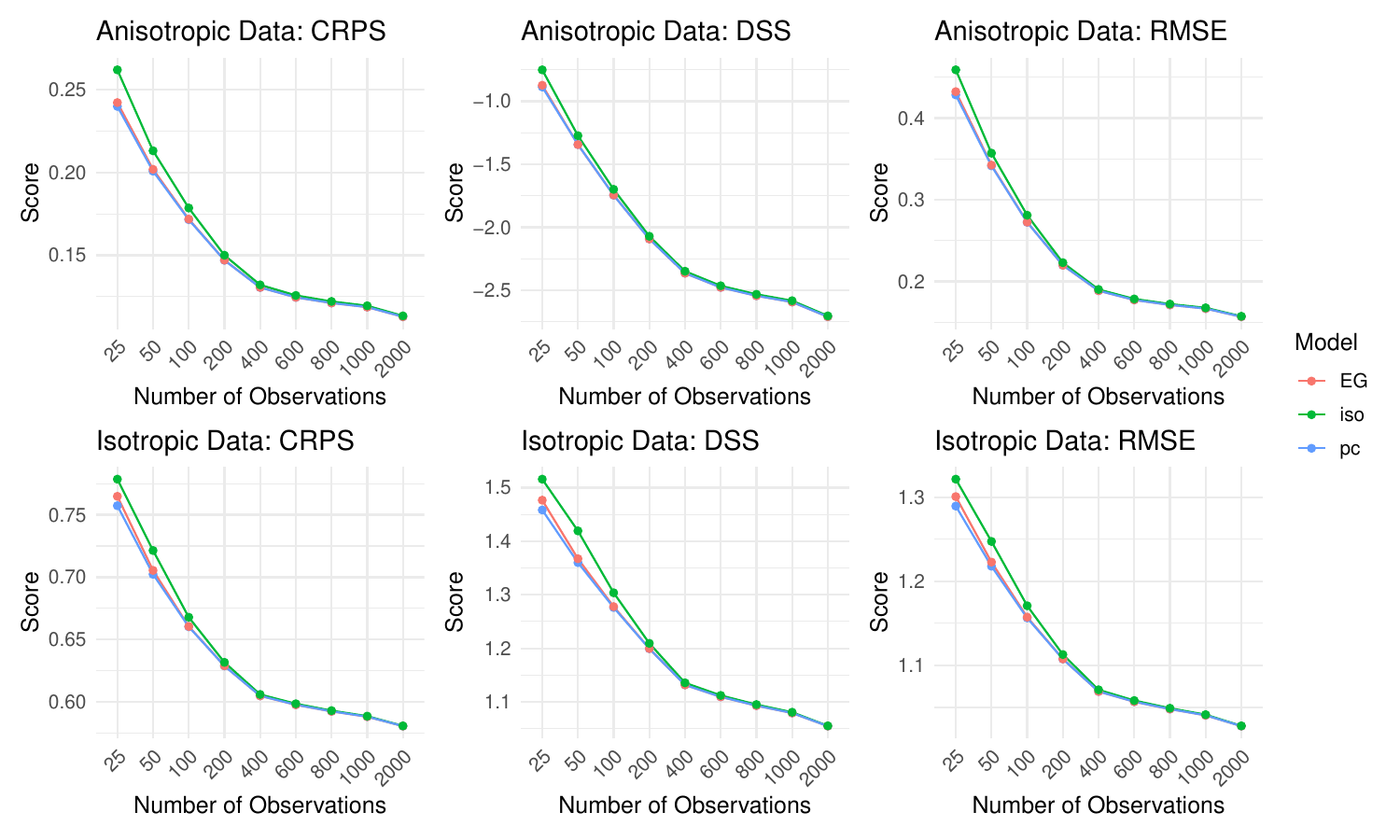}
        \caption{Scores}
        \label{fig:scores_combined}
    \end{subfigure}
    \hfill
    \begin{subfigure}[b]{0.48\linewidth}
        \centering
        \includegraphics[width=\linewidth]{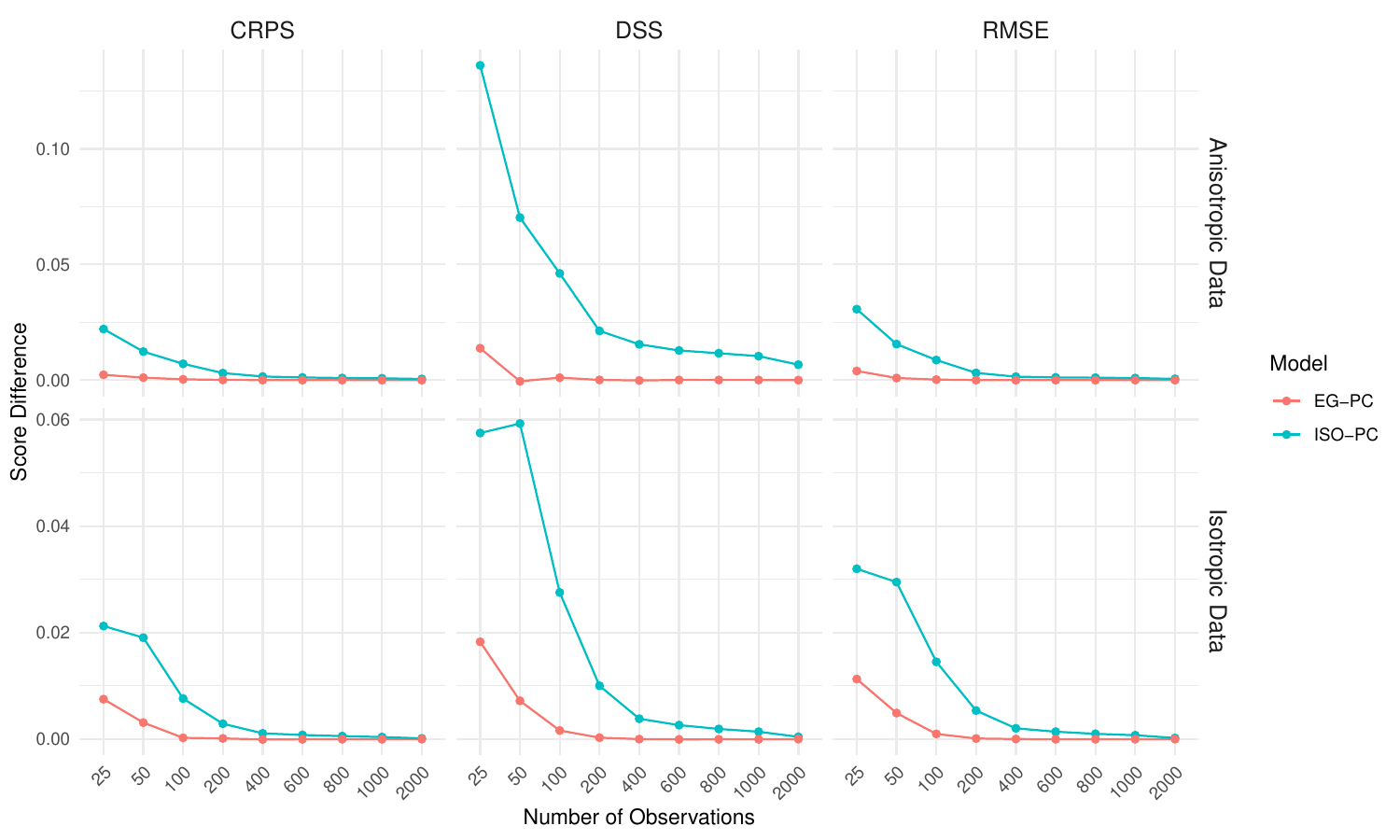}
        \caption{Difference of scores}
        \label{fig:scores_diff_combined}
    \end{subfigure}
    \caption{Comparison of the scores of the model for each prior for anisotropic and isotropic data, as the number of observations increases}
    \label{fig:diff_scores_sim}
\end{figure}

Finally, we also plot the interval score for the parameters $(\log(\kappa ), v_1, v_2, \log(\sigma _u))$. Given a credible interval $(L_F,U_F)$ with confidence level $\alpha$ , the interval score is defined by
\begin{align*}
    S_{\text {INT }}(F, y)=U_F-L_F+\frac{2}{\alpha}\left(L_F-y\right) \mathbb{I}\left(y<L_F\right)+\frac{2}{\alpha}\left(y-U_F\right) \mathbb{I}\left(y>U_F\right).
\end{align*}
The interval score is a proper scoring rule consistent for equal-tail error probability intervals: $S(F, G)$ is minimized for the narrowest $P I$ that has expected coverage $1-\alpha$. The results are shown in \Cref{fig: Simulation_images-CI_scores}. As can be seen from the plots, the interval score for the parameters $(\log(\kappa ),  \sigma _u)$ is generally better for the anisotropic model. This is especially pronounced for $\log(\kappa )$ when there is less data. The interval scores for $v_1$ and $v_2$ are not shown for the isotropic model, as these parameters are not estimated in the isotropic model.
\begin{figure}[H]
    \centering
    \includegraphics[width=\textwidth]{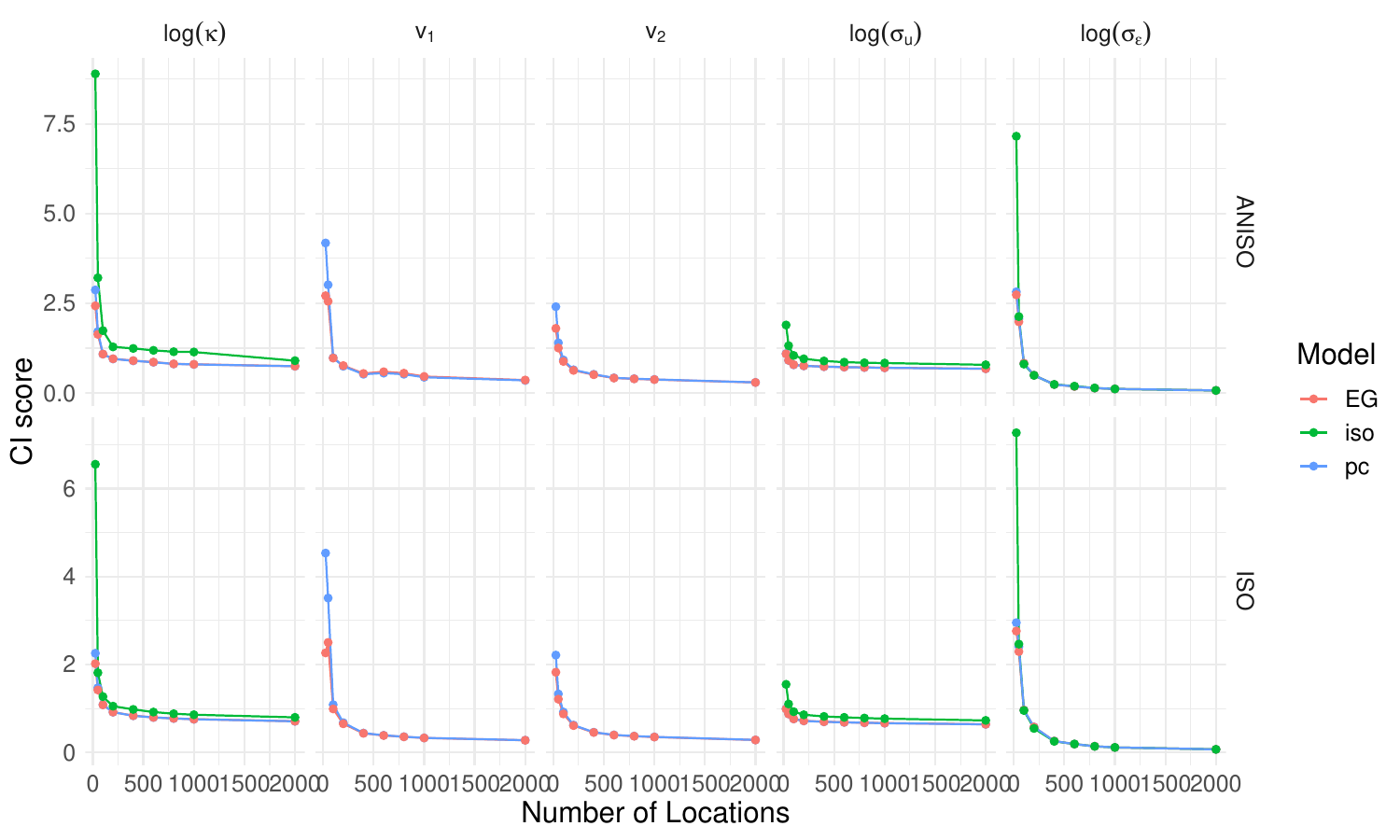}
    \caption{From top to bottom, anisotropic and isotropic data are observed. From left to right, the interval score for the parameters $(\log(\kappa ), v_1, v_2, \log(\sigma _u),\log(\sigma _\epsilon ))$ with a varying number of observations is shown.}
    \label{fig: Simulation_images-CI_scores}
\end{figure}

\section*{Code and Data Availability}
The code and data needed to reproduce all simulations and figures in this paper are publicly available at \url{https://github.com/LiamLlamazares/anisotropic-pc-priors-jasa} (release \texttt{v1.0-jasa}) and archived with a permanent DOI at \url{https://doi.org/10.5281/zenodo.19947749}. The repository pins the version of the \texttt{fmesher} package used in the manuscript; see the \texttt{README.md} for installation and reproduction instructions. The Norwegian precipitation dataset analysed in \Cref{precipitation section} is included in the repository.

\section*{Acknowledgments}
The authors thank the reviewers and the editor for their constructive feedback, which substantially improved the manuscript.

\section*{Funding}
Liam Llamazares-Elias is supported by The Maxwell Institute Graduate School in Modelling, Analysis, and Computation, a Centre for Doctoral Training funded by the EPSRC (grant EP/S023291/1), the Scottish Funding Council, Heriot-Watt University and the University of Edinburgh.

\section*{Disclosure Statement}
The authors report there are no competing interests to declare.

\bibliography{biblio.bib}

\end{document}